\documentclass[11pt]{article}
\usepackage{smile}

\usepackage{fullpage}
\usepackage{lscape}
\usepackage{bigints}
\usepackage{framed}
\usepackage{mdframed}
\usepackage{enumerate}
\usepackage[inline]{enumitem}
\usepackage[T1]{fontenc}
\usepackage{moresize}
\usepackage{bm}
\usepackage{bbm}
\usepackage{dsfont}
\usepackage{amsmath}
\usepackage{amssymb}
\usepackage{amsthm}
\usepackage{amsfonts}
\usepackage{stmaryrd}
\usepackage{array}
\usepackage{mathrsfs}
\usepackage{mathtools} 
\usepackage{extarrows}
\usepackage{stackrel}
\usepackage{relsize,exscale}
\usepackage{scalerel}
\usepackage[nodisplayskipstretch]{setspace}
\usepackage{color}
\usepackage[usenames,dvipsnames]{xcolor}
\usepackage{cancel}
\usepackage{soul}
\usepackage{undertilde}
\usepackage{xfrac}
\usepackage{siunitx}
\usepackage{graphicx}
\usepackage{float}
\usepackage{rotating}
\usepackage{subcaption}
\usepackage{overpic}
\usepackage[all]{xy}
\usepackage{tikz}
\usetikzlibrary{arrows,matrix,positioning,calc,automata,patterns}
\usepackage{booktabs}
\usepackage{dcolumn}
\usepackage{multirow}
\usepackage{diagbox}
\usepackage{tabularx}
\usepackage{verbatim}
\usepackage{listings}
\usepackage[ruled,vlined]{algorithm2e}
\usepackage{fancyvrb}
\usepackage{hyperref}
\usepackage[round]{natbib}
\usepackage{sectsty}

\hypersetup{
    bookmarks=true,         
    unicode=false,          
    pdftoolbar=true,        
    pdfmenubar=true,        
    pdffitwindow=false,     
    pdfstartview={FitH},    
    pdftitle={My title},    
    pdfauthor={Author},     
    pdfsubject={Subject},   
    pdfcreator={Creator},   
    pdfproducer={Producer}, 
    pdfkeywords={key1, key2}, 
    pdfnewwindow=true,      
    colorlinks=true,        
    linkcolor=blue,         
    citecolor=blue,         
    filecolor=blue,         
    urlcolor=cyan           
}

\usepackage{stackengine}
\stackMath
\newcommand\tenq[2][1]{%
\def\useanchorwidth{T}%
\ifnum#1>1%
\stackunder[0pt]{\tenq[\numexpr#1-1\relax]{#2}}{\!\scriptscriptstyle\thicksim}%
\else%
\stackunder[1pt]{#2}{\!\scriptstyle\thicksim}%
\fi%
}

\makeatletter
\DeclareRobustCommand\widecheck[1]{{\mathpalette\@widecheck{#1}}}
\def\@widecheck#1#2{%
    \setbox\z@\hbox{\m@th$#1#2$}%
    \setbox\tw@\hbox{\m@th$#1%
       \widehat{%
          \vrule\@width\z@\@height\ht\z@
          \vrule\@height\z@\@width\wd\z@}$}%
    \dp\tw@-\ht\z@
    \@tempdima\ht\z@ \advance\@tempdima2\ht\tw@ \divide\@tempdima\thr@@
    \setbox\tw@\hbox{%
       \raise\@tempdima\hbox{\scalebox{1}[-1]{\lower\@tempdima\box
\tw@}}}%
    {\ooalign{\box\tw@ \cr \box\z@}}}
\makeatother

\def\given{\,|\,}

\def\tr{\mathop{\text{tr}}\kern.2ex}

\def\tZ{{\tilde Z}}

\def\P{{\mathrm P}}

\def\E{{\mathrm E}}

\def\d{{\mathrm d}}

\newcommand{\F}{\mathrm{F}}

\newcolumntype{L}[1]{>{\raggedright\let\newline\\\arraybackslash\hspace{0pt}}m{#1}}
\newcolumntype{C}[1]{>{  \centering\let\newline\\\arraybackslash\hspace{0pt}}m{#1}}
\newcolumntype{R}[1]{>{ \raggedleft\let\newline\\\arraybackslash\hspace{0pt}}m{#1}}
\newcolumntype{d}[1]{D{.}{.}{#1}}
\newcolumntype{H}{>{\setbox0=\hbox\bgroup}c<{\egroup}@{}}
\newcolumntype{Z}{>{\setbox0=\hbox\bgroup}c<{\egroup}@{\hspace*{-\tabcolsep}}}
\newcolumntype{b}{X}
\newcolumntype{s}{>{\hsize=.5\hsize}X}

\numberwithin{equation}{section}

\newtheorem{theorem}{Theorem}[section]
\newtheorem{lemma}{Lemma}[section]
\newtheorem{proposition}{Proposition}[section]
\newtheorem{assumption}{Assumption}[section]
\newtheorem{corollary}{Corollary}[section]
\newtheorem{innercustomgeneric}{\customgenericname}
\providecommand{\customgenericname}{}
\newcommand{\newcustomtheorem}[2]{%
  \newenvironment{#1}[1]
  {%
   \renewcommand\customgenericname{#2}%
   \renewcommand\theinnercustomgeneric{##1}%
   \innercustomgeneric
  }
  {\endinnercustomgeneric}
}
\newcustomtheorem{customdefinition}{Definition}
\newcustomtheorem{customdefinitions}{Definitions}
\newcustomtheorem{customtheorem}{Theorem}
\newcustomtheorem{customassumption}{Assumption}
\newcustomtheorem{customlemma}{Lemma}
\newcustomtheorem{customexample}{Example}
\theoremstyle{definition}

\newtheorem{example}{Example}[section]
\newtheorem{remark}{Remark}[section]

\usepackage{enumitem}
\makeatletter
\newcommand{\mylabel}[2]{#2\def\@currentlabel{#2}\label{#1}}
\makeatother

\graphicspath{{./fig3/}}

\allowdisplaybreaks

\begin{document}

\setlength{\abovedisplayskip}{5pt}
\setlength{\belowdisplayskip}{5pt}
\setlength{\abovedisplayshortskip}{5pt}
\setlength{\belowdisplayshortskip}{5pt}
\hypersetup{colorlinks,breaklinks,urlcolor=blue,linkcolor=blue}

\title{\LARGE 
Unifying regression-based and design-based causal inference in time-series experiments}

\author{Zhexiao Lin\thanks{Department of Statistics, University of California, Berkeley, CA 94720, USA; e-mail: {\tt zhexiaolin@berkeley.edu}}~~~and~
Peng Ding\thanks{Department of Statistics, University of California, Berkeley, CA 94720, USA; e-mail: {\tt pengdingpku@berkeley.edu}}
}

\maketitle

\vspace{-1em}

\begin{abstract}
  Time-series experiments, also called switchback experiments or N-of-1 trials, play increasingly important roles in modern applications in medical and industrial areas. Under the potential outcomes framework, recent research has studied time-series experiments from the design-based perspective, relying solely on the randomness in the design to drive the statistical inference. Focusing on simpler statistical methods, we examine the design-based properties of regression-based methods for estimating treatment effects in time-series experiments. We demonstrate that the treatment effects of interest can be consistently estimated using ordinary least squares with an appropriately specified working model and transformed regressors. Our analysis allows for estimating a diverging number of treatment effects simultaneously, and establishes the consistency and asymptotic normality of the regression-based estimators. Additionally, we show that asymptotically, the heteroskedasticity and autocorrelation consistent variance estimators provide conservative estimates of the true, design-based variances. Importantly, although our approach relies on regression, our design-based framework allows for misspecification of the regression model.
\end{abstract}
    

\section{Introduction to time-series experiments}

Time-series experiments, also referred to as switchback experiments or N-of-1 trials, have attracted increasing attention across a variety of domains, including clinical trials \citep{lillie2011n}, online platforms \citep{bojinov2023design}, and other applied settings \citep{gabler2011n,mirza2017history,hawksworth2024methodological}. Unlike traditional randomized controlled trials, which assign treatments at the unit level and maintain a fixed assignment throughout the study, time-series experiments repeatedly randomize treatment assignments over time for the same, single experimental unit. Time-series experiments are closely related to longitudinal experiments \citep{robins1986new}, but focus on dynamic treatment effects on the same unit rather than average effects over a population. This experiment design allows for capturing dynamic causal effects in environments where both treatments and outcomes evolve over time and where treatment effects vary across time periods. Such designs are particularly attractive in modern applications where repeated interventions are feasible, the number of experimental units is limited, and the research interest is to estimate time-varying or lagged treatment effects. \cite{bojinov2019time} provided a concrete example of a time-series experiment conducted at AHL Partners LLP, a quantitative hedge fund that randomly assigned large trading orders to either human traders or computer algorithms to compare execution performance. Their example is a time-series experiment with sequential treatment assignments, where the outcome of interest is the relative performance across 10 equity-index futures markets over the course of a year.

Despite the wide applications of time-series experiments, the statistical analysis presents challenges, mainly due to the temporal dependence of outcomes and the risk of model misspecification. In time-series experiments, the treatment effects may persist beyond the period of assignment, inducing the dependence of the observed outcomes on past treatments. Consequently, the methods for analyzing traditional randomized experiments is not applicable. The design-based perspective offers a principled way to avoid strong assumptions on the outcome model. By viewing the treatment assignment as the sole source of randomness, we propose treatment effect estimates with desirable properties without strong model assumptions. This perspective builds on a rich literature on design-based analysis for various experimental settings, including regression analysis for completely randomized experiments \citep{lin2013agnostic}, matched-pairs experiments \citep{fogarty2018regression}, split-plot designs \citep{zhao2022reconciling,mukerjee2022causal}, cluster-randomized experiments \citep{middleton2015unbiased, su2021model}, and network experiments \citep{aronow2017estimating,leung2022causal, gao2023causal}. Recent comprehensive reviews are available in \citet{shi2024some} and \citet{ding2025randomization}. However, the design-based analysis of regression-based estimators in time-series experiments remains under explored.

In this paper, we develop a unified design-based framework for analyzing regression-based estimators in time-series experiments. We show that regressing the observed outcomes on current and past treatments using an appropriately specified working model and transformed regressors yields estimates of the lagged treatment effects of interest, even when the number of regressors diverges with the length of the time series. A key feature of our working model is the centering and scaling of treatment indicators, which ensures the regression coefficients to recover the causal estimands of interest. We establish the consistency and asymptotic normality of the ordinary least squares (OLS) and show that heteroskedasticity and autocorrelation consistent (HAC) variance estimators provide conservative variance estimates. Our results do not require the regression model to be correctly specified, reflecting the robustness inherent to the design-based analysis. Our work is related to \citet{bojinov2019time} and \citet{liang2025randomization}, who also develop design-based analyses for time-series experiments. However, our paper differs from theirs at two levels. First, at the level of estimands, we focus on treatment effects defined independently of the realized treatment path and outcome model assumptions. Second, at the level of estimators, we focus on simpler, regression-based point and variance estimators. We will provide a detailed comparison in the next section.

The remainder of the paper is organized as follows. Section~\ref{sec:setup} introduces the setup of the time-series experiment, defines the causal estimand of interest, and presents our regression-based estimator. Section~\ref{sec:theory} develops the asymptotic theory, establishing consistency, asymptotic normality, and variance estimation. Section~\ref{sec:extensions} discusses the extension to the continuous treatment. Section~\ref{sec:empirical} presents empirical studies, including simulation evidence and an application to the trading experiment data previously analyzed by \citet{bojinov2019time}. Section~\ref{sec:discussion} concludes with a discussion of possible directions for future research. The online supplementary materials contain extensions to the regression with interaction terms and general exposure mapping, as well as the technical details.

\textbf{Notation.} For a vector $a = (a_1,\ldots,a_d)^\top \in \mathbb{R}^d$, let $\lVert a \rVert_1 = \sum_{i=1}^d \lvert a_i \rvert$, and $\lVert a \rVert_2 = ( \sum_{i=1}^d a_i^2 )^{1/2}$. For a matrix $A = (a_{ij}) \in \mathbb{R}^{m \times n}$, let $\|A\|_{\mathrm{F}} = ( \sum_{i=1}^m \sum_{j=1}^n a_{ij}^2 )^{1/2}$ denote the Frobenius norm, and $\lVert A \rVert_2$ denote the spectral norm. Following the R syntax, we use $\textbf{lm}(Y \sim X)$ to denote the linear regression model with outcome $Y$ and regressors $X$.

\section{Setup, estimand and estimator}\label{sec:setup}

\subsection{Setup}

Consider a time-series experiment with a binary treatment and a single unit repeatedly measured at $T$ time points. At each time point $t=1,\ldots,T$, we apply $Z_t=0$ for the control and $Z_t=1$ for the treatment. For each time point $t$, define the treatment history up to $t$ as
$z_{t:1} = (z_t, z_{t-1}, \ldots, z_1)$. Let $\{Y_t(z_T, \ldots, z_1):z_T, \ldots, z_1 = 0 \text{ or } 1\}_{t=1}^T$ denote the potential outcomes, which can depend on the whole treatment history. We first simplify the potential outcomes by imposing the no-anticipation assumption below.
\begin{assumption}[No-anticipation]\label{asp:no-anticipation}
  For every possible treatment path $(z_1,\ldots,z_T)$, the potential outcome at time $t$ depends only on the treatment history up to and including time $t$, but does not depend on any future treatments $z_{t+1}, \ldots, z_T$:
\[
  Y_t(z_T,\ldots,z_{t+1}, z_t,\ldots,z_1) = Y_t(z_t,\ldots,z_1),~\text{for all}~t=1,\ldots,T.
\]
\end{assumption}

Under the no-anticipation assumption (Assumption~\ref{asp:no-anticipation}), the observed outcome simplifies to $Y_t=Y_t(z_{t:1})$, where $z_{t:1} = (z_t,\ldots,z_1)$ represents the treatment history up to and including time $t$, following the notation in \cite{bojinov2019time}. The no-anticipation assumption holds when treatments are rapidly randomized (e.g., online experiments), or treatments are administered sequentially (e.g., clinical N-of-1 trials). However, when there exists forward-looking behavior, where the experimental unit can predict future treatments, the no-anticipation assumption would be violated. We do not consider violations of the no-anticipation assumption in this paper.

The following example illustrates the notation.

\begin{example}
  Consider a time-series experiment with $T=3$ time points. Assume no-anticipation (Assumption~\ref{asp:no-anticipation}). At time $t=1$, there are two potential outcomes $Y_1(0)$ and $Y_1(1)$. At time $t=2$, there are four potential outcomes $Y_2(0,0)$, $Y_2(0,1)$, $Y_2(1,0)$, and $Y_2(1,1)$. At time $t=3$, there are eight potential outcomes $Y_3(0,0,0)$, $Y_3(0,0,1)$, $Y_3(0,1,0)$, $Y_3(0,1,1)$, $Y_3(1,0,0)$, $Y_3(1,0,1)$, $Y_3(1,1,0)$, and $Y_3(1,1,1)$. Given any treatment history $(Z_3,Z_2,Z_1)$, the observed outcome path is $Y_1 = Y_1(Z_1)$, $Y_2 = Y_2(Z_2,Z_1)$, and $Y_3 = Y_3(Z_3,Z_2,Z_1)$, and all other potential outcomes are not observed.
\end{example}

Time-series experiments and network experiments are closely related in that both extend the classical potential outcomes framework to accommodate structured interference. In network experiments, a unit’s potential outcome may depend not only on its own treatment but also on the treatments of its neighbors in a network \citep{aronow2017estimating}.
Similarly, in time-series experiments with a single unit observed over multiple periods, outcomes may depend not only on the current treatment but also on past or future treatments, giving rise to temporal interference. The no-anticipation assumption specifies a form of temporal interference: outcomes at time $t$ may depend on the entire treatment history up to $t$, but not on future treatments \citep{bojinov2019time}. In this sense, no-anticipation can be seen as a generalization of the Stable Unit Treatment Value Assumption (SUTVA) \citep{rubin1980randomization}, which requires that there are no hidden versions of treatment and there is no interference between units.

We introduce the following assumption to describe the design of the time-series experiment considered in this paper.

\begin{assumption}[Time-series experiment]\label{asp:design}
  Assume the treatments $Z_1,\ldots,Z_T$ are independent. Let $p_t = \P(Z_t=1) \in (0,1)$ for each $t = 1,\ldots,T$. Assume the values of $(p_1,\ldots,p_T)$ are known by the design.
\end{assumption}

To estimate the lagged treatment effects, which capture how treatments administered in earlier periods continue to influence outcomes at later times, we consider regressing the observed outcomes on some past treatments. We motivate the working regression model by the standard causal inference setting with only one time point. Assume there is only one binary treatment $Z$ with $\P(Z) = p$ known, and one pair of potential outcomes $(Y(0),Y(1))$ and observe $Y = Y(Z)$. Then the standard inverse propensity score weighted estimator for the average treatment effect $\E[Y(1)-Y(0)]$ is the sample mean version of $\E[(ZY/p) - (1-Z)Y/(1-p)] = \E[((Z-p)/p(1-p))Y]$. Extending the idea to time-series experiments with multiple time points, we normalize the treatment indicator $Z_t$ as
\begin{align}\label{eq:z_t}
  \tZ_t = (Z_t - \E[Z_t])/\Var[Z_t] = (Z_t-p_t)/[p_t(1-p_t)].
\end{align}
Suppose we are interested in the lagged treatment effects up to the number of lags $K$. Now we consider the OLS regression without the intercept:
\begin{align}\label{eq:ols}
  \textbf{lm}(Y_t \sim \tZ_t + \tZ_{t-1} + \cdots + \tZ_{t-K})
\end{align}
using observations from $t = K+1,\ldots,T$. Let $\tilde\tau = (\tilde\tau_0,\ldots,\tilde\tau_K)^\top \in \bR^{K+1}$ denote the OLS coefficient vector from \eqref{eq:ols}.

Several key points should be noted. First, the centering of $Z_t$ in $\tZ_t$ by $\E[Z_t]$ ensures that the estimand is invariant to location shifts in potential outcomes. Specifically, if the potential outcome $(Y_1,\ldots,Y_T)$ is shifted to $(Y_1+c,\ldots,Y_T+c)$ for some constant $c$, the estimand remains unchanged. This transformation invariance property guarantees that the estimand captures causal effects rather than the absolute level of the outcomes. Second, the scaling of $Z_t$ in $\tZ_t$ by $\Var[Z_t]$ ensures equal weighting across all time points in the estimand. If differential weighting across lags is desired to reflect differences in treatment assignment probabilities, the scaling can be adjusted accordingly. See the discussion about general treatment effects and the connection with \cite{bojinov2019time} in Section~\ref{sec:comparisons}. 

To ensure that the OLS estimator matches the scale of the causal estimand of interest, where the difference in scale arises from dividing $Z_t$ by $\Var[Z_t]$, we apply a linear transformation to $\tilde\tau$. Intuitively, this step rescales the regression coefficients back into the natural scale of the treatment effect. Since $\E[\tZ_t^2] = 1/\Var[Z_t]$ for each $t = 1,\ldots,T$, we define 
\begin{align}\label{eq:hat_tau_k}
\hat\tau_k = w_k^{-1} \tilde\tau_k \quad \text{for } k = 0,\ldots,K,
\end{align} 
where $w_k$ is the harmonic mean of the variances of the lag-$k$ treatment indicators: 
\begin{align}\label{eq:w_k}
  w_k = \left[\frac{1}{T-K} \sum_{t=K+1}^T \Var[Z_{t-k}]^{-1}\right]^{-1} = \left[\frac{1}{T-K} \sum_{t=K+1}^T \left(p_{t-k}\left(1-p_{t-k}\right)\right)^{-1}\right]^{-1}.
\end{align}
By Assumption~\ref{asp:design}, the treatment probabilities are fixed by design, and thus $w_k$ is a fixed constant. Let $$\hat\tau = (\hat\tau_0,\ldots,\hat\tau_K)^\top = (w_0^{-1} \tilde\tau_0,\ldots,w_K^{-1} \tilde\tau_K)^\top \in \bR^{K+1}$$ denote the OLS estimate from \eqref{eq:ols} after the linear transformation \eqref{eq:hat_tau_k}.  In this paper, we focus on studying the design-based properties of $\hat\tau$. Under the design-based perspective, we treat the potential outcomes as fixed and the randomness solely comes from the treatment assignments.

To write $\hat\tau$ explicitly, define the outcome vector $\mY_K = (Y_{K+1},\ldots,Y_T)^\top \in \bR^{T-K}$ and the regressor matrix $\mZ_K = (\tZ_{K+1:1}^\top,\ldots,\tZ_{T:T-K}^\top)^\top \in \bR^{(T-K) \times (K+1)}$ where $\tZ_{t:t'} = (\tZ_t,\tZ_{t-1},\ldots,\tZ_{t'})^\top$ for $t>t'$.  Let $\mW_K \in \bR^{(K+1)\times(K+1)}$ be the diagonal weight matrix $\mW_K = {\rm diag}(w_k:k=0,\ldots,K)$. Then, the OLS estimator in \eqref{eq:ols} after the linear transformation \eqref{eq:hat_tau_k} can be expressed as
\begin{align}\label{eq:hat_tau_ols}
  \hat{\tau} = \mW_K^{-1} \tilde\tau = \mW_K^{-1} (\mZ_K^\top \mZ_K)^{-1} \mZ_K^\top \mY_K.
\end{align}

The OLS regression is intuitive for estimating lagged treatment effects because it directly links the outcome to both current and past treatments. Each coefficient $\hat\tau_k$ captures the average impact of a treatment administered $k$ periods earlier on the current outcome. In this way, the estimator naturally decomposes the outcome into lag-specific contributions, providing a clear measure of how treatments influence present and future outcomes over time. We will show that $\hat\tau$ can recover the causal estimand in the next subsection.


\subsection{Causal estimand}

We define the treatment effect of lag-$k$ treatment on the outcome at time $t$ as:
\begin{align}\label{eq:tau_t}
  \tau_{t,k} = \E[Y_t(Z_{t:t-k+1},1,Z_{t-k-1:1}) - Y_t(Z_{t:t-k+1},0,Z_{t-k-1:1})],
\end{align}
where the expectation is taken over the joint distribution of $(Z_{t:t-k+1}, Z_{t-k-1:1})$, holding $Z_{t-k}$ at 1 or 0. This estimand $\tau_{t,k}$ represents the average treatment effect on the outcome at time $t$ by setting the treatment at time $t-k$ to 1 rather than 0, averaging over all other treatment assignments. We then define the average treatment effect of lag-$k$ treatment on the outcome over time as:
\begin{align}\label{eq:tau}
  \tau_k = \frac{1}{T-K} \sum_{t=K+1}^T \tau_{t,k}.
\end{align}

Let $\tau = (\tau_0,\ldots,\tau_K)^\top \in \bR^{K+1}$ be the vector of lagged treatment effects. In this section, we focus on the interpretation of $\tau$, and we establish the asymptotic behavior of $\hat{\tau}$ as an estimator of $\tau$ in Section~\ref{sec:theory}.

The lag-$k$ treatment effect $\tau_k$ in \eqref{eq:tau} does not rely on any specific model assumptions about the potential outcomes. However, its interpretation becomes more explicit under additional modeling assumptions. We consider a linear model for the potential outcome to facilitate the interpretation of $\tau_k$.

\begin{example}[Linear model]\label{ex:linear}
  Consider the linear model for the potential outcome:
  \begin{align}\label{eq:linear}
    Y_t = Y_t(z_{t:1})=\sum_{k=0}^{t-1} \beta_{t,k} z_{t-k} + \epsilon_t(z_{t:1}),
  \end{align}
  for all $t = 1,\ldots,T$ and all treatment realizations $z_{t:1} \in \{0,1\}^t$ with non-random coefficients $\{\beta_{t,k}\}_{0 \le k < t}$. Under the model \eqref{eq:linear}, the lag-$k$ treatment effect can be expressed as:
  \[
    \tau_k = \frac{1}{T-K} \sum_{t = K+1}^T \beta_{t,k} + \frac{1}{T-K} \sum_{t = K+1}^T \E[\epsilon_t(Z_{t:t-k+1},1,Z_{t-k-1:1}) - \epsilon_t(Z_{t:t-k+1},0,Z_{t-k-1:1})].
  \]
  If the treatment switch at time $t-k$ does not affect the error term $\epsilon_t$, for example, if $\epsilon_t$ does not depend on $z_{t:1}$, then $\tau_k = (T-K)^{-1} \sum_{t = K+1}^T \beta_{t,k}$. Further assuming homogeneity of lag-$k$ treatment effects over time, i.e., $\beta_{t,k} = \beta_k$ for all $t > k$, we obtain $\tau_k = \beta_k$. On the other hand, under heterogeneous lag-$k$ treatment effects over time, the estimand $\tau_k$ captures the average treatment effect. This implies that the estimand $\tau_k$ remains robust to time-varying treatment effects. 
\end{example}

The linear model in Example~\ref{ex:linear} contains several important time series models, e.g., the autoregressive model and the moving-average model. In the following two examples, we illustrate the interpretation of $\tau_k$ under these two special models.

\begin{example}[Autoregressive model]\label{ex:autoregressive}
Consider the autoregressive model of order $p$ for the potential outcome:
\[
  Y_t(z_{t:1}) = \mu_t(z_{t:1}) + \sum_{k=1}^p \phi_k Y_{t-k}(z_{t-k:1}) + \epsilon_t(z_{t:1}),
\]
for all $t = 1,\ldots,T$ and all treatment realizations $z_{t:1} \in \{0,1\}^t$. Further assume $\mu_t(z_{t:1}) = \mu_t(z_t)$ and $\epsilon_t(z_{t:1}) = \epsilon_t(z_{t})$ for all $t = 1,\ldots,T$. Then we can write the potential outcome as 
\[
  Y_t(z_{t:1}) = \sum_{k=0}^{t-1} \Psi_k \mu_{t-k}(z_{t-k}) + \sum_{k=0}^{t-1} \Psi_k \epsilon_{t-k}(z_{t-k}),
\]
where the $\Psi_k$'s are defined recursively by $\Psi_0 = 1$ and $\Psi_j = \sum_{k=1}^{\min\{p,j\}} \phi_k \Psi_{j-k}$. The coefficient $\Psi_k$ is the $k$-step impulse-response coefficient from the autoregressive model of order $p$ \citep{hamilton2020time}, which measures the effect of a one-time shock today on future outcomes. We can further write the potential outcome as
\[
  Y_t(z_{t:1}) = \sum_{k=0}^{t-1} \Psi_k (\mu_{t-k}(1) - \mu_{t-k}(0))z_{t-k} + \sum_{k=0}^{t-1} \Psi_k \mu_{t-k}(0) + \sum_{k=0}^{t-1} \Psi_k \epsilon_{t-k}(z_{t-k}).
\]
Define $\beta_{t,k} = \Psi_k (\mu_{t-k}(1) - \mu_{t-k}(0))$ to see that it is a special case of Example~\ref{ex:linear}. If the error term satisfies $\epsilon_t(1) = \epsilon_t(0)$ for all $t$, then by Example~\ref{ex:linear}, we have
\[
  \tau_k = \Psi_k \Big[\frac{1}{T-K} \sum_{t = K+1}^T (\mu_{t-k}(1) - \mu_{t-k}(0))\Big].
\]
Therefore, $\tau_k$ reflects the average treatment effect on the mean function $\mu_t$ at lag $k$ over time, multiplied by the $k$-step impulse-response coefficient from the autoregressive model of order $p$.
\end{example}

\begin{example}[Moving-average model]\label{ex:moving-average}
  Consider the moving-average model of order $q$ for the potential outcome:
  \[
    Y_t(z_{t:1}) = \mu_t(z_{t:1}) + \sum_{k=1}^q \theta_k \epsilon_{t-k}(z_{t-k:1})+ \epsilon_t(z_{t:1}),
  \]
  for $t = 1,\ldots,T$ and all treatment realizations $z_{t:1} \in \{0,1\}^t$. Further assume $\mu_t(z_{t:1}) = \mu_t(z_t)$ and $\epsilon_t(z_{t:1}) = \epsilon_t(z_{t})$ for all $t = 1,\ldots,T$. Then we can write the potential outcome as
  \[
    Y_t(z_{t:1}) = \mu_t(z_{t}) + \sum_{k=1}^q \theta_k \epsilon_{t-k}(z_{t-k})+ \epsilon_t(z_{t}).
  \]
  We can further write the potential outcome as
  \[
    Y_t(z_{t:1}) = (\mu_t(1) - \mu_t(0) + \epsilon_t(1) - \epsilon_t(0)) z_t + \sum_{k=1}^q \theta_k (\epsilon_{t-k}(1) - \epsilon_{t-k}(0)) z_{t-k} + \mu_t(0) + \epsilon_t(0) + \sum_{k=1}^q \theta_k \epsilon_{t-k}(0).
  \]
  Define $\beta_{t,0} = \mu_{t}(1) - \mu_{t}(0) + \epsilon_t(1) - \epsilon_t(0)$, $\beta_{t,k} = \theta_k (\epsilon_{t-k}(1) - \epsilon_{t-k}(0))$ for $0<k\le q$, and $\beta_{t,k}=0$ for $k >q$ to see that it is a special case of Example~\ref{ex:linear}. Then by Example~\ref{ex:linear}, we have
  \begin{align*}
    \tau_0 &= \frac{1}{T-K} \sum_{t = K+1}^T (\mu_{t}(1) - \mu_{t}(0) + \epsilon_t(1) - \epsilon_t(0)), \\
    \tau_k &= \theta_k \Big[\frac{1}{T-K} \sum_{t = K+1}^T (\epsilon_{t-k}(1) - \epsilon_{t-k}(0))\Big],~~\text{for } 0<k\le q, \\
    \tau_k &= 0,~~\text{for } k > q.
  \end{align*}
  Therefore, $\tau_0$ is the average treatment effect on the mean function $\mu_t$ and the error term $\epsilon_t$ over time. For $0<k\le q$, $\tau_k$ is the average treatment effect on the error term $\epsilon_t$ at lag $k$ over time, multiplied by $\theta_k$, which is the $k$-step coefficient of the moving-average model. Moreover, $\tau_k = 0$ for $k$ larger than the moving average order $q$.
  
\end{example}

\begin{remark}
  In general, the estimand $\tau_k$ in \eqref{eq:tau} depends on the treatment assignment probabilities $(p_1,\ldots,p_T)$ because for each $\tau_{t,k}$, the expectation is taken with respect to the joint distribution of $(Z_{t:t-k+1}, Z_{t-k-1:1})$. Hence, changing the design probabilities may change the estimand. However, under the linear model considered in Examples~\ref{ex:linear}-\ref{ex:moving-average}, the estimand $\tau_k$ is independent of the treatment assignment probabilities because the treatment effects enter linearly and the potential outcome model is additive in the treatment indicators $Z_t$'s. In this case, $\tau_k$ reflects the structural effect of lag-$k$ treatment on the outcome, independent of the randomization probabilities.
\end{remark}

\section{Theory}\label{sec:theory}

\subsection{Consistency and asymptotic normality}\label{sec:asymptotic}

In this section, we analyze the asymptotic properties of $\hat\tau$ as $T \to \infty$ \citep{bojinov2019time,liang2025randomization}, analogous to the asymptotic analysis in traditional randomized experiments by letting the sample size grow to infinity \citep{freedman2008regression, lin2013agnostic,li2017general}. By the representation of $\hat\tau$ in \eqref{eq:hat_tau_ols}, we decompose $\hat\tau$ as:
\begin{align}\label{eq:hat_tau}
  \hat{\tau} = \tau + \mW_K^{-1} [(T-K)^{-1}\mZ_K^\top \mZ_K]^{-1} [(T-K)^{-1}\mZ_K^\top (\mY_K - \mZ_K \mW_K \tau)].
\end{align}

Therefore, the key is to analyze the asymptotic behaviors of $(T-K)^{-1}\mZ_K^\top \mZ_K$ and $(T-K)^{-1}\mZ_K^\top (\mY_K - \mZ_K \mW_K \tau)$, since $\mW_K$ is a fixed diagonal matrix. We impose the following assumption on the treatment probabilities of the time-series experiment.

\begin{assumption}[Treatment probabilities condition]\label{asp:ps}
  There exist some constant $\epsilon \in (0,1/2]$ such that $\epsilon \le \inf_{t \ge 1} p_t \le \sup_{t \ge 1} p_t \le 1-\epsilon$.
\end{assumption}

Assumption~\ref{asp:ps} assumes that the treatment probabilities are uniformly bounded away from 0 and 1 for all time points, which rules out extreme treatment probabilities. This corresponds to the classic overlap assumption for causal inference with observational studies \citep{d2021overlap}, which assume the propensity score are bounded away from 0 and 1 for almost all covariates.

Then we impose weak dependence conditions on potential outcomes sequence $[Y_t = Y_t(Z_{t:1})]_{t=1}^\infty$, which rule out strong dependence of the observed outcomes on the past treatments. We consider two types of weak dependence: decaying carryover effects and $m$-dependence.
Our results can be extended to more general weak dependence structures using the same proof techniques.

We first introduce the decaying carryover effects assumption, which assumes the carryover effects of the treatment decay over time.
\begin{assumption}[decaying carryover effects condition]\label{asp:ani}
Let
\[
\theta_{T,k} = \max_{k < t \le T} \E\left[\left\lvert Y_t(Z_{t:t-k+1},Z_{t-k:1}) - Y_t(Z_{t:t-k+1},Z'_{t-k:1})\right\rvert \right],
\]
where $Z'_{t-k:1}$ is an identically distributed copy of $Z_{t-k:1}$ and is independent of $Z_{t:1}$, and the expectation is taken over the joint distribution of $(Z_{t:t-k+1},Z_{t-k:1})$ and the distribution of $Z'_{t-k:1}$. Assume that
\begin{enumerate}[label=(\roman*)]
  \item\label{asp:ani:i} $\sum_{k=1}^T \theta_{T,k}/T \to 0$ as $T \to \infty$.
  \item\label{asp:ani:ii} $\theta_{T,k} = O(k^{-1-\delta})$ for some $\delta>0$.
\end{enumerate}
\end{assumption}





We then introduce the $m$-dependence assumption \citep{hoeffding1948central}, which restricts the influence of past treatments to some given number of time points.

\begin{assumption}[$m$-dependence condition]\label{asp:mdep}
  There exist some $m = m(T)$ such that $Y_t = Y_t(Z_{t:1}) = Y_t(Z_{t:t-m})$ for all $m < t \le T$ with 
  \begin{enumerate}[label=(\roman*)]
    \item\label{asp:mdep:i} $m/T \to 0$.
    \item\label{asp:mdep:ii} $m^4/T \to 0$.
  \end{enumerate}
\end{assumption}

Assumption~\ref{asp:ani} requires that the carryover effects of the treatment decay sufficiently fast. Similar conditions appear in the previous time-series experiments literature \citep{ni2023design, liang2025randomization}. The decaying carryover effects assumption is analogous to the approximate neighborhood interference assumption in network experiments \citep{leung2022causal}, and can be viewed as a causal analogue of mixing conditions in traditional time series analysis \citep{rosenblatt1956central, bradley2005basic}. Assumption~\ref{asp:mdep} requires that the observed outcome at time $t$ depends only on the most recent $m$ treatments. Assumptions~\ref{asp:ani} and~\ref{asp:mdep} are related but do not necessarily imply one another. Assumption~\ref{asp:ani} states that the carryover effects of the treatment decay over time, while still allowing dependence on all prior treatments. In contrast, Assumption~\ref{asp:mdep} permits arbitrary dependence on past treatments but restricts this dependence to a given number of previous time points. In both Assumptions~\ref{asp:ani} and~\ref{asp:mdep}, the first condition is weaker than the second condition, as first condition is required for consistency and the second condition is required for asymptotic normality.

\begin{remark}
  Recall the linear potential outcome model in Example~\ref{ex:linear}, and suppose the error term $\epsilon_t$ is independent of the treatment path $Z_{t:1}$, that is,
  \begin{align}\label{eq:linear_model}
    Y_t = Y_t(z_{t:1}) = \sum_{k=0}^{t-1} \beta_{t,k} z_{t-k} + \epsilon_t.
  \end{align}
  Under the model \eqref{eq:linear_model}, the quantity $\theta_{T,k}$ in the decaying carryover effects assumption (Assumption~\ref{asp:ani}) can be expressed as
  \[
    \theta_{T,k} = \max_{k < t \le T} \E\left[\left\lvert \sum_{j=k}^{t-1} \beta_{t,j} Z_{t-j} - \sum_{j=k}^{t-1} \beta_{t,j} Z'_{t-j} \right\rvert \right] \le 2 \max_{k < t \le T} \sum_{j=k}^{t-1} p_{t-j} (1-p_{t-j}) \lvert \beta_{t,j} \rvert,
  \]
  where $(Z'_1,\ldots,Z'_{t-k})$ is an independent and identically distributed copy of $(Z_1,\ldots,Z_{t-k})$ for all $k < t \le T$. This upper bound shows that the decay rate of $\theta_{T,k}$ is determined by how quickly the coefficients $\beta_{t,j}$ decrease as the lag $j$ increases. In contrast, the $m$-dependence assumption (Assumption~\ref{asp:mdep}) is equivalent to for all $k < t \le T$:
  \[
    \beta_{t,k} = 0 \quad \text{for all } m < k \le t-1.
  \]
  Thus, Assumption~\ref{asp:mdep} enforces a strict truncation of the carryover effects after $m$ lags.
\end{remark}

\begin{remark}
  The decaying carryover effects assumption (Assumption~\ref{asp:ani}) involves the quantity $\theta_{T,k}$, which depends on the distribution of treatment assignments. One might therefore be concerned with this dependence. However, $\theta_{T,k}$ is upper bounded by
  \[
    \theta_{T,k} \le \max_{k < t \le T} \max_{\substack{z_{t:t-k+1} \in \{0,1\}^{k} \\ z_{t-k:1},z'_{t-k:1} \in \{0,1\}^{t-k}}}\left\lvert Y_t(z_{t:t-k+1},z_{t-k:1}) - Y_t(z_{t:t-k+1},z'_{t-k:1})\right\rvert,
  \]
  where the maximum is taken over all possible realizations of treatment assignments $z_{t:t-k+1}, z_{t-k:1}, z'_{t-k:1}$. Consequently, we may impose a stronger assumption that depends only on the potential outcomes but not on the distribution of treatment assignments.
\end{remark}

\begin{remark}
  What distinguishes time-series experiments from network experiments is that the primary causal estimands of interest are lagged treatment effects, i.e., the impact of current and past treatments on the current outcome, and the associated structure of potential outcomes. The temporal structure enables us to regress outcomes directly on normalized treatment indicators from multiple past periods. This necessitates asymptotic analysis of regression-based estimators and HAC variance estimation in this section. The causal estimands, identification strategy, and asymptotic theory in time-series experiments are fundamentally different from and not covered by the existing network experiments literature.
\end{remark}


In the decomposition \eqref{eq:hat_tau}, we can write
\begin{align}
  \frac{1}{T-K} \mZ_K^\top (\mY_K - \mZ_K \mW_K \tau) = \frac{1}{T-K} \sum_{t=K+1}^T \tZ_{t:t-K} (Y_t - \tZ_{t:t-K}^\top \mW_K \tau).
\end{align} 
Consider the sequence $[\tZ_{t:t-K} (Y_t - \tZ_{t:t-K}^\top \mW_K \tau)]_{t=K+1}^\infty$. Under Assumption~\ref{asp:ani}, for a fixed $K$, the carryover effects of the sequence decay with the same asymptotic order as $\theta_{T,k}$. Under Assumption~\ref{asp:mdep}, the sequence is $(K \vee m)$-dependent. Therefore, we can establish the consistency and apply the central limit theorem (CLT) for the weak dependence sequence \citep{chen2004normal, chandrasekhar2023general}.

To ensure well-behaved asymptotic properties of $\hat\tau$, we impose the following standard moment conditions on the potential outcomes.

\begin{assumption}[Moment condition]\label{asp:moment}
  The potential outcomes $Y_t(z_{t:1})$ are uniformly bounded for all $t$ and $z_{t:1} \in \{0,1\}^t$, i.e., $\sup_{t\ge 1} \sup_{z_{t:1} \in \{0,1\}^t} \lvert Y_t(z_{t:1}) \rvert < \infty$. Additionally, $\lVert \tau \rVert_1$ is uniformly bounded for all $T$, i.e., $\sup_{T \ge 1} \lVert \tau \rVert_1 < \infty$.
\end{assumption}

In Assumption~\ref{asp:moment}, if we further assume $K$ is fixed, the uniformly boundedness of $\lVert \tau \rVert_1$ follows directly from the uniformly boundedness of the potential outcomes $Y_t$. However, when $K$ is diverging, the uniformly boundedness of $\lVert \tau \rVert_1$ implies that $\tau_k$ decays to zero as $k \to \infty$, which indicates that the treatment effect decays over time. It is possible to relax the uniformly boundedness assumption on $\tau$, but it will complicate the asymptotic analysis. We impose Assumption~\ref{asp:moment} for the simplicity of presentation. 

We first show the consistency of $\hat\tau$ in the following theorem.

\begin{theorem}[Consistency]\label{thm:consistency}
  Assume Assumption~\ref{asp:no-anticipation}, Assumption~\ref{asp:design}, Assumption~\ref{asp:ps}, Assumption~\ref{asp:ani}\ref{asp:ani:i} with $K \sum_{k=1}^T \theta_{T,k}/T \to 0$ or Assumption~\ref{asp:mdep}\ref{asp:mdep:i} with $K m/T \to 0$, Assumption~\ref{asp:moment}, and $K^2/T \to 0$. Then we have
  \begin{align*}
    \lVert \hat\tau - \tau \rVert_2 \stackrel{\sf p}{\longrightarrow} 0.
  \end{align*}
\end{theorem}

Since the dimension of $\tau$ can potentially grow with $T$, we establish the CLT for a linear projection of $\hat\tau$ rather than $\hat\tau$ itself. To characterize the asymptotic variance, we introduce a covariance matrix $\mV \in \bR^{(K+1) \times (K+1)}$ below. Since $\hat\tau$ is the estimate of $\tau$ and $\tilde\tau = \mW_K \hat\tau$ by linear transformation \eqref{eq:hat_tau_k}, $\tilde\tau$ is the estimate of $\mW_K \tau$. We define the oracle residual of OLS regression \eqref{eq:ols} at time $t$ as 
\begin{align}\label{eq:residual}
  U_t = Y_t - \tZ_{t:t-K}^\top \mW_K \tau = Y_t - \sum_{k=0}^K \tZ_{t-k} w_k \tau_k.
\end{align}
Then by \eqref{eq:hat_tau}, for $k,k'=1,\ldots,K+1$, define the $k,k'$-th element of $\mV$ as:
\begin{align}\label{eq:mV}
  [\mV]_{k, k'} = \frac{1}{T-K}\Cov\Big[\sum_{t=K+1}^T \tZ_{t-k+1} U_t, \sum_{t=K+1}^T \tZ_{t-k'+1} U_t \Big].
\end{align}

By definition of $\mV$ in \eqref{eq:mV}, $\mV$ depends on $T$ but not on the observed treatment path and potential outcomes path. This covariance matrix captures the dependence structure of the residuals across different lags. To ensure the limit distirbution of $\hat\tau$ does not degenerate, we need the following assumption on the smallest eigenvalue of $\mV$.

\begin{assumption}[Eigenvalue condition]\label{asp:eigen}
  The smallest eigenvalue of $\mV$ is uniformly bounded away from zero for all $T$, i.e., $\inf_{T \ge 1} \lambda_{\min}(\mV) > 0$, where $\lambda_{\min}(\mV)$ is the smallest eigenvalue of $\mV$.
\end{assumption}

We can then establish the CLT for the linear projection of $\hat\tau$ in the following theorem.

\begin{theorem}[Asymptotic normality]\label{thm:clt}
  Assume Assumption~\ref{asp:no-anticipation}, Assumption~\ref{asp:design}, Assumption~\ref{asp:ps}, Assumption~\ref{asp:ani}\ref{asp:ani:ii} with $K$ fixed or Assumption~\ref{asp:mdep}\ref{asp:mdep:ii} with $K^4/T \to 0$, Assumption~\ref{asp:moment}, and Assumption~\ref{asp:eigen}. Then for any $\lambda_K \in \bR^{K+1}$ satisfying $0 < \inf_{T \ge 1} (\lVert \lambda_K \rVert_2/\lVert \lambda_K \rVert_1) \le \sup_{T \ge 1} (\lVert \lambda_K \rVert_2/\lVert \lambda_K \rVert_1) < \infty$, we have
  \begin{align*}
    \sqrt{T-K} (\lambda_K^\top \mV \lambda_K)^{-1/2} \lambda_K^\top (\hat\tau - \tau) \stackrel{\sf d}{\longrightarrow} N(0,1).
  \end{align*}
\end{theorem}

By the Cramér–Wold theorem, Theorem~\ref{thm:clt} directly implies a multivariate CLT for $\hat\tau$ when $K$ is fixed.

\begin{corollary}\label{cor:clt}
    Under the conditions of Theorem~\ref{thm:clt}, if $K$ is fixed, we have
    \begin{align*}
      \sqrt{T-K} \mV^{-1/2} (\hat\tau - \tau) \stackrel{\sf d}{\longrightarrow} N(0, \mI_{K+1}).
    \end{align*}
\end{corollary}

In Appendix, we establish three lemmas for proving Theorem~\ref{thm:consistency} (Consistency) and Theorem~\ref{thm:clt} (Asymptotic normality). Recall the decomposition of $\hat\tau$ in \eqref{eq:hat_tau}. The first lemma shows that $(T-K)^{-1} \mZ_K^\top \mZ_K$ asymptotically converges to $\mW_K^{-1}$ given Assumption~\ref{asp:ps}. The second and third lemmas show the asymptotic behavior of $(T-K)^{-1} \mZ_K^\top (\mY_K - \mZ_K \mW_K \tau)$ given all other assumptions.


\subsection{HAC variance estimation}\label{sec:hac}

Recall the oracle residual $U_t$ defined in \eqref{eq:residual}. Let 
\begin{align}\label{eq:residual_sample}
\hat U_t = Y_t - \tZ_{t:t-K}^\top \tilde\tau
\end{align}
be the sample residual of OLS regression \eqref{eq:ols} at time $t$. Let $\mU_K = {\rm diag}(U_t:t=K+1,\ldots,T) \in \bR^{(T-K) \times (T-K)}$ and $\hat \mU_K = {\rm diag}(\hat U_t:t=K+1,\ldots,T) \in \bR^{(T-K) \times (T-K)}$ be the diagonal matrices of the oracle and sample residuals, respectively.

To account for autocorrelation, we introduce a kernel weighting matrix $\mQ_K \in \bR^{(T-K) \times (T-K)}$. Specifically, we use the Bartlett kernel \citep{bartlett1950periodogram}, defined entrywise as
\[
[\mQ_K]_{ij} =
\begin{cases}
1 - \dfrac{\lvert j-i \rvert}{L+1}, & \lvert j-i \rvert \le L, \\
0, & \lvert j-i \rvert > L,
\end{cases}
\]
for $1 \le i,j \le T-K$, where $L$ is the bandwidth parameter controlling the effective dependence range. A common choice is $L = \lfloor T^{1/4} \rfloor$ \citep{greene2003econometric}. By construction, $\mQ_K$ is positive semidefinite. The HAC estimator with the Bartlett kernel corresponds to Newey–West estimator \citep{newey1987simple}. Our results can be generalized to HAC estimators with other kernels, see \cite{andrews1991heteroskedasticity} for different choices of kernels.

To construct the HAC variance estimator for $\hat\tau$, we first define the HAC variance estimator for $\tilde\tau$ by OLS regression \eqref{eq:ols} as:
\begin{align*}
  \tilde\mV = (T-K) (\mZ_K^\top \mZ_K)^{-1} (\mZ_K^\top \hat\mU_K \mQ_K \hat\mU_K \mZ_K) (\mZ_K^\top \mZ_K)^{-1},
\end{align*}
and then define the HAC variance estimator for $\hat\tau$ by linear transformation \eqref{eq:hat_tau_k} as:
\begin{align}\label{eq:hac_var}
  \hat\mV = \mW_K^{-1} \tilde\mV \mW_K^{-1} = (T-K) \mW_K^{-1} (\mZ_K^\top \mZ_K)^{-1} (\mZ_K^\top \hat\mU_K \mQ_K \hat\mU_K \mZ_K) (\mZ_K^\top \mZ_K)^{-1} \mW_K^{-1}.
\end{align}
Intuitively, the HAC variance estimator corrects the usual OLS variance formula by allowing residuals at nearby time points to be correlated. The kernel matrix $\mQ_K$ assigns larger weights to correlations between residuals that are closer in time and gradually downweights correlations as the lag length increases, with the bandwidth $L$ determining how far into the past these correlations are taken into account. This adjustment ensures that the estimated variance remains consistent under general forms of heteroskedasticity and autocorrelation in the sampling-based inference literature. We borrow their intuition and analyze $\hat\mV$ from the design-based framework.

Recall the lag-$k$ treatment effect defined in \eqref{eq:tau_t}. Define $\mb_{K} \in \bR^{(T-K) \times (K+1)}$ with $$[\mb_{K}]_{t-K,k+1} = \tau_{t,k} - \frac{w_k}{p_{t-K}(1-p_{t-K})}\tau_k,$$ for $t=K+1,\ldots,T$ and $k=0,\ldots,K$, which measures the heterogeneity of the treatment effect and treatment probability across time points. We then define the bias term as
\[
\mB_K = \mb_K^\top \mQ_K \mb_K.
\]
By construction, $\mB_K$ is positive semidefinite because $\mQ_K$ is positive semidefinite under the Bartlett kernel. Intuitively, $\mB_K$ captures the bias component in the HAC variance estimator that arises when treatment probabilities vary over time or when treatment effects are heterogeneous across time points. The following theorem establishes that $\hat\mV$ is a conservative estimator of $\mV$.
\begin{theorem}\label{thm:var}Assume Assumption~\ref{asp:no-anticipation}, Assumption~\ref{asp:design}, Assumption~\ref{asp:ps}, Assumption~\ref{asp:ani}\ref{asp:ani:ii} with $K$ fixed, $L \to \infty$ or Assumption~\ref{asp:mdep}\ref{asp:mdep:ii} with $K^4/T \to 0, Km^2 / L \to 0$, Assumption~\ref{asp:moment} and $K^4 L^2/T \to 0$. Then
\begin{align*}
  \lVert \hat\mV - \mV - \mB_K \rVert_{\rm F} \stackrel{\sf p}{\longrightarrow} 0.
\end{align*}
If the treatment probabilities are constant across time ($p_t = p$ for all $t$) and the treatment effects are homogeneous over time ($\tau_{t,k} = \tau_k$ for all $t$ and $k$), then $\mb_{K} = 0$, which implies $\mB_K = 0$. In this case, the HAC variance estimator $\hat \mV$ is consistent for $\mV$, i.e., $\lVert \hat\mV - \mV\rVert_{\rm F} \stackrel{\sf p}{\longrightarrow} 0$.
\end{theorem}

Theorem~\ref{thm:var} generalizes the classical variance results of \citet{neyman1923application} for completely randomized experiments. In the special case with constant treatment probabilities and homogeneous treatment effects, the additional term $\mB_K$ vanishes, and the HAC variance estimator consistently estimates the true variance. This mirrors \cite{neyman1923application}’s result that the randomization-based variance estimator is consistent under constant effects. When treatment effects or assignment probabilities vary over time, however, the extra positive semidefinite term $\mB_K$ inflates the variance, making the HAC estimator conservative, just as the randomization-based variance estimator is conservative without the constant effects assumption. In this sense, our framework extends Neyman’s design-based analysis from randomized experiments to time-series experiments.

\begin{remark}
  Equivalently, from \eqref{eq:ols} and \eqref{eq:hat_tau_k}, $\hat\tau$ can be obtained directly from the following OLS regression without an intercept:
  \begin{align}\label{eq:ols_w}
    \textbf{lm}(Y_t \sim w_0 \tZ_t + w_1 \tZ_{t-1} + \cdots + w_K \tZ_{t-K}).
  \end{align}
  In this formulation, the outcome vector remains $\mY_K$, but the regressor matrix is rescaled to $\mZ_K \mW_K$. The only difference lies in whether the scaling is applied before or after the OLS regression. Thus, the regression in \eqref{eq:ols_w} yields the same estimator as in \eqref{eq:ols} after applying the linear transformation \eqref{eq:hat_tau_k}, and both procedures lead to the same HAC variance estimator. When the $p_t$'s are constant across time $t = 1,\ldots,T$, all treatments share the same normalization with the same centering and scaling. In this case, there is no need for the normalization and the OLS fit is equivalent to regressing the observed outcomes on unnormalized treatment indicators with intercepts: $\textbf{lm}(Y_t \sim 1 + Z_t + Z_{t-1} + \ldots + Z_{t-K})$ using samples with $t = K+1,\ldots,T$.
\end{remark}

\subsection{Fisher randomization test}
The Fisher randomization test (FRT) is a powerful tool for analyzing randomized experiments \citep{imbens2015causal,ding2024first}. We can consider two types of hypothesis tests: the sharp null hypothesis and the weak null hypothesis. The sharp null hypothesis states that the potential outcomes are identical under all treatment assignments: 
\[
  H_0: Y_t(z_{t:1}) = Y_t(z'_{t:1}) \text{ for any } z_{t:1},z'_{t:1} \in \{0,1\}^t \text{ and for all } t=\{1,\ldots,T\},
\]
which allows for an exact FRT. The weak null hypothesis tests whether the treatment effects are zero only for a subset of lags:
\[
  H_0: \tau_k = 0 \text{ for all } k \in \cS,
\]
where $\cS \subset \{0,\ldots,K\}$ denotes an arbitrary set of lags we wish to test, which leads to an asymptotic test by using the asymptotic distribution of $\hat\tau$ (Theorem~\ref{thm:clt}) and the HAC variance estimator $\hat\mV$ (Theorem~\ref{thm:var}). One can run FRTs with studentized statistics based on our regression estimator $\hat\tau$ together with the HAC variance estimator $\hat\mV$ \citep{wu2021randomization}. Studentization works here because a conservative variance estimator $\hat\mV$ is available in general, which is also consistent under the sharp null hypothesis. This highlights the value of our asymptotic analysis: even when the goal is to perform a FRT, having a conservative variance estimator is crucial to apply studentization. 

Importantly, the definition of the weak null hypothesis differs across frameworks. In \citet{bojinov2019time}, the weak null hypothesis is formulated in terms of their path-dependent causal estimands and focuses on a single lag, reflecting the estimands defined in their design-based analysis. In our setting, the estimands are multiple lagged treatment effects, so the weak null hypothesis naturally takes a different form, where we test multiple lag-specific effects simultaneously. Our contribution is to show how such regression-based estimators can be used to construct asymptotic tests for lag-specific treatment effects, providing a complementary perspective that is closely aligned with standard regression practice.

\subsection{Relationship between full OLS, marginal OLS and weighted least squares}\label{sec:full_marginal}

Although the OLS procedure in \eqref{eq:ols} involves regressing $Y_t$ on the full set of transformed treatment indicators $\tZ_t,\ldots,\tZ_{t-K}$, all of our results remain valid when regressing $Y_t$ on a subset of these transformed treatment indicators. Suppose we are interested in estimating the treatment effect at lag $k$. There are two natural approaches. One approach is the full OLS in \eqref{eq:ols}, where we regress $Y_t$ on $\tZ_t,\ldots,\tZ_{t-K}$ to obtain the estimate $\tilde\tau_k$, and then apply the linear transformation \eqref{eq:hat_tau_k} to obtain the estimate $\hat\tau_k = w_k^{-1} \tilde\tau_k$. The other approach is the marginal OLS, where we regress $Y_t$ on $\tZ_{t-k}$, i.e.,
\begin{align}\label{eq:marginal_ols}
  \textbf{lm}(Y_t \sim \tZ_{t-k}),
\end{align}
to obtain the estimate $\tilde\tau_{{\rm marginal},k}$ and then apply the linear transformation \eqref{eq:hat_tau_k} to obtain the estimate $\hat\tau_{{\rm marginal},k} = w_k^{-1} \tilde\tau_{{\rm marginal},k}$.
Both approaches yield consistent and asymptotic normal estimators ($\hat\tau_k$ and $\hat\tau_{{\rm marginal},k}$) for $\tau_k$ with conservative variance estimators by HAC variance estimation. The key differences lie in their efficiency and variance estimation. First, the two approaches have different asymptotic variances. Second, the two approaches have different variance estimators, since the variance estimator in the marginal OLS is based on noisier residuals, which inflates the standard errors. As a direct result from Theorem~\ref{thm:var}, both approaches share the same bias term $\mB_K$ in HAC variance estimation. Thus, comparing the HAC variance estimators $\hat\mV$ reduces to comparing the asymptotic variances $\mV$, since the two approaches have the same level of conservativeness. See Appendix for more technical details on this discussion. Therefore, we focus on comparing the asymptotic variances below.

By Theorem~\ref{thm:clt}, the asymptotic variance from the full OLS \eqref{eq:ols} is $(T-K)^{-1} \Var[\sum_{t=K+1}^T \tZ_{t-k} (Y_t - \sum_{\ell=0}^K \tZ_{t-\ell} w_\ell \tau_\ell)],$ while the asymptotic variance from the marginal OLS \eqref{eq:marginal_ols} is $(T-K)^{-1} \Var[\sum_{t=K+1}^T \tZ_{t-k} (Y_t - \tZ_{t-k} w_k \tau_k)]$. Intuitively, the asymptotic variance depends on the residual $Y_t - \sum_{\ell=0}^K \tZ_{t-\ell} w_\ell \tau_\ell$ for full OLS and $Y_t - \tZ_{t-k} w_k \tau_k$ for marginal OLS. By adding more regressors in full OLS, the residual should have lower variance, which imples the asymptotic variance will be smaller. To compare these asymptotic variances explicitly, we consider the linear model on the potential outcomes and establish the following result.
\begin{proposition}\label{prop:var}
  Consider a linear model with homogeneous effects and nonrandom error: $Y_t = Y_t(z_{t:1})=\sum_{k=0}^{t-1} \beta_k z_{t-k} + \epsilon_t$. Assume a constant treatment probability $p_t=p$ for all $t$. Then, the asymptotic variance of $\hat\tau_k$ from the full OLS is
  \begin{align*}
    \frac{1}{T-K} \sum_{t=K+1}^T \sum_{\ell=K+1}^{t-1} \beta_\ell^2 + \frac{1}{T-K} \frac{1}{p(1-p)}\sum_{t=K+1}^T \Big(p\sum_{k=0}^{t-1} \beta_k + \epsilon_t \Big)^2,
  \end{align*}
  while the asymptotic variance of $\hat\tau_{{\rm marginal},k}$ from the marginal OLS is
  \begin{align*}
    \frac{1}{T-K} \sum_{t=K+1}^T \sum_{\ell=0,\ell \neq k}^{t-1} \beta_\ell^2 + \frac{1}{T-K} \frac{1}{p(1-p)}\sum_{t=K+1}^T \Big(p\sum_{k=0}^{t-1} \beta_k + \epsilon_t \Big)^2.
  \end{align*}
\end{proposition}

By Proposition~\ref{prop:var}, under the linear model with homogeneous effects and nonrandom error, full OLS has a smaller asymptotic variance than marginal OLS. The variance difference is $(T-K)^{-1} \sum_{t=K+1}^T \sum_{\ell = 0,\ell \neq k}^K \beta_\ell^2 = \sum_{\ell = 0,\ell \neq k}^K \beta_\ell^2$ which is typically large, as smaller lags often have strong effects. Furthermore, full OLS achieves lower variance unless $\beta_\ell = 0$ for all $\ell > K$. However, including too many lags may lead to high model complexity. Therefore, it is crucial for practitioners to carefully choose the number of lags in the model. This selection can be guided by standard model selection techniques from the linear regression literature.

\begin{remark}
  \cite{gao2023causal} consider the network experiment setting and estimate the average of potential outcomes separately for each group. Compared with the weighted least squares (WLS) approach in \cite{gao2023causal}, our method differs in two key ways. First, we estimate the causal effect directly rather than estimating the average potential outcomes. Second, we normalize the treatment indicator directly, instead of applying weights to each observation with an unnormalized treatment indicator. Our OLS approach allows us to estimate treatment effects across multiple lags simultaneously. Since treatment indicators at different time points should receive different weights due to variations in propensity scores, and each observation contains treatment indicators from multiple time points, assigning a single weight per observation is not feasible. However, the marginal OLS estimator \eqref{eq:marginal_ols} is asymptotically equivalent with a version of WLS analogous to \cite{gao2023causal}. If further the propensity scores are constant, the two estimators coincide. See Appendix for more details on the equivalence between two methods.
\end{remark}

\subsection{Comparisons with the existing estimands}\label{sec:comparisons}

\cite{bojinov2019time}  defined the general treatment effect of lag $k$ at time $t$ as:
\begin{align*}
  \sum_{\substack{z_{t:t-k+1} \in \{0,1\}^k \\ z_{t-k-1:1} \in \{0,1\}^{t-k-1}}} a_{z_{t:t-k+1}, z_{t-k-1:1}}[Y_t(z_{t:t-k+1},1,z_{t-k-1:1}) - Y_t(z_{t:t-k+1},0,z_{t-k-1:1})],
\end{align*}
for some non-stochastic weights satisfying $\sum_{z_{t:t-k+1}, z_{t-k-1:1}}a_{z_{t:t-k+1}, z_{t-k-1:1}} = 1$ and $a_{z_{t:t-k+1}, z_{t-k-1:1}} \ge 0$. Then the general average treatment effect is defined as the corresponding average over all time points, analogous to the definition of $\tau_k$ based on $\tau_{t,k}$.

By the design under Assumption~\ref{asp:design}, we can write $\tau_{t,k}$ in \eqref{eq:tau_t} explicitly as
\begin{align*}
  \tau_{t,k} = \sum_{\substack{z_{t:t-k+1} \in \{0,1\}^k \\ z_{t-k-1:1} \in \{0,1\}^{t-k-1}}} a_{z_{t:t-k+1}, z_{t-k-1:1}}[Y_t(z_{t:t-k+1},1,z_{t-k-1:1}) - Y_t(z_{t:t-k+1},0,z_{t-k-1:1})],
\end{align*}
where
\[
  a_{z_{t:t-k+1}, z_{t-k-1:1}} = \prod_{s=1, s \neq t-k}^t p_s^{z_s} (1-p_s)^{1-z_s}.
\]
Therefore, $\tau_{t,k}$ is a special case of the general treatment effect in \cite{bojinov2019time}. However, \cite{bojinov2019time} argued that the general average treatment effect cannot be estimated without strong assumptions, so they changed the estimand to the causal effect as a function of the observed treatment path, which is random (See Section 3.2 in \cite{bojinov2019time}). In contrast, we focus on a special class of the general treatment effects, which does not depend on the observed treatment path, and can be estimated by OLS under time-series experiments with independent treatments.

In general, we can adjust the weights in $\tZ_t$ to estimate a broad class of average treatment effects. Consider a new normalization: $\tZ_{h,t} = h_t(p_1,\ldots,p_t)(Z_t - \E[Z_t])/\Var[Z_t]$ for some function $h_t:\{0,1\}^t \to \bR$ at each time $t$. Define the weight as $w_{h,k} = [(T-K)^{-1} \sum_{t=K+1}^T h_t(p_1,\ldots,p_t)^2 (p_{t-k}(1-p_{t-k}))^{-1}]^{-1} $. We then consider the OLS regression without an intercept: $\textbf{lm}(Y_t \sim \tZ_{h,t} + \tZ_{h,t-1} + \cdots + \tZ_{h,t-K})$ using observations from $t = K+1,\ldots,T$, and let $(\tilde\tau_{h,0},\ldots,\tilde\tau_{h,K})$ denote the OLS coefficients. For each $k=0,\ldots,K$, we define $
\hat\tau_{h,k} = w_{h,k}^{-1} \tilde\tau_{h,k}$. The corresponding estimands of this procedure are $(\tau_{h,0},\ldots,\tau_{h,K})$, where 
\begin{align*}
  \tau_{h,k} =&\frac{1}{T-K} \sum_{t=K+1}^T h_t(p_1,\ldots,p_t) \tau_{t,k}\\
  =& \frac{1}{T-K} \sum_{t=K+1}^T h_t(p_1,\ldots,p_t) \E[Y_t(Z_{t:t-k+1},1,Z_{t-k-1:1}) - Y_t(Z_{t:t-k+1},0,Z_{t-k-1:1})].
\end{align*}

\cite{liang2025randomization} also studied time-series experiments and the corresponding design-based analysis. Our paper differs from theirs in two key ways. First, their analysis assumes a linear convolution model (their Section 2) and a circular convolution model (their Section 3), along with a Bernoulli randomization scheme where $p_t = 1/2$ for all $t$. In contrast, our framework does not impose any structural assumptions on the potential outcomes, allowing for arbitrary and unknown heterogeneous effects over time. Second, we consider more general experimental designs with time-varying treatment probabilities. Third, \cite{liang2025randomization} proposed a method-of-moments estimator, established a central limit theorem, and proposed a variance estimator. In our work, we focus on OLS estimation and analyze its asymptotic properties, along with the performance of the HAC variance estimator derived from the same OLS procedure. Both OLS and HAC variance estimation are widely used in practice, and are implemented in standard statistical software routines. We provide theoretical support for their application in time-series experiments.

\section{Extension to continuous treatment}\label{sec:extensions}

Many modern experimental settings feature treatments that take discrete or continuous values rather than binary values. For instance, in online platforms, treatment intensity may correspond to the proportion of user exposure to a new algorithm or the duration for which an experimental interface is displayed. In clinical trials, treatments often involve dosage levels or exposure durations, which are naturally measured on a continuous scale. In this section, we extend our framework to accommodate treatments with continuous values.

In settings where the experiment involves multiple discrete treatment levels in addition to a control group, one natural approach is to run separate OLS regressions for each treatment level relative to control. Our results then can be applied. Therefore, we only focus on the case where the treatment variable $Z_t \in \bR$ is continuous. We introduce the following assumption, which is the continuous treatment analogue of Assumption~\ref{asp:design}.

\begin{assumption}[Time-series experiment for continuous treatment]\label{asp:design_cont}
  Assume the treatments $Z_1,\ldots,Z_T$ are independent. Assume the values of $(\E[Z_1],\ldots,\E[Z_T])$ and $(\Var[Z_1],\ldots,\Var[Z_T])$ exist and are known by the design.
\end{assumption}

As in the binary treatment case, we normalize the treatment variable by considering $\tZ_t = (Z_t - \E[Z_t])/\Var[Z_t]$. We then fit the following OLS regression model: $$\textbf{lm}(Y_t \sim \tZ_t  + \cdots + \tZ_{t-K}),$$ where the coefficient vector is denoted by $\tilde\tau$ and its transformed version by $\hat\tau = \mW_K^{-1} \tilde\tau$ as \eqref{eq:hat_tau_ols} before with the same $\mW_K$. The causal estimand is the vector of lagged treatment effects $\tau= (\tau_0,\ldots,\tau_K)^\top \in \bR^{K+1}$, where the treatment effect of lag-$k$ treatment on the outcome at time $t$ is defined as
\begin{align*}
  \tau_{t,k} = \E\left[\frac{(Z_{t-k} - \E[Z_{t-k}])}{\Var[Z_{t-k}]} Y_t\right],
\end{align*}
and the average treatment effect of lag-$k$ treatment on the outcome over time is defined as
\begin{align*}
  \tau_k = \frac{1}{T-K} \sum_{t=K+1}^T \tau_{t,k},
\end{align*}
which extend the notion of the treatment effect of lag-$k$ in \eqref{eq:tau_t} and \eqref{eq:tau} to the continuous treatment setting. This estimand $\tau_k$ represents the average slope of the potential outcome with respect to the treatment assigned $k$ periods earlier, while averaging over all other past and concurrent treatment assignments.

In the literature on continuous treatments, the potential outcome under treatment intensity $z$ is referred to as the (causal) dose-response curve, denoted $Y(z)$, while its derivative $(\d/\d z) Y(z)$ is known as the causal derivative effect; see \citet{zhang2025doubly} and references therein. Building on this perspective, we adapt the notion of derivative effects to the setting of time-series experiments. Specifically, we define the effect of treatment at lag-$k$ on the outcome at time $t$ as
\begin{align*}
  \tau_{t,k}'(z) = \frac{\d}{\d z} \E[Y_t(Z_{t:t-k+1},z,Z_{t-k-1:1})],
\end{align*}
where the expectation is taken over the joint distribution of $(Z_{t:t-k+1}, Z_{t-k-1:1})$. Intuitively, $\tau_{t,k}'(z)$ measures the local sensitivity of the expected outcome at time $t$ to a marginal change in treatment intensity at lag-$k$, holding the rest of the treatment path fixed. Aggregating across time, we can equivalently represent the estimand $\tau_k$ as a weighted average of these local effects:
\begin{align*}
  \tau_k = \frac{1}{T-K} \sum_{t=K+1}^T \left(\int w_{t,k}(z) \d z\right)^{-1} \int \tau_{t,k}'(z) w_{t,k}(z) \d z,
\end{align*}
where the weighting function $w_{t,k}(z)$ is given by
\begin{align*}
  w_{t,k}(z) = \{\E[Y_t(Z_{t:1}) \given Z_{t-k} \ge z] - \E[Y_t(Z_{t:1}) \given Z_{t-k} < z]\} \P(Z_{t-k} 
  \ge z) \P(Z_{t-k} < z).
\end{align*}

Thus, the estimand $\tau_k$ can be interpreted as the time-averaged, weighted average derivative effect at lag $k$. This places our estimand within the broader class of weighted average derivatives \citep{newey1993efficiency}, and connects naturally to discussions of how OLS and instrument variable identify such objects \citep{graham2022semiparametrically, borusyak2024negative}. Importantly, both the dose–response curve and the derivative effect are nonregular parameters and cannot generally be estimated at parametric rates. In contrast, we consider the weighted average of the derivative effect, which admits parametric convergence rates.

We next impose regularity conditions on the treatment distribution and outcome to extend our asymptotic results to the continuous treatment setting. We consider the following assumption on the distribution of $Z_t$, which serves as a continuous treatment analogue of Assumption~\ref{asp:ps}.

\begin{assumption}[Treatment distributions  condition]\label{asp:ps_cont}
  There exists some constant $\epsilon>0$ such that $\epsilon \le \inf_{t \ge 1} \Var[Z_t]$.
\end{assumption}

We also consider the following assumption, which serves as a continuous treatment analogue of Assumption~\ref{asp:moment}.
\begin{assumption}[Moment condition]\label{asp:moment_cont}
  The potential outcomes $Y_t(z_{t:1})$ are uniformly bounded for all $t$ and all $z_{t:1} \in \bR^t$, i.e., $\sup_{t \ge 1, z_{t:1} \in \bR^t} \lvert Y_t(z_{t:1}) \rvert < \infty$. Additionally, $\lVert \tau \rVert_1$ is uniformly bounded for all $T$, i.e., $\sup_{T} \lVert \tau \rVert_1 < \infty$.
\end{assumption}

With these conditions in place, we obtain the consistency and asymptotic normality of the OLS estimator. These results are the continuous treatment counterparts to Theorem~\ref{thm:consistency} and Theorem~\ref{thm:clt}.

\begin{theorem}[Consistency]\label{thm:consistency_cont}
  Assume Assumption~\ref{asp:no-anticipation}, Assumption~\ref{asp:design_cont}, Assumption~\ref{asp:ps_cont}, Assumption~\ref{asp:ani}\ref{asp:ani:i} with $K \sum_{k=1}^T \theta_{T,k}/T \to 0$ or Assumption~\ref{asp:mdep}\ref{asp:mdep:i} with $K m/T \to 0$, Assumption~\ref{asp:moment_cont}, and $K^2/T \to 0$. Then we have
  \begin{align*}
    \lVert \hat\tau - \tau \rVert_2 \stackrel{\sf p}{\longrightarrow} 0.
  \end{align*}
\end{theorem}

\begin{theorem}[Asymptotic normality]\label{thm:clt_cont}
  Assume Assumption~\ref{asp:no-anticipation}, Assumption~\ref{asp:design_cont}, Assumption~\ref{asp:ps_cont}, Assumption~\ref{asp:ani}\ref{asp:ani:ii} with $K$ fixed or Assumption~\ref{asp:mdep}\ref{asp:mdep:ii} with $K^4/T \to 0$, Assumption~\ref{asp:moment_cont}, and Assumption~\ref{asp:eigen}. Then for any $\lambda_K \in \bR^{K+1}$ satisfying $0 < \inf_{T \ge 1} (\lVert \lambda_K \rVert_2/\lVert \lambda_K \rVert_1) \le \sup_{T \ge 1} (\lVert \lambda_K \rVert_2/\lVert \lambda_K \rVert_1) < \infty$, we have
  \begin{align*}
    \sqrt{T-K} (\lambda_K^\top \mV \lambda_K)^{-1/2} \lambda_K^\top (\hat\tau - \tau) \stackrel{\sf d}{\longrightarrow} N(0,1),
  \end{align*}
  where $\mV$ is defined in \eqref{eq:mV}.
\end{theorem}

To estimate the asymptotic variance $\mV$ in Theorem~\ref{thm:clt_cont}, we can use exactly the same HAC variance estimator $\hat\mV$ in \eqref{eq:hac_var}. Define $\mb_{K} \in \bR^{(T-K) \times (K+1)}$ with $$[\mb_{K}]_{t-K,k+1} = \tau_{t,k} - \frac{w_k}{\Var[Z_{t-k}]}\tau_k,$$ for $t=K+1,\ldots,T$ and $k=0,\ldots,K$, and define the bias term as
\[
\mB_K = \mb_K^\top \mQ_K \mb_K.
\] The following theorem establish that $\hat\mV$ is a conservative estimator of $\mV$, which is parallel to Theorem~\ref{thm:var}.
\begin{theorem}\label{thm:var_cont}Assume Assumption~\ref{asp:no-anticipation}, Assumption~\ref{asp:design_cont}, Assumption~\ref{asp:ps_cont}, Assumption~\ref{asp:ani}\ref{asp:ani:ii} with $K$ fixed, $L \to \infty$ or Assumption~\ref{asp:mdep}\ref{asp:mdep:ii} with $K^4/T \to 0, Km^2 / L \to 0$, Assumption~\ref{asp:moment_cont} and $K^4 L^2/T \to 0$. Then
  \begin{align*}
    \lVert \hat\mV - \mV - \mB_K \rVert_{\rm F} \stackrel{\sf p}{\longrightarrow} 0.
  \end{align*}
  If the treatment variances are constant across time ($\Var[Z_t] = \sigma^2$ for all $t$) and the treatment effects are homogeneous over time ($\tau_{t,k} = \tau_k$ for all $t$ and $k$), then $\mb_{K} = 0$, which implies $\mB_K = 0$. In this case, the HAC variance estimator $\hat \mV$ is consistent for $\mV$, i.e., $\lVert \hat\mV - \mV\rVert_{\rm F} \stackrel{\sf p}{\longrightarrow} 0$.
\end{theorem}






\section{Empirical studies}\label{sec:empirical}

In this section, we conduct simulations to examine the finite-sample performance of $\hat\tau$ and the HAC variance estimator $\hat\mV$ for estimating $\mV$. We then illustrate the method on data from a trading experiment.

\subsection{Simulations}

We consider the autoregressive model in Example~\ref{ex:autoregressive}. Specifically, we take a first-order autoregressive model,
\[
  Y_t(z_{t:1}) = \mu_t(z_{t:1}) + \phi Y_{t-1}(z_{t-1:1}) + \epsilon_t(z_{t:1}),
\]
for any treatment history $z_{t:1} \in \{0,1\}^t$. We set $\mu_t(z_{t:1}) = \mu_t(z_t) = \mu(z_t)$ with $\mu(1)=0.5$ and $\mu(0)=0$, take $\phi = 0.5$, and let $\epsilon_t(z_{t:1}) = \epsilon_t \sim N(0,1)$ i.i.d across $t$. Under this specification, $Y_t(z_{t:1}) = \sum_{k=0}^{t-1} \phi^k \mu(z_{t-k}) + \sum_{k=0}^{t-1} \phi^k \epsilon_{t-k}$. Then the true lag-$k$ treatment effect is $\tau_k = \phi^k (\mu(1) - \mu(0))$.

We vary the time horizon $T \in \{100,200,500,1000,2000,5000,10000\}$, fix the number of lags at $K=5$, and use a constant treatment probability $p_t=p=0.5$. Adopting a design-based perspective, randomness arises only from the assignment path $Z_1,\ldots,Z_T$. For each $T$, we therefore fix a realization $[\epsilon_1,\ldots,\epsilon_T]$ and repeatedly resample $Z_1,\ldots,Z_T$. For each draw we compute the OLS estimator $\hat\tau$ and the HAC variance estimator $\hat\mV$ with the Bartlett kernel. A common choice is $L = \lfloor T^{1/4} \rfloor$ \citep{greene2003econometric}. Since the treatment probabilities are constant across time and the treatment effects are homogeneous over time, the HAC variance estimator $\hat\mV$ is consistent for the asymptotic variance $\mV$ by Theorem~\ref{thm:var}.

Figure \ref{fig1} presents histograms of $\hat\tau_k - \tau_k$ for $k=0,\ldots,5$ across replications under $T=1000$. The empirical distributions are close to the normal distribution, supporting our central limit theorem (Theorem \ref{thm:clt}). This provides empirical evidence that the regression-based estimator achieves approximate normality even at moderate sample sizes.

\begin{figure}[htbp]
  \centering
  \includegraphics[width=0.8\linewidth]{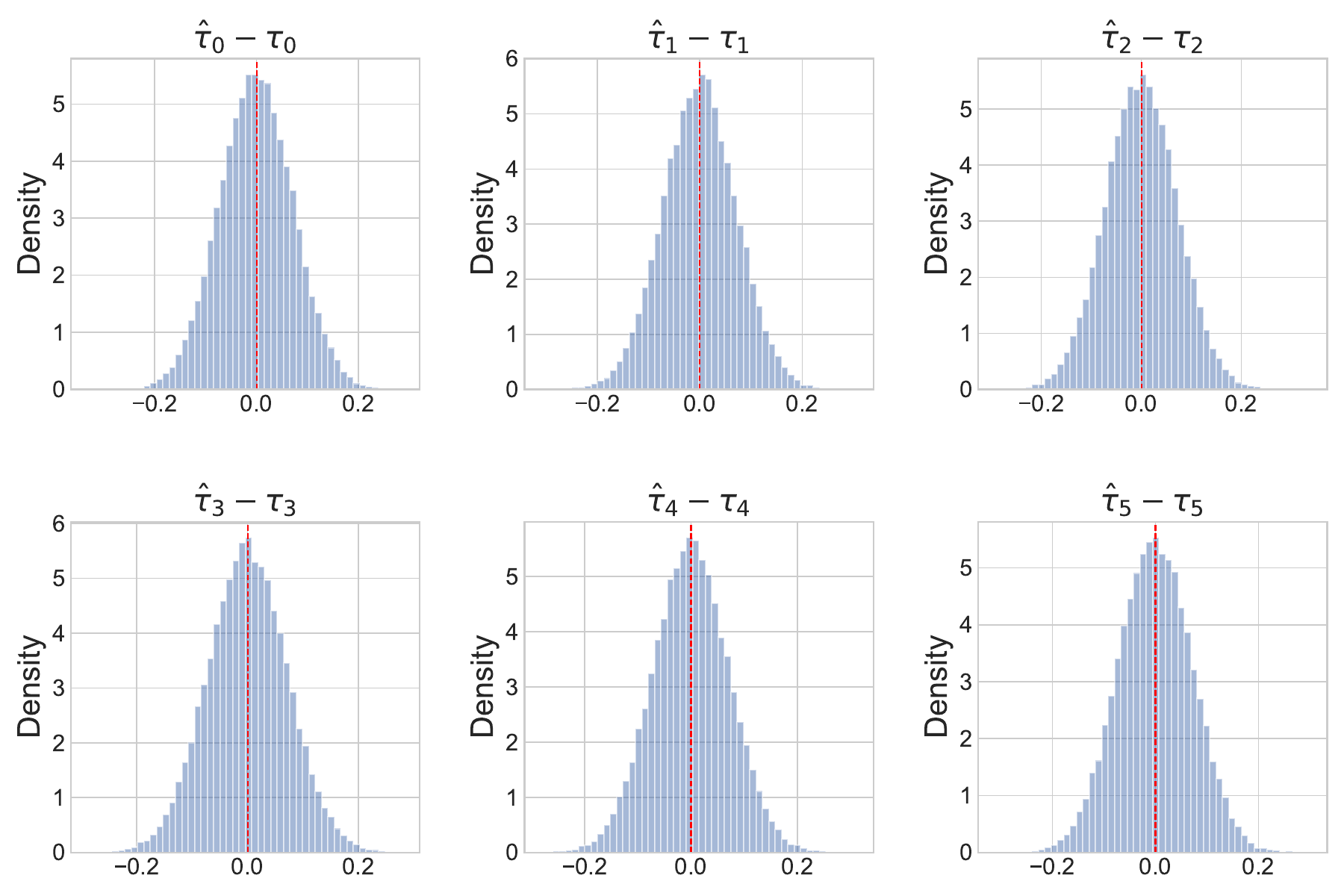}
  \caption{Histograms of $\hat\tau_k-\tau_k$ for $k=0,\ldots,5$.}
  \label{fig1}
\end{figure}

Figure \ref{fig2} displays violin plots of $\lVert \hat\tau - \tau \rVert_2$ as a function of $T$. The $L_2$ error decreases with larger sample sizes, aligning with the $T^{-1/2}$ convergence rate established in Theorem \ref{thm:clt}.

\begin{figure}[htbp]
  \centering
  \includegraphics[width=0.8\linewidth]{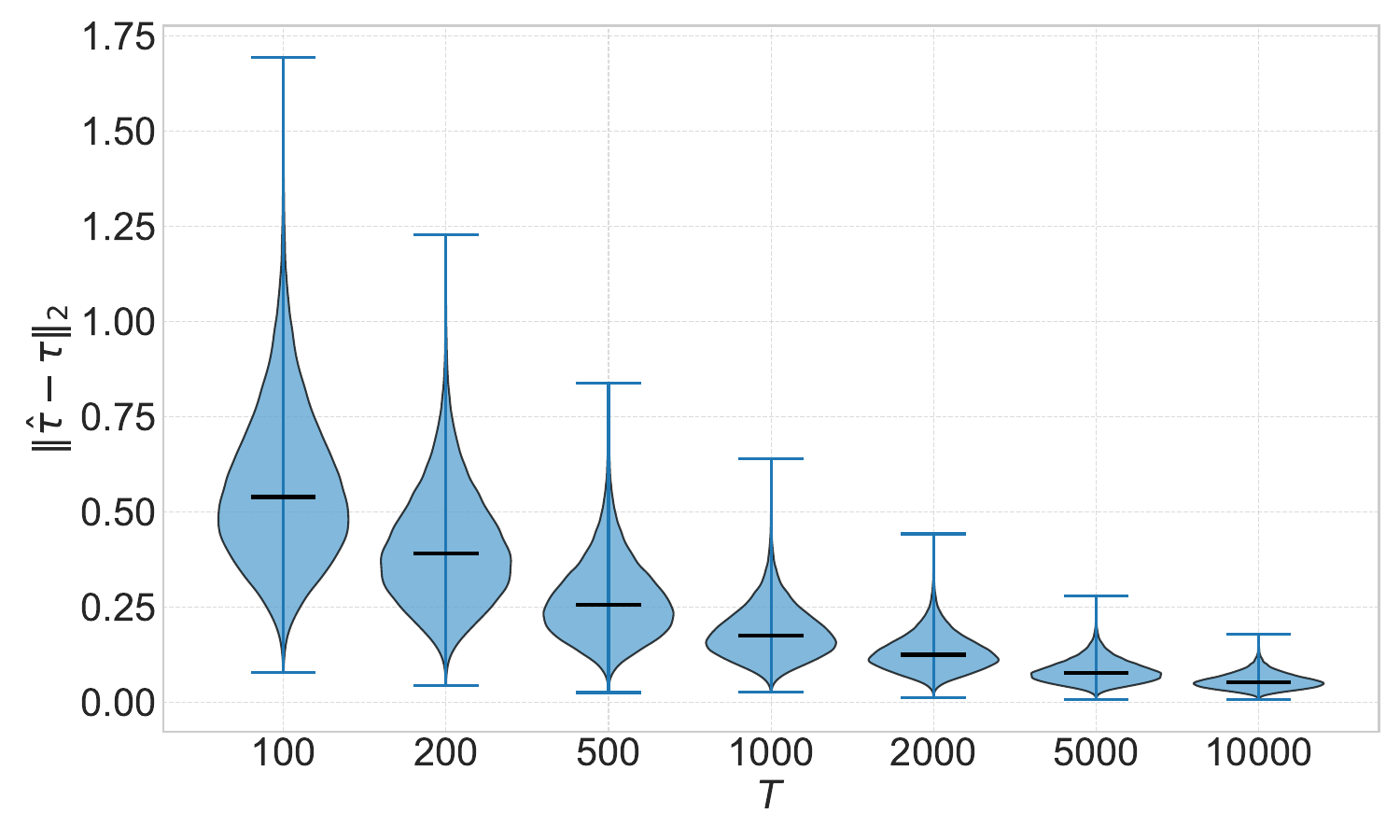}
  \caption{$\lVert \hat{\tau} - \tau \rVert_2$ versus $T$. The center bar is the median.}
  \label{fig2}
\end{figure}

Figure \ref{fig3} investigates the performance of the HAC variance estimator $\hat\mV$. The plot shows the normalized Frobenius error $\lVert \hat\mV - \mV \rVert_{\mathrm F}/\lVert \mV \rVert_{\mathrm F}$ for different kernel bandwidth choices $L \in \{0,1,5,\lfloor T^{1/4} \rfloor\}$, averaged across replications. When no autocorrelation adjustment is made ($L=0$), the estimator remains biased and does not improve with $T$. By contrast, increasing the bandwidth significantly reduces error, with $L = T^{1/4}$ yielding the best finite-sample performance and approaching zero error as $T$ grows. This confirms our theoretical variance results (Theorem \ref{thm:var}) and demonstrates the importance of selecting an appropriate increasing bandwidth.

\begin{figure}[htbp]
  \centering
  \includegraphics[width=0.8\linewidth]{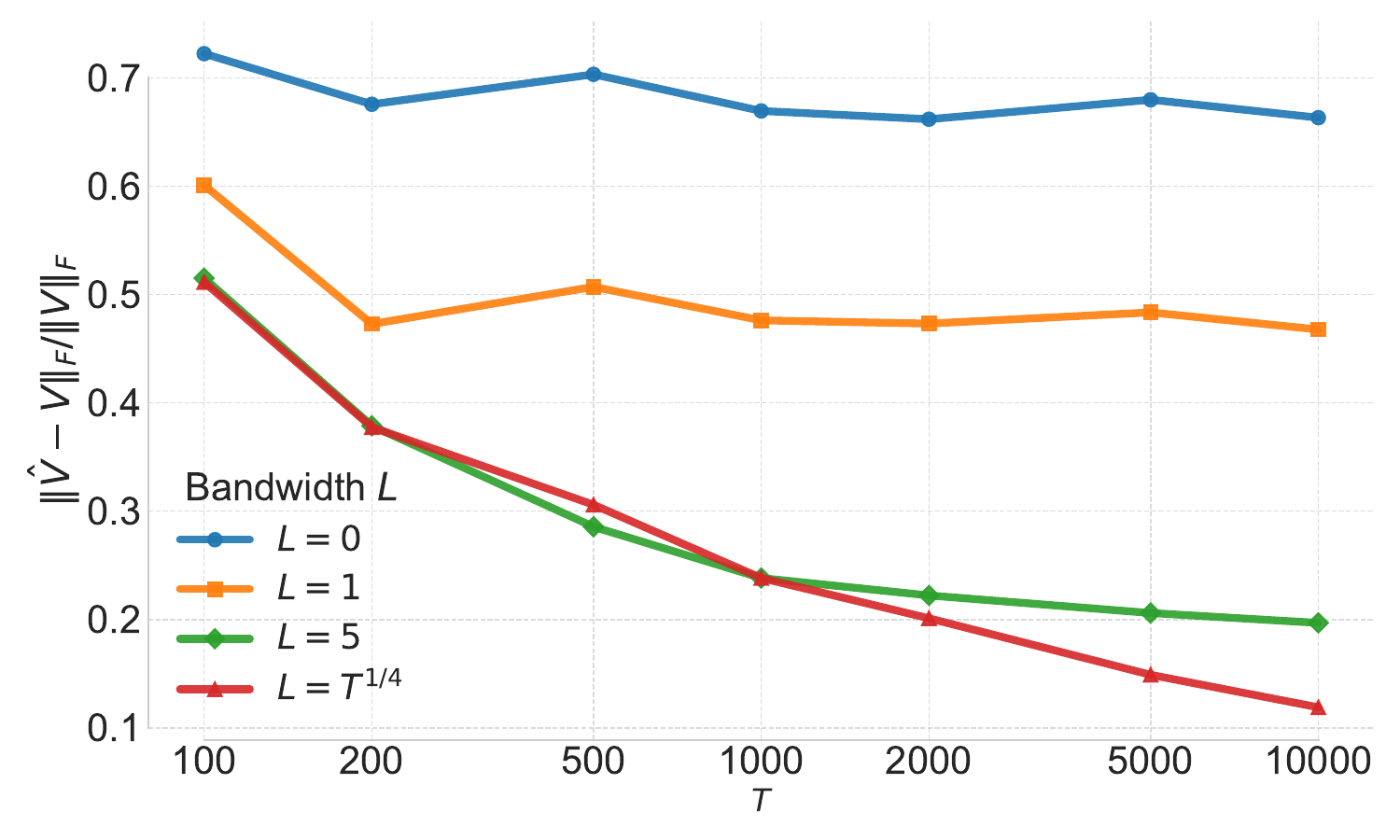}
  \caption{Normalized Frobenius error \(\lVert \hat\mV-\mV \rVert_{\mathrm F}/\lVert \mV \rVert_{\mathrm F}\) versus $T$ for different bandwidths $L$.}
  \label{fig3}
\end{figure}

Table \ref{tab4} reports the empirical coverage probabilities of nominal 95\% confidence intervals for $\hat\tau_k$, constructed using the HAC variance estimator $\hat\mV$ with bandwidth $L = T^{1/4}$. Across all lags $k=0,\ldots,5$ and all sample sizes, the coverage probabilities remain very close to the nominal level, even for relatively small $T=100$. As $T$ increases, the coverage stabilizes tightly around 0.95, confirming that the asymptotic normal approximation in Theorem \ref{thm:clt} provides accurate inference in finite samples.

\begin{table}[h]
  \centering
  \caption{Empirical 95\% coverage for $\hat\tau_k$ using $\hat\mV$ with $L=T^{1/4}$.}
  \begin{tabular}{ccccccc}
  \hline
  $T$ & $\tau_0$ & $\tau_1$ & $\tau_2$ & $\tau_3$ & $\tau_4$ & $\tau_5$ \\
  \hline
  100 & 0.934 & 0.935 & 0.932 & 0.933 & 0.934 & 0.934 \\
  200 & 0.944 & 0.943 & 0.940 & 0.943 & 0.942 & 0.943 \\
  500 & 0.943 & 0.946 & 0.947 & 0.945 & 0.948 & 0.945 \\
  1000 & 0.949 & 0.947 & 0.946 & 0.947 & 0.946 & 0.950 \\
  2000 & 0.949 & 0.949 & 0.948 & 0.949 & 0.950 & 0.950 \\
  5000 & 0.949 & 0.949 & 0.950 & 0.949 & 0.951 & 0.951 \\
  10000 & 0.950 & 0.950 & 0.950 & 0.949 & 0.950 & 0.950 \\
  \hline
  \end{tabular}
  \label{tab4}
\end{table}

\subsection{Real data analysis}

We analyze the trading experiment described in \cite{bojinov2019time}, which was conducted by AHL Partners, a quantitative trading firm specializing in financial futures. The firm executes orders in one of two ways: using a human trader or using an algorithmic trader. For the experiment, AHL restricted attention to orders whose sizes lay between fixed lower and upper bounds (time-invariant within each market, but varying across markets). These eligible orders were randomly assigned to either a human trader or an algorithmic trader. Therefore, treatments were assigned independently across time, following an i.i.d. Bernoulli design, ensuring that the setup satisfied our no-anticipation assumption. AHL provided 10 independent datasets from 2016, each corresponding to a different equity index futures market spanning the US, Europe, and Asia. We analyze these markets separately and regard each as a distinct experimental unit. To protect confidentiality, AHL did not reveal which of methods ``A'' and ``B'' corresponds to human versus algorithmic traders, nor did they disclose the identities of the individual markets. For consistency, we label A as the control group and B as the treatment group. The outcome of interests $Y$ is the slippage, which is a simple function of the trade prices and volumes at which an order was executed. Lower slippage indicates better execution quality and thus serves as a proxy for higher trading performance.

With $K=4$, Table \ref{tab1} reports, for each market, the sample length, mean outcomes by group (A = control, B = treatment), the joint p-value for $H_0:\tau=\mathbf{0}$, and the per-lag estimates $\hat\tau_k$ with their corresponding p-values. We find clear heterogeneity across markets: the null is rejected in Markets 4, 6, 7, and 9, indicating the presence of dynamic treatment effects, while other markets show no joint evidence of impact. In both Market 6 ($\hat\tau_0=5.65,\,p_0=0.00$) and Market 9 ($\hat\tau_0=2.43,\,p_0=0.00$), ``A'' (control) performed significantly better than ``B'' (treatment), and there is no market where ``B'' performed significantly better than ``A'', consistent with the findings of \cite{bojinov2019time}. Market 4 shows lagged effects at longer lags ($\hat\tau_2=4.89,p_2=0.05;\ \hat\tau_4=6.69,p_4=0.02$), whereas Market 7 shows a negative lag-2 effect ($\hat\tau_2=-3.77,p_2=0.00$). Some isolated single-lag signals (e.g., Market 5’s positive lag-1 estimate) do not lead to a joint rejection.

\begin{table}[htbp]
  \centering
  \caption{Trading experiment results.}
  \scriptsize
\begin{tabular}{rrrrrrrrrrrrrrr}
  \toprule
  Market & Length & A & B & $p$ & $\hat{\tau}_0$ & $p_0$ & $\hat{\tau}_1$ & $p_1$ & $\hat{\tau}_2$ & $p_2$ & $\hat{\tau}_3$ & $p_3$ & $\hat{\tau}_4$ & $p_4$ \\
  \midrule
  1 & 252 & 1.48 & 2.29 & 0.87 & 0.79 & 0.41 & 0.36 & 0.73 & 0.25 & 0.78 & 0.19 & 0.86 & 0.90 & 0.38 \\
  2 & 149 & 4.11 & 3.35 & 0.90 & -1.25 & 0.46 & 0.39 & 0.81 & -0.83 & 0.68 & -1.98 & 0.30 & -0.23 & 0.90 \\
  3 & 376 & -0.05 & 1.07 & 0.89 & 1.36 & 0.28 & -1.11 & 0.46 & 0.46 & 0.72 & 0.68 & 0.65 & -0.98 & 0.48 \\
  4 & 78 & 3.38 & 3.23 & 0.03 & 0.88 & 0.72 & 2.92 & 0.23 & 4.89 & 0.05 & 3.56 & 0.14 & 6.69 & 0.02 \\
  5 & 103 & 0.57 & 0.63 & 0.21 & -0.80 & 0.72 & 5.88 & 0.03 & -3.99 & 0.11 & -1.36 & 0.64 & -0.99 & 0.64 \\
  6 & 80 & -1.48 & 3.72 & 0.01 & 5.65 & 0.00 & -1.50 & 0.49 & 0.04 & 0.99 & 1.18 & 0.51 & 0.34 & 0.86 \\
  7 & 247 & 1.99 & 1.64 & 0.00 & -0.43 & 0.74 & -2.20 & 0.08 & -3.77 & 0.00 & 2.44 & 0.07 & -1.31 & 0.27 \\
  8 & 143 & -0.08 & -0.07 & 0.70 & -0.29 & 0.91 & -0.97 & 0.68 & -2.63 & 0.32 & -1.72 & 0.50 & 0.56 & 0.81 \\
  9 & 417 & -2.19 & 0.64 & 0.00 & 2.43 & 0.00 & 0.09 & 0.91 & 1.30 & 0.08 & -0.70 & 0.40 & 0.09 & 0.90 \\
  10 & 426 & 0.80 & 2.10 & 0.11 & 1.56 & 0.13 & 1.67 & 0.16 & 1.32 & 0.25 & -1.58 & 0.19 & -1.66 & 0.15 \\
  \bottomrule
\end{tabular}
  \label{tab1}
\end{table}

\section{Discussion}\label{sec:discussion}

This paper develops a unified design-based framework for regression-based estimators in time-series experiments, establishing consistency, asymptotic normality, and conservative variance estimation without requiring correct model specification. We also extend our framework to the continuous treatment. In this section, we discuss the extensions we have done and some future directions. 

In Appendix~\ref{sec:interaction}, we extend our framework to regression with interaction terms. Beyond estimating the main effects of treatments at different lags, it is often of interest to investigate interaction effects between treatments administered at distinct time periods. Including interaction terms is particularly useful in settings where the effect of a treatment depends on its recent history or on combinations of past interventions. In online platforms, repeated exposures (e.g., showing a user the same advertisement or recommendation on consecutive days) may produce reinforcement or fatigue, leading to nonlinear cumulative effects. In clinical and behavioral studies, the efficacy of an intervention (e.g., medication, exercise, or reminders) often depends on treatment consistency, with consecutive applications generating biological accumulation or tolerance effects.

In Appendix~\ref{sec:general}, we extend our framework to general exposure mapping. So far, we have proposed working models for the treatment effects through individual lags or interactions of consecutive lags. While this approach captures direct and pairwise temporal dependencies, it may be restrictive in settings where outcomes depend on more general features of the treatment path. To address this, we consider exposure  mappings that summarize treatment assignments over multiple periods into lower-dimensional summary statistics. Specifically, we divide the past treatment assignments into several time blocks, and the treatments within each block are aggregated into a single exposure mapping. For example, in behavioral or clinical studies, treatments (e.g., reminders or interventions) often have cumulative effects that depend on how consistently the treatment has been received in the past rather than on any single recent treatment. This allows estimation of short-term, medium-term and long-term persistence in treatment effects.

We now discuss some future directions. First, while our focus has been on the single-unit, time-series setting, a natural next step is to consider panel experiments \citep{bojinov2021panel}, where multiple units are observed over time and treatments are repeatedly randomized both across units and across time. Panel experiments introduce additional layers of complexity beyond single-unit time-series experiments. Interference may occur both across time within each unit and across units at a given time, leading to a rich class of potential outcomes. Researchers may be interested not only in lagged effects within units, but also in spillover effects across units and dynamic causal effects aggregated across the population. Recent work has begun to formalize and analyze these designs \citep{bojinov2021panel,ni2023design,han2024population}. Our results provide a foundation for extending design-based regression analysis beyond single-unit time-series experiments. The working model we specify, together with HAC variance estimation, can be naturally generalized to panel settings. Such an extension would allow for the joint estimation of both lagged effects within units and spillover effects across units, thereby offering a unified design-based framework for dynamic causal inference in panel experiments.

Second, while our analysis assumes that treatment assignments are independent across time, many practical experiments, particularly in online platforms and sequential decision-making contexts, employ adaptive designs, where treatment assignment probabilities depend on past outcomes or treatment history. Our framework naturally extends to such adaptive experimental settings. The key observation is that one can normalize the treatment indicators conditionally on the past history instead of unconditionally. With this normalization, the same regression-based procedure and design-based analysis carry through, and asymptotic properties can be established using martingale theory.

Third, the regression-based procedure allows for incorporating lagged outcomes and additional covariates to improve efficiency. This type of regression adjustment parallels the use of covariate adjustment in classical randomized experiments. In the time-series experiments, one can augment the regression with covariates that are unaffected by the treatment of interest. Incorporating lagged outcomes is particularly relevant in time series literature, as it aligns with the structure of autoregressive models and helps account for temporal dependence in the outcome. By absorbing predictable variation, such adjustments improve efficiency without introducing bias to the treatment effect estimation.

\section*{Acknowledgment}

We thank Iavor Bojinov for providing the data and Iavor Bojinov, Avi Feller and participants of the 2025 NBER-NSF Time Series Conference for helpful comments. Lin was partially supported by  the Two Sigma PhD Fellowship. Ding was supported by the U.S. National Science Foundation (1945136, 2514234).

{
\bibliographystyle{apalike}
\bibliography{AMS}
}

\newpage

\appendix

\section*{Supplementary materials for ``Unifying regression-based and design-based causal inference in
time-series experiments''}

In Appendix~\ref{sec:interaction}, we dicuss the extension to the regression with interaction terms. 

In Appendix~\ref{sec:general}, we discuss the extension to the general exposure mapping. 

In Appendix~\ref{sec:ols_wls}, we provide more details on the relationship between full OLS, marginal OLS and WLS. 

In Appendix~\ref{sec:lemmas}, we provide the lemmas for the main results and their proofs. 

In Appendix~\ref{sec:proofs}, we provide the proofs of all the results.

{\bf Additional notation.} For any $a,b \in \bR$, write $a \vee b = \max\{a,b\}$ and $a \wedge b = \min\{a,b\}$. For any two real sequences $\{a_n\}$and $\{b_n\}$, write $a_n \lesssim b_n$ (or equivalently, $b_n \gtrsim a_n$) if $a_n = O(b_n)$.

\section{Regression with interaction terms}\label{sec:interaction}

In practice, we may be interested in studying the interaction effects between treatments at different lags. An additional advantage of the regression-based analysis is its flexibility to incorporate interaction terms. For instance, one can include regressors such as $\tZ_{t-k}\tZ_{t-k'}$ to study whether the joint presence of treatments at two different lags produces an effect beyond the sum of their individual contributions. To accommodate this possibility, we extend the OLS specification by including interaction terms between consecutive lags:
\begin{align*}
  \textbf{lm}(Y_t \sim \tZ_t + \cdots + \tZ_{t-K} + \tZ_t \tZ_{t-1} + \cdots + \tZ_{t-K+1}\tZ_{t-K}).
\end{align*}
We denote the corresponding estimator by $\tilde\tau_{\rm int} = (\tilde\tau_{{\rm int},0},\ldots,\tilde\tau_{{\rm int},K}, \tilde\tau_{{\rm int},0,1},\ldots,\tilde\tau_{{\rm int},K-1,K})^\top \in \bR^{2K+1}$. The extended regression specification includes $K+1$ main-effect terms corresponding to individual lags and $K$ interaction terms between consecutive lags. In total, the model has $2K+1$ regressors. This framework can be further generalized to incorporate higher-order interactions by augmenting the regression with corresponding interaction terms, similar to regression with interactions in classic factorial experiments \citep{zhao2023covariate}. To illustrate the key idea in time-series experiments, we focus here on two-way consecutive interactions. 

We recall the normalized treatment variable $\tZ_t$ defined in \eqref{eq:z_t} and the main effect weights $w_k$ defined in \eqref{eq:w_k}. We set $\hat\tau_{{\rm int},k} = w_k^{-1} \tilde\tau_{{\rm int},k}$ as \eqref{eq:hat_tau_k}. For the interaction terms, we introduce analogous weights that account for the joint variation in the two treatment lags. Specifically, we define
\begin{align*}
  w_{k-1,k} =& \left[\frac{1}{T-K} \sum_{t=K+1}^T (\Var[Z_{t-k}]\Var[Z_{t-k+1}])^{-1}\right]^{-1}  \\
  =& \left[\frac{1}{T-K} \sum_{t=K+1}^T [p_{t-k}p_{t-k+1}(1-p_{t-k})(1-p_{t-k+1})]^{-1}\right]^{-1},
\end{align*}
and define $\hat\tau_{{\rm int},k-1,k} = w_{k-1,k}^{-1} \tilde\tau_{{\rm int},k-1,k}$. Collecting the main and interaction effects, we obtain the transformed OLS estimator $\hat\tau_{\rm int} = (\hat\tau_{{\rm int},0},\ldots,\hat\tau_{{\rm int},K}, \hat\tau_{{\rm int},0,1},\ldots,\hat\tau_{{\rm int},K-1,K})^\top \in \bR^{2K+1}$.
The causal estimand is $\tau_{\rm int} = (\tau_0,\ldots,\tau_K,\tau_{0,1},\ldots,\tau_{K-1,K})^\top$, which jointly captures both lagged main effects and interaction effects. The lag-$k$ main effect $\tau_k$ is defined as \eqref{eq:tau}. The interaction effect between lag $k-1$ and lag $k$ is given by
\begin{align*}
  \tau_{k-1,k} = \frac{1}{T-K} \sum_{t=K+1}^T \E[&Y_t(Z_{t:t-k+2},1,1,Z_{t-k-1:1}) - Y_t(Z_{t:t-k+2},1,0,Z_{t-k-1:1})\\
  &-Y_t(Z_{t:t-k+2},0,1,Z_{t-k-1:1}) + Y_t(Z_{t:t-k+2},0,0,Z_{t-k-1:1})],
\end{align*}
where the expectation is taken over the joint distribution of $(Z_{t:t-k+2}, Z_{t-k-1:1})$. Intuitively, $\tau_{k-1,k}$ measures the incremental effect of applying treatment at both lags $k-1$ and lag $k$ simultaneously, beyond the sum of their individual effects.

To write $\hat\tau_{\rm int}$ explicitly, recall the outcome vector $\mY_K = (Y_{K+1},\ldots,Y_T)^\top \in \bR^{T-K}$ and define the regressor matrix $\mZ_{{\rm int},K} = (\tZ_{{\rm int},K+1:1}^\top,\ldots,\tZ_{{\rm int},T:T-K}^\top)^\top \in \bR^{(T-K) \times (2K+1)}$ where $\tZ_{{\rm int},t:t'} = (\tZ_t,\tZ_{t-1},\ldots,\tZ_{t'}, \tZ_t \tZ_{t-1},\ldots,\tZ_{t'+1} \tZ_{t'})^\top$ for $t>t'$.  Let $\mW_{{\rm int},K} \in \bR^{(2K+1)\times(2K+1)}$ be the diagonal weight matrix $\mW_{{\rm int},K} = {\rm diag}(w_k:k=0,\ldots,K, w_{k-1,k}:k=1,\ldots,K)$. Then, $\hat\tau_{\rm int}$ can be expressed as
\begin{align*}
  \hat\tau_{\rm int} = \mW_{{\rm int},K}^{-1} \tilde\tau_{{\rm int}} = \mW_{{\rm int},K}^{-1} (\mZ_{{\rm int},K}^\top \mZ_{{\rm int},K})^{-1} \mZ_{{\rm int},K}^\top \mY_K.
\end{align*}
Analogous to the decomposition of $\hat\tau_k$ in \eqref{eq:hat_tau}, we can decompose $\hat\tau_{\rm int}$ as
\[
  \hat{\tau}_{\rm int} = \tau + \mW_{{\rm int},K}^{-1} [(T-K)^{-1}\mZ_{{\rm int},K}^\top \mZ_{{\rm int},K}]^{-1} [(T-K)^{-1}\mZ_{{\rm int},K}^\top (\mY_K - \mZ_{{\rm int},K} \mW_{{\rm int},K} \tau)].
\]

To study the asymptotic properties of the OLS estimator $\hat\tau_{\rm int}$, analogous to \eqref{eq:residual}, we define the oracle residual of OLS regression at time $t$ as $$U_{{\rm int}, t} = Y_t - \sum_{k=0}^K \tZ_{t-k} w_k \tau_k - \sum_{k=1}^K \tZ_{t-k} \tZ_{t-k+1} w_{k-1,k} \tau_{k-1,k}.$$ We then define the asymptotic variance matrix $\mV_{\rm int} \in \bR^{(2K+1) \times (2K+1)}$ analogous to \eqref{eq:mV} as
\begin{align*}
  &[\mV_{\rm int}]_{k, k'} = \frac{1}{T-K}\Cov\Big[\sum_{t=K+1}^T \tZ_{t-k+1} U_{{\rm int}, t}, \sum_{t=K+1}^T \tZ_{t-k'+1} U_{{\rm int}, t}\Big],\\
  & [\mV_{\rm int}]_{K+1+k, K+1+k'} = \frac{1}{T-K}\Cov\Big[\sum_{t=K+1}^T \tZ_{t-k} \tZ_{t-k+1} U_{{\rm int}, t}, \sum_{t=K+1}^T \tZ_{t-k'}\tZ_{t-k'+1} U_{{\rm int}, t}\Big],\\
  & [\mV_{\rm int}]_{K+1+k, k'} = [\mV_{\rm int}]_{k', K+1+k} = \frac{1}{T-K}\Cov\Big[\sum_{t=K+1}^T \tZ_{t-k} \tZ_{t-k+1} U_{{\rm int}, t}, \sum_{t=K+1}^T \tZ_{t-k'+1} U_{{\rm int}, t}\Big].
\end{align*}

The design assumptions and treatment probability conditions carry over unchanged from the main effects model. Under Assumption~\ref{asp:moment}, the uniformly boundedness of $\lVert \tau_{\rm int} \rVert_1$ across $T$ follows immediately, ensuring that both main and interaction effects remain well-behaved. To guarantee invertibility of the asymptotic covariance, we impose the following eigenvalue condition.

\begin{assumption}[Eigenvalue condition]\label{asp:eigen_int}
  The smallest eigenvalue of $\mV_{\rm int}$ is uniformly bounded away from zero for all $T$, i.e., $\inf_{T \ge 1} \lambda_{\min}(\mV_{\rm int}) > 0$.
\end{assumption}

With these assumptions in place, we obtain the consistency and asymptotic normality of the extended estimator $\hat\tau_{\rm int}$. These results parallel Theorem~\ref{thm:consistency} and Theorem~\ref{thm:clt}, now incorporating both main and interaction effects.

\begin{theorem}[Consistency]\label{thm:consistency_int}
  Assume Assumption~\ref{asp:no-anticipation}, Assumption~\ref{asp:design}, Assumption~\ref{asp:ps}, Assumption~\ref{asp:ani}\ref{asp:ani:i} with $K \sum_{k=1}^T \theta_{T,k}/T \to 0$ or Assumption~\ref{asp:mdep}\ref{asp:mdep:i} with $K m/T \to 0$, Assumption~\ref{asp:moment}, and $K^2/T \to 0$. Then we have
  \begin{align*}
    \lVert \hat\tau_{\rm int} - \tau_{\rm int}\rVert_2 \stackrel{\sf p}{\longrightarrow} 0.
  \end{align*}
\end{theorem}

\begin{theorem}[Asymptotic normality]\label{thm:clt_int}
  Assume Assumption~\ref{asp:no-anticipation}, Assumption~\ref{asp:design}, Assumption~\ref{asp:ps}, Assumption~\ref{asp:ani}\ref{asp:ani:ii} with $K$ fixed or Assumption~\ref{asp:mdep}\ref{asp:mdep:ii} with $K^4/T \to 0$, Assumption~\ref{asp:moment}, and Assumption~\ref{asp:eigen_int}. Then for any $\lambda_K \in \bR^{2K+1}$ satisfying $0 < \inf_{T \ge 1} (\lVert \lambda_K \rVert_2/\lVert \lambda_K \rVert_1) \le \sup_{T \ge 1} (\lVert \lambda_K \rVert_2/\lVert \lambda_K \rVert_1) < \infty$, we have
  \begin{align*}
    \sqrt{T-K} (\lambda_K^\top \mV_{\rm int} \lambda_K)^{-1/2} \lambda_K^\top (\hat\tau_{\rm int} - \tau_{\rm int}) \stackrel{\sf d}{\longrightarrow} N(0,1).
  \end{align*}
\end{theorem}

To estimate the asymptotic variance $\mV_{\rm int}$ in Theorem~\ref{thm:clt_int}, we apply the HAC variance estimator introduced in Section~\ref{sec:hac}. Let $$\hat U_{{\rm int}, t} = Y_t - \sum_{k=0}^K \tZ_{t-k} w_k \hat\tau_{{\rm int},k} - \sum_{k=1}^K \tZ_{t-k} \tZ_{t-k+1} w_{k-1,k} \hat\tau_{{\rm int},k-1,k}$$ be the sample residual at time $t$ analogous to \eqref{eq:residual_sample}. Using these residuals, we can then construct the HAC variance estimator in the same way as in \eqref{eq:hac_var}, and the consistency result is parallel to Theorem~\ref{thm:var}, since the inclusion of interaction terms simply augments the number of regressors from $K+1$ to $2K+1$.

We note that the $k$-th element of both $\tau$ and $\tau_{\rm int}$ coincides with $\tau_k$. This raises a natural question: do we actually need to include the interaction terms? The answer closely parallels the comparison between full OLS and marginal OLS in Section~\ref{sec:full_marginal}: both procedures target the same estimand, but differ in their asymptotic variances and variance estimators.

By Theorems~\ref{thm:clt} and \ref{thm:clt_int}, the asymptotic variance of $\hat\tau_k$ without interaction terms is $(T-K)^{-1} \Var[\sum_{t=K+1}^T \tZ_{t-k} (Y_t - \sum_{\ell=0}^K \tZ_{t-\ell} w_\ell \tau_\ell)]$, while the asymptotic variance with interaction terms is $(T-K)^{-1} \Var[\sum_{t=K+1}^T \tZ_{t-k} (Y_t - \sum_{\ell=0}^K \tZ_{t-\ell} w_\ell \tau_\ell - \sum_{\ell=1}^K \tZ_{t-\ell} \tZ_{t-\ell+1} w_{\ell-1,\ell} \tau_{\ell-1,\ell})]$. To compare these variances explicitly, we consider a linear potential outcome model with interaction terms.

\begin{proposition}\label{prop:int}
  Consider a linear model with homogeneous effects and nonrandom error: $Y_t = Y_t(z_{t:1})=\sum_{k=0}^{t-1} \beta_k z_{t-k} + \sum_{k=1}^{t-1} \beta_{k-1,k} z_{t-k} z_{t-k+1} + \epsilon_t$. Assume a constant treatment probability $p_t=p$. Then, the asymptotic variance for OLS with interaction terms is
  \begin{align*}
    &\frac{1}{T-K} \Var\Big[\sum_{t=K+1}^T \tZ_{t-k} (Y_t - \sum_{\ell=0}^K \tZ_{t-\ell} w_\ell \tau_\ell - \sum_{\ell=1}^K \tZ_{t-\ell} \tZ_{t-\ell+1} w_{\ell-1,\ell} \tau_{\ell-1,\ell})\Big] \\
    =& \frac{1}{T-K} \sum_{t=K+1}^T \sum_{\ell=K+1}^{t-1} (\beta_\ell^2+ p^2 \beta_{\ell-1,\ell}^2 + p^2 \beta_{\ell,\ell+1}^2) + \frac{1}{T-K} \sum_{t=K+1}^T \sum_{\ell=K+1}^{t-1} p(1-p)\beta_{\ell-1,\ell}^2\\
    &+ \frac{1}{T-K} \frac{1}{p(1-p)}\sum_{t=K+1}^T \Big(p\sum_{k=0}^{t-1} \beta_k + p^2 \sum_{k=1}^{t-1} \beta_{k-1,k} + \epsilon_t \Big)^2,
  \end{align*}
  while the asymptotic variance for OLS without interaction terms is
  \begin{align*}
    &\frac{1}{T-K} \Var\Big[\sum_{t=K+1}^T \tZ_{t-k} (Y_t - \sum_{\ell=0}^K \tZ_{t-\ell} w_\ell \tau_\ell)\Big] \\
    =& \frac{1}{T-K} \sum_{t=K+1}^T \sum_{\ell=K+1}^{t-1} (\beta_\ell^2+ p^2 \beta_{\ell-1,\ell}^2 + p^2 \beta_{\ell,\ell+1}^2) + \frac{1}{T-K} \sum_{t=K+1}^T \sum_{\ell=1}^{t-1} p(1-p)\beta_{\ell-1,\ell}^2\\
    &+ \frac{1}{T-K} \frac{1}{p(1-p)}\sum_{t=K+1}^T \Big(p\sum_{k=0}^{t-1} \beta_k + p^2 \sum_{k=1}^{t-1} \beta_{k-1,k} + \epsilon_t \Big)^2.
  \end{align*}
\end{proposition}

We face a natural trade-off when deciding whether to include interaction terms in the regression. Even if the primary goal is to estimate the main effects, Proposition~\ref{prop:int} shows that incorporating interaction terms can reduce asymptotic variance: the gain is exactly $(T-K)^{-1} \sum_{t=K+1}^T \sum_{\ell=1}^{K} p(1-p)\beta_{\ell-1,\ell}^2 = \sum_{\ell=1}^{K} p(1-p)\beta_{\ell-1,\ell}^2$, which is strictly positive whenever the interaction effects $\beta_{\ell-1,\ell}$ are non-negligible. Thus, including interactions improves asymptotic efficiency in the presence of interaction effects. However, this benefit comes at a cost. When the true interactions are small or absent, the variance reduction is negligible, while the inclusion of additional regressors increases model complexity and inflates finite-sample variance. Therefore, this creates a trade-off: including interactions can improve asymptotic efficiency but may sacrifice finite sample performance.

\section{General exposure mapping}\label{sec:general}

In practice, we may be interested in using the exposure mappings to summarize the treatment history. Define a collection of exposure mappings: $$g_{t,0} = g_0(Z_{t:t-k_0}), \quad g_{t,1} = g_1(Z_{t-k_0-1:t-k_1}), \ldots, \quad g_{t,S} = g_S(Z_{t-k_{S-1}-1:t-k_S}),$$ where $S \ge 0$ indexes the exposure mappings. The exposure lags satisfy $0 \le k_0 < k_1 < \ldots < k_S = K$, and each function $g_s$ is a known mapping from a block of treatment assignments to $\{0,1\}$: $g_0: \bR^{k_0+1} \to \{0,1\}, g_1: \bR^{k_1-k_0} \to \{0,1\}, \ldots, g_S: \bR^{k_S-k_{S-1}} \to \{0,1\}$. Intuitively, we divide the treatment history into $S+1$ consecutive time blocks, and within each block, the sequence of past treatments is summarized by a binary value from the exposure mapping. We assume that the potential outcome can be expressed in terms of these exposure mappings: $Y_t(Z_{t:1}) = Y_t(g_{t,0},g_{t,1},\ldots,g_{t,S})$. 
We treat these exposure mappings as treatment indicators in Section~\ref{sec:setup}. For each exposure mapping $g_{t,s}$, we define a normalized version $\tilde{g}_{t,s} = (g_{t,s} - \E[g_{t,s}])/\Var[g_{t,s}]$ and corresponding weight 
\[
  w_{{\rm g},s} = \left[\frac{1}{T-K} \sum_{t=K+1}^T \left(\Var[g_{t,s}]\right)^{-1}\right]^{-1},
\]
for $s=0,1,\ldots,S$. We consider the following OLS regression:
\[
  \textbf{lm}(Y_t \sim \tilde{g}_{t,0} + \tilde{g}_{t,1} + \ldots + \tilde{g}_{t,S}).
\]
We denote the corresponding estimator by $\tilde\tau_{\rm g} = (\tilde \tau_{{\rm g},0},\ldots,\tilde \tau_{{\rm g},S})$. We define $\hat\tau_{{\rm g},s} = w_{{\rm g},s}^{-1} \tilde\tau_{{\rm g},s}$ for $s=0,1,\ldots,S$ analogous to \eqref{eq:hat_tau_k}, and obtain the transformed OLS estimator $\hat\tau_{\rm g} = (\hat\tau_{{\rm g},0},\ldots,\hat\tau_{{\rm g},S})$. The causal estimand is $\tau_{\rm g} = (\tau_{{\rm g},0},\ldots,\tau_{{\rm g},S})$, where each component $\tau_{{\rm g}, s}$  captures the causal effect of the $s$-th exposure mapping. Formally,
\begin{align*}
  \tau_{{\rm g}, s} = \frac{1}{T-K} \sum_{t=K+1}^T \E[&Y_t(g_{t,0},\ldots,g_{t,s-1},1,g_{t,s+1},\ldots,g_{t,S})- Y_t(g_{t,0},\ldots,g_{t,s-1},0,g_{t,s+1},\ldots,g_{t,S})],
\end{align*}
where the expectation averages over $(g_{t,0},\ldots,g_{t,s-1},g_{t,s+1},\ldots,g_{t,S})$, holding the exposure mapping $g_{t,s}$ fixed at either 1 or 0.

This exposure mapping framework unifies and generalizes many commonly studied models. If we take $g_0(z) = \cdots = g_S(z) = z$, and let $k_s = s$ with $S=K$, then each exposure mapping corresponds to a single treatment indicator at a lag. In this case, the model reduces to the OLS regression introduced in Section~\ref{sec:setup}, where each treatment indicator enters separately as a regressor. If instead we set $S=0$, $k_0 = K$, and define $g_0(z) = g(z)$ for some function $g : \{0,1\}^K \to \{0,1\}$, then the potential outcomes depend only on a single transformation of the treatment history over the most recent $K$ periods. This corresponds to models in which treatment effects are summarized by a single exposure mapping.

To write $\hat\tau_{\rm g}$ explicitly, recall the outcome vector $\mY_K = (Y_{K+1},\ldots,Y_T)^\top \in \bR^{T-K}$ and define the regressor matrix $\mZ_{{\rm g},K} = (\tZ_{{\rm g},K+1}^\top,\ldots,\tZ_{{\rm g},T}^\top)^\top \in \bR^{(T-K) \times (S+1)}$ where $\tZ_{{\rm g},t} = (g_{t,0},\ldots,g_{t,S})^\top$ for $t=K+1,\ldots,T$.  Let $\mW_{{\rm g},K} \in \bR^{(S+1)\times(S+1)}$ be the diagonal weight matrix $\mW_{{\rm g},K} = {\rm diag}(w_{{\rm g},s}:s=0,\ldots,S)$. Then, $\hat\tau_{\rm g}$ can be expressed as
\[
  \hat{\tau} = \mW_{{\rm g},K}^{-1} \tilde\tau_{{\rm g}} = \mW_{{\rm g},K}^{-1} (\mZ_{{\rm g},K}^\top \mZ_{{\rm g},K})^{-1} \mZ_{{\rm g},K}^\top \mY_K.
\]
Analogous to the decomposition of $\hat\tau_k$ in \eqref{eq:hat_tau}, we can decompose $\hat\tau_{\rm g}$ as
\[
  \hat{\tau}_{\rm g} = \tau + \mW_{{\rm g},K}^{-1} [(T-K)^{-1}\mZ_{{\rm g},K}^\top \mZ_{{\rm g},K}]^{-1} [(T-K)^{-1}\mZ_{{\rm g},K}^\top (\mY_K - \mZ_{{\rm g},K} \mW_{{\rm g},K} \tau)].
\]

To study the asymptotic properties of the OLS estimator $\hat\tau_{\rm g}$, analogous to \eqref{eq:residual}, we define the oracle residual of OLS regression at time $t$ as $$U_{{\rm g}, t} = Y_t - \sum_{s=0}^S w_{{\rm g},s} \tau_{{\rm g},s} \tilde{g}_{t,s}.$$ We then define the asymptotic variance matrix $\mV_{\rm g} \in \bR^{(S+1) \times (S+1)}$ analogous to \eqref{eq:mV} as
\begin{align*}
  [\mV_{\rm g}]_{s, s'} = \frac{1}{T-K}\Cov\Big[\sum_{t=K+1}^T g_{t,s} U_{{\rm g}, t}, \sum_{t=K+1}^T g_{t,s'} U_{{\rm g}, t} \Big].
\end{align*}

To ensure well-defined and non-degenerate exposure mappings, we impose the following assumption, which is the analogue of Assumption~\ref{asp:ps} in the exposure setting.

\begin{assumption}[Treatment distributions  condition]\label{asp:ps_general}
  There exists some constant $\epsilon>0$ such that $\epsilon \le \inf_{t \ge 1, s=0,1,\ldots,S} \Var[g_{t,s}]$.
\end{assumption}

We also impose a boundedness condition on potential outcomes and exposure effects, analogous to Assumption~\ref{asp:moment}.
\begin{assumption}[Moment condition]\label{asp:moment_general}
  The potential outcomes $Y_t(z_{t:1})$ are uniformly bounded for all $t$ and all $z_{t:1} \in \bR^t$, i.e., $\sup_{t \ge 1, z_{t:1} \in \bR^t} \lvert Y_t(z_{t:1}) \rvert < \infty$. Additionally, $\lVert \tau_{\rm g} \rVert_1$ is uniformly bounded for all $T$, i.e., $\sup_{T} \lVert \tau_{\rm g} \rVert_1 < \infty$.
\end{assumption}

To avoid degeneracy, we impose an eigenvalue condition on $\mV_{\rm g}$.

\begin{assumption}[Eigenvalue condition]\label{asp:eigen_general}
  The smallest eigenvalue of $\mV_{\rm g}$ is uniformly bounded away from zero for all $T$, i.e., $\inf_{T \ge 1} \lambda_{\min}(\mV_{\rm g}) > 0$.
\end{assumption}

With these assumptions, we obtain the consistency and asymptotic normality of the exposure-effect estimator $\hat\tau_{\rm g}$.

\begin{theorem}[Consistency]\label{thm:consistency_general}
  Assume Assumption~\ref{asp:no-anticipation}, Assumption~\ref{asp:design}, Assumption~\ref{asp:ps_general}, Assumption~\ref{asp:ani}\ref{asp:ani:i} with $K \sum_{k=1}^T \theta_{T,k}/T \to 0$ or Assumption~\ref{asp:mdep}\ref{asp:mdep:i} with $K m/T \to 0$, Assumption~\ref{asp:moment_general}, and $K^2/T \to 0$. Then we have
  \begin{align*}
    \lVert \hat\tau_{\rm g} - \tau_{\rm g} \rVert_2 \stackrel{\sf p}{\longrightarrow} 0.
  \end{align*}
\end{theorem}

\begin{theorem}[Asymptotic normality]\label{thm:clt_general}
  Assume Assumption~\ref{asp:no-anticipation}, Assumption~\ref{asp:design}, Assumption~\ref{asp:ps_general}, Assumption~\ref{asp:ani}\ref{asp:ani:ii} with $K$ fixed or Assumption~\ref{asp:mdep}\ref{asp:mdep:ii} with $K^4/T \to 0$, Assumption~\ref{asp:moment_general}, and Assumption~\ref{asp:eigen_general}. Then for any $\lambda_K \in \bR^{S+1}$ satisfying $0 < \inf_{T \ge 1} (\lVert \lambda_K \rVert_2/\lVert \lambda_K \rVert_1) \le \sup_{T \ge 1} (\lVert \lambda_K \rVert_2/\lVert \lambda_K \rVert_1) < \infty$, we have
  \begin{align*}
    \sqrt{T-K} (\lambda_K^\top \mV_{\rm g} \lambda_K)^{-1/2} \lambda_K^\top (\hat\tau_{\rm g} - \tau_{\rm g}) \stackrel{\sf d}{\longrightarrow} N(0,1).
  \end{align*}
\end{theorem}

To estimate the asymptotic variance $\mV_{\rm g}$ in Theorem~\ref{thm:clt_general}, we apply the HAC variance estimator introduced in Section~\ref{sec:hac}. Let $$\hat U_{{\rm g}, t} = Y_t - \sum_{s=0}^S w_{{\rm g},s} \hat\tau_{{\rm g},s} \tilde{g}_{t,s}$$ be the sample residual at time $t$ analogous to \eqref{eq:residual_sample}. Using these residuals, we can then construct the HAC variance estimator in the same way as in \eqref{eq:hac_var}, and the consistency result is parallel to Theorem~\ref{thm:var}, since the inclusion of exposure mappings simply changes the number of regressors from $K+1$ to $S+1$.


\section{More details on the relationship between full OLS, marginal OLS and WLS}\label{sec:ols_wls}

In this section, we provide more details on the relationship between full OLS, marginal OLS and WLS discussed in Section~\ref{sec:full_marginal}.

\subsection{Full OLS and marginal OLS}

By Theorem~\ref{thm:clt}, the full OLS estimator satisfies the following CLT:
\[
\sqrt{T-K} \mV_{k+1,k+1}^{-1/2} (\hat\tau_k - \tau_k) \stackrel{\sf d}{\longrightarrow} N(0, 1),
\]
where $\mV$ denotes the asymptotic variance matrix of the full OLS estimator defined in \eqref{eq:mV}. The corresponding asymptotic variance component for lag $k$ can be written as
\[
  \mV_{k+1,k+1} = \frac{1}{T-K} \Var\left[\sum_{t=K+1}^T \tZ_{t-k} U_t\right] = \frac{1}{T-K} \Var\left[\sum_{t=K+1}^T \tZ_{t-k} \left(Y_t - \sum_{\ell=0}^K \tZ_{t-\ell} w_\ell \tau_\ell\right)\right].
\]

Define the asymptotic variance of the marginal OLS estimator as
\[
V_{{\rm marginal},k} = \frac{1}{T-K} \Var\left[\sum_{t=K+1}^T \tZ_{t-k} (Y_t - \tZ_{t-k} w_k \tau_k)\right].
\]
Analogous to Theorem~\ref{thm:clt}, the marginal OLS estimator also satisfies the following CLT:
\[
\sqrt{T-K} V_{{\rm marginal},k}^{-1/2} (\hat\tau_k - \tau_k) \stackrel{\sf d}{\longrightarrow} N(0, 1).
\]

Both $\mV$ and $V_{{\rm marginal},k}$ can be estimated using the HAC variance estimator.
By Theorem~\ref{thm:var}, we have for the full OLS estimator:
\[
\lVert \hat\mV - \mV - \mB_K \rVert_\F \stackrel{\sf p}{\longrightarrow} 0.
\]

Let $b_k \in \bR^{T-K}$ be defined by
\[
[b_k]_{t-K} = \tau_{t,k} - \frac{w_k}{p_{t-k}(1-p_{t-k})} \tau_k, \quad t=K+1,\ldots,T.
\]
Then $[\mB_K]_{k+1,k+1} = b_k^\top \mQ_K b_k$, which implies for the full OLS estimator:
\[
\lvert \hat\mV_{k+1,k+1} - \mV_{k+1,k+1} - b_k^\top \mQ_K b_k \rvert \stackrel{\sf p}{\longrightarrow} 0.
\]

Similarly, let $\hat\mV_{{\rm marginal},k}$ denote the HAC variance estimator for the marginal OLS. Then for the marginal OLS estimator, following the proof of Theorem~\ref{thm:var}, we have
\[
\lvert \hat\mV_{{\rm marginal},k} - V_{{\rm marginal},k} - b_k^\top \mQ_K b_k \rvert \stackrel{\sf p}{\longrightarrow} 0,
\]
since by the definition of $b_k$, for $t=K+1,\ldots,T$,
\[
\E[\tZ_{t-k}(Y_t - \tZ_{t-k} w_k \tau_k)] = \tau_{t,k} - \frac{w_k}{p_{t-k}(1-p_{t-k})} \tau_k = [b_k]_{t-K}.
\]

Hence, both $\hat\mV$ and $\hat\mV_{{\rm marginal},k}$ share the same bias term $b_k^\top \mQ_K b_k$ when estimating $\mV_{k+1,k+1}$ and $V_{{\rm marginal},k}$, respectively. The comparison between full and marginal OLS variance estimators therefore reduces to their asymptotic true variances.

\subsection{Marginal OLS and WLS}

Assume without loss of generality that we are interested in the treatment effect at lag $0$. To directly estimate the causal effect, we center the treatment and assign a weight of $1/[p_t(1-p_t)]$ to each time point $t$. For comparison, \cite{gao2023causal} did not center the treatment and use weight $1/p_t$ and $1/(1-p_t)$ since they estimated the average potential outcomes for each group seperately. The only distinction between marginal OLS and WLS is that marginal OLS applies weights to the treatment indicators, whereas WLS applies weights to the observations. Let $\tilde\mZ = (Z_1 - p_1,\ldots,Z_T - p_T)^\top$, $\tilde \mW = \diag(1/[p_t(1-p_t)],t=1,\ldots,T)$, and $\tilde\mY = (Y_1,\ldots,Y_T)^\top$. Then the WLS estimator can be written as
\begin{align*}
  (\tilde\mZ^\top \tilde\mW \tilde\mZ)^{-1}  (\tilde\mZ^\top \tilde\mW \tilde\mY) = \Big(\frac{1}{T}\sum_{t=1}^T \frac{(Z_t-p_t)^2}{p_t(1-p_t)}\Big)^{-1} \Big(\frac{1}{T}\sum_{t=1}^T \tZ_t Y_t\Big).
\end{align*}

In contrast, marginal OLS estimator can be written as
\begin{align*}
  \hat{\tau} = \Big(\Big(\frac{1}{T} \sum_{t=1}^T \frac{1}{p_t(1-p_t)}\Big)^{-1}\frac{1}{T}\sum_{t=1}^T \frac{(Z_t-p_t)^2}{p_t^2(1-p_t)^2}\Big)^{-1} \Big(\frac{1}{T}\sum_{t=1}^T \tZ_t Y_t\Big).
\end{align*}

The numerators of these two estimators are identical, and under certain conditions, both denominators converge to 1. If $p_t$'s values are constant, the two estimators coincide. However, when we infer treatment effects across multiple lags using full OLS, WLS cannot be generalized.

\section{Lemmas for the main results}\label{sec:lemmas}

We need the following lemmas to prove the main results.

\begin{lemma}\label{lemma:matrix}
  Assume Assumption~\ref{asp:no-anticipation}, Assumption~\ref{asp:design} and Assumption~\ref{asp:ps}. If $K^2/T \to 0$, then
  \begin{align*}
    \left\lVert \frac{1}{T-K} \mZ_K^\top \mZ_K - \mW_K^{-1} \right\rVert_\F \stackrel{\sf p}{\longrightarrow} 0.
  \end{align*}
\end{lemma}

\begin{lemma}\label{lemma:consistency}
  Assume Assumption~\ref{asp:no-anticipation}, Assumption~\ref{asp:design}, Assumption~\ref{asp:ani}\ref{asp:ani:i} with $K \sum_{k=1}^T \theta_{T,k}/T \to 0$ or Assumption~\ref{asp:mdep}\ref{asp:mdep:i} with $K m/T \to 0$, Assumption~\ref{asp:moment} and $K^2/T \to 0$. Then we have
  \begin{align*}
    \left\lVert \frac{1}{T-K} \mZ_K^\top (\mY_K - \mZ_K \mW_K \tau) \right\rVert_2 \stackrel{\sf p}{\longrightarrow} 0.
  \end{align*}
\end{lemma}

\begin{lemma}\label{lemma:clt}
  Assume Assumption~\ref{asp:no-anticipation}, Assumption~\ref{asp:design}, Assumption~\ref{asp:ani}\ref{asp:ani:ii} with $K$ fixed or Assumption~\ref{asp:mdep}\ref{asp:mdep:ii} with $K^4/T \to 0$, Assumption~\ref{asp:moment} and Assumption~\ref{asp:eigen}. Then for any $\lambda_K \in \bR^{K+1}$ satisfying $0 < \inf_{T \ge 1} (\lVert \lambda_K \rVert_2/\lVert \lambda_K \rVert_1) \le \sup_{T \ge 1} (\lVert \lambda_K \rVert_2/\lVert \lambda_K \rVert_1) < \infty$, we have
  \begin{align*}
    \sqrt{T-K} (\lambda_K^\top \mV \lambda_K)^{-1/2} \lambda_K^\top \left[\frac{1}{T-K} \mZ_K^\top (\mY_K - \mZ_K \mW_K \tau)\right] \stackrel{\sf d}{\longrightarrow} N(0,1).
  \end{align*}
\end{lemma}

Now we prove these lemmas in the following subsections.

\subsection{Proof of Lemma~\ref{lemma:matrix}}

\begin{proof}[Proof of Lemma~\ref{lemma:matrix}]

Recall 
\[
\frac{1}{T-K}\mZ_K^\top \mZ_K = \frac{1}{T-K} \sum_{t=K+1}^T \tZ_{t:t-K} \tZ_{t:t-K}^\top,
\]
with the $(i,j)$-th entry
\[
\left[\frac{1}{T-K} \mZ_K^\top \mZ_K\right]_{ij} = \frac{1}{T-K} \sum_{t=K+1}^T \tZ_{t-i+1} \tZ_{t-j+1}.
\]

Under Assumptions~\ref{asp:design} and \ref{asp:ps}, the sequence $Z_t$ is independent over $t$, and $\E[\tZ_{t}]=0,\E[\tZ_{t}^2]=\Var[Z_t]^{-1}=\big(p_t(1-p_t)\big)^{-1}$, with $p_t\in(0,1)$ uniformly bounded away from $0$ and $1$. If $i\neq j$, independence and centering imply $\E[\tZ_{t-i+1}\tZ_{t-j+1}]=0$. If $i=j$, $\E[\tZ_{t-i+1}^2]=\big(p_{t-i+1}(1-p_{t-i+1})\big)^{-1}$. Therefore 
\[
\E\left[\frac{1}{T-K} \mZ_K^\top \mZ_K\right] = \mW_K^{-1}.
\]

For $i\neq j$, we have $\E[\tZ_{t-i+1}\tZ_{t-j+1}]=0$ and $\sup_t \E[(\tZ_{t-i+1}\tZ_{t-j+1})^2] < \infty$ since the treatment probability $p_t$ is uniformly bounded away from $0$ and $1$. Moreover, because $\tZ_t$ is independent across $t$, we have for $t \neq t'$,
\[
\Cov(\tZ_{t-i+1}\tZ_{t-j+1},\tZ_{t'-i+1}\tZ_{t'-j+1})=0.
\]
Consequently,
\begin{align*}
  &\Var\left[\left[\frac{1}{T-K}\mZ_K^\top \mZ_K - \mW_K^{-1}\right]_{ij}\right] = \Var\left[\frac{1}{T-K} \sum_{t=K+1}^T \tZ_{t-i+1} \tZ_{t-j+1}\right] \\
  =& \frac{1}{(T-K)^2} \sum_{t=K+1}^T \Var[\tZ_{t-i+1} \tZ_{t-j+1}] \le \frac{1}{T-K} \sup_t \E[(\tZ_{t-i+1} \tZ_{t-j+1})^2] \lesssim \frac{1}{T-K}.
\end{align*}

For $i=j$, we have $\sup_t \E[\tZ_{t-i+1}^2] < \infty$ and $\tZ_t$ is independent across $t$. Thus
\begin{align*}
  &\Var\left[\left[\frac{1}{T-K}\mZ_K^\top \mZ_K - \mW_K^{-1}\right]_{ij}\right] = \Var\left[\frac{1}{T-K} \sum_{t=K+1}^T \tZ_{t-i+1}^2 \right] \\
  =& \frac{1}{(T-K)^2} \sum_{t=K+1}^T \Var[\tZ_{t-i+1}^2] \le \frac{1}{T-K} \sup_t \E[\tZ_{t-i+1}^2] \lesssim \frac{1}{T-K}.
\end{align*}

Combine the entrywise bounds:
\[
\E\left\Vert\frac{1}{T-K}\mZ_K^\top \mZ_K - \mW_K^{-1}\right\Vert_\F^2 = \sum_{i,j=1}^{K+1} \Var\left[\left[\frac{1}{T-K}\mZ_K^\top \mZ_K - \mW_K^{-1}\right]_{ij}\right] \lesssim \frac{K^2}{T-K}.
\]
By Markov's inequality, if $K^2/T \to 0$, then
\[
\left\lVert \frac{1}{T-K}\mZ_K^\top \mZ_K - \mW_K^{-1}\right\rVert_\F \stackrel{\sf p}{\longrightarrow} 0.
\]
This completes the proof.

\end{proof}

\subsection{Proof of Lemma~\ref{lemma:consistency}}

\begin{proof}[Proof of Lemma~\ref{lemma:consistency}]
Expanding the term, we have:
\begin{align}\label{eq:appendix_1}
  \frac{1}{T-K} \mZ_K^\top (\mY_K - \mZ_K \mW_K \tau) = \frac{1}{T-K} \sum_{t=K+1}^T \tZ_{t:t-K} (Y_t - \tZ_{t:t-K}^\top \mW_K \tau).
\end{align}

We first consider the mean of \eqref{eq:appendix_1}. By the independence of $\tZ_t$ and $\E[\tZ_t] = 0$, we can show that $\E[(T-K)^{-1}\mZ_K^\top (\mY_K - \mZ_K \mW_K \tau)] = 0$ since for any $k = 0, \ldots, K$, 
\begin{align*}
  &\E\left[\frac{1}{T-K} \sum_{t=K+1}^T \tZ_{t-k}\left(Y_t - \tZ_{t:t-K}^\top \mW_K \tau\right)\right] = \E\left[\frac{1}{T-K} \sum_{t=K+1}^T \tZ_{t-k}\left(Y_t - \sum_{\ell=0}^K \tZ_{t-\ell}w_\ell \tau_\ell\right)\right]\\
  =& \frac{1}{T-K} \sum_{t=K+1}^T \left(\E[\tZ_{t-k} Y_t] - \Var[\tZ_{t-k}] w_k \tau_k\right) \\
  =& \left(\frac{1}{T-K} \sum_{t=K+1}^T \E[\tZ_{t-k} Y_t]\right) - \left(\frac{1}{T-K} \sum_{t=K+1}^T \Var[\tZ_{t-k}]\right) w_k \tau_k \\
  =& \left(\frac{1}{T-K} \sum_{t=K+1}^T \tau_{t,k} \right) - \left(\frac{1}{T-K} \sum_{t=K+1}^T \Var[Z_{t-k}]^{-1}\right) w_k \tau_k = \tau_k - \tau_k = 0.
\end{align*}

We then consider the variance of \eqref{eq:appendix_1}. For any $k = 0, \ldots, K$, we have
\begin{align*}
  &\Var\left[\frac{1}{T-K} \sum_{t=K+1}^T \tZ_{t-k}\left(Y_t - \tZ_{t:t-K}^\top \mW_K \tau\right)\right] = \Var\left[\frac{1}{T-K} \sum_{t=K+1}^T \tZ_{t-k}\left(Y_t - \sum_{\ell=0}^K \tZ_{t-\ell}w_\ell \tau_\ell\right)\right]\\
  =& \frac{1}{(T-K)^2} \sum_{t=K+1}^T \Var\left[\tZ_{t-k}\left(Y_t - \sum_{\ell=0}^K \tZ_{t-\ell}w_\ell \tau_\ell\right)\right] \\
  &+ \frac{1}{(T-K)^2} \sum_{t,t'=K+1, t\neq t'}^T \Cov\left[\tZ_{t-k}\left(Y_t - \sum_{\ell=0}^K \tZ_{t-\ell}w_\ell \tau_\ell\right), \tZ_{t'-k}\left(Y_{t'} - \sum_{\ell=0}^K \tZ_{t'-\ell}w_\ell \tau_\ell\right)\right].
\end{align*}

By Assumption~\ref{asp:moment},
\[
\sup_{t} \Var\left[\tZ_{t-k}\left(Y_t - \sum_{\ell=0}^K \tZ_{t-\ell}w_\ell \tau_\ell\right)\right] \lesssim \sup_t \lvert Y_t \rvert^2 + \sup_t \lVert \tau \rVert_1^2 < \infty,
\]
which implies
\[
  \frac{1}{(T-K)^2} \sum_{t=K+1}^T \Var\left[\tZ_{t-k}\left(Y_t - \sum_{\ell=0}^K \tZ_{t-\ell}w_\ell \tau_\ell\right)\right] \lesssim \frac{1}{T-K}.
\]

By the independence of $\tZ_t$, for any $t,t' \ge K+1$ with $\lvert t-t' \rvert > K$, we have
\begin{align*}
  &\Cov\left[\tZ_{t-k}\left(Y_t - \sum_{\ell=0}^K \tZ_{t-\ell}w_\ell \tau_\ell\right), \tZ_{t'-k}\left(Y_{t'} - \sum_{\ell=0}^K \tZ_{t'-\ell}w_\ell \tau_\ell\right)\right] \\
  =& \Cov\left[\tZ_{t-k} Y_t, \tZ_{t'-k} Y_{t'}\right] - \Cov\left[\tZ_{t-k}Y_t, \tZ_{t'-k} \sum_{\ell=0}^K \tZ_{t'-\ell}w_\ell \tau_\ell\right] - \Cov\left[\tZ_{t-k} \sum_{\ell=0}^K \tZ_{t-\ell}w_\ell \tau_\ell, \tZ_{t'-k} Y_{t'}\right].
\end{align*}

We handle the dependence structure in two cases corresponding to Assumptions~\ref{asp:ani} and~\ref{asp:mdep}, respectively.

\subsubsection{Decaying carryover effects condition (Assumption~\ref{asp:ani})}

Assume $t>t'$ without loss of generality. Consider $Z'_{t':1} = (Z'_{t'}, \ldots, Z'_1)^\top$ independent and identically distributed with $Z_{t':1}$. Then
\begin{align*}
  &\E\left[\tZ_{t-k} Y_t \tZ_{t'-k} Y_{t'}\right] = \E\left[\tZ_{t-k} Y_t(Z_{t:1}) \tZ_{t'-k} Y_{t'}(Z_{t':1})\right]\\
  =& \E\left[\tZ_{t-k} Y_t(Z_{t:1}) \tZ_{t'-k} Y_{t'}(Z_{t':1})\right] - \E\left[\tZ_{t-k} Y_t(Z_{t:{t'+1}}, Z'_{t':1}) \tZ_{t'-k} Y_{t'}(Z_{t':1})\right] \\
  &+ \E\left[\tZ_{t-k} Y_t(Z_{t:{t'+1}}, Z'_{t':1}) \tZ_{t'-k} Y_{t'}(Z_{t':1})\right].
\end{align*}

Since $(Z_{t:t'+1}, Z'_{t':1})$ is independent of $Z_{t':1}$, we have
\[
  \E\left[\tZ_{t-k} Y_t(Z_{t:{t'+1}}, Z'_{t':1}) \tZ_{t'-k} Y_{t'}(Z_{t':1})\right] = \E\left[\tZ_{t-k} Y_t(Z_{t:1})\right] \E\left[\tZ_{t'-k} Y_{t'}(Z_{t':1})\right] = \E\left[\tZ_{t-k} Y_t\right] \E\left[\tZ_{t'-k} Y_{t'}\right].
\]

By Assumption~\ref{asp:moment} and Assumption~\ref{asp:ani}, we have
\begin{align*}
  &\left\lvert \E\left[\tZ_{t-k} Y_t(Z_{t:1}) \tZ_{t'-k} Y_{t'}(Z_{t':1})\right] - \E\left[\tZ_{t-k} Y_t(Z_{t:{t'+1}}, Z'_{t':1}) \tZ_{t'-k} Y_{t'}(Z_{t':1})\right] \right\rvert \\
  =& \left\lvert \E\left[\tZ_{t-k} \left(Y_t(Z_{t:1}) - Y_t(Z_{t:{t'+1}}, Z'_{t':1}) \right) \tZ_{t'-k} Y_{t'}(Z_{t':1})\right] \right\rvert \\
  \le& \E\left[\left\lvert Y_t(Z_{t:1}) - Y_t(Z_{t:{t'+1}}, Z'_{t':1})\right\rvert \right] \lesssim \theta_{T,\lvert t-t' \rvert}.
\end{align*}

Therefore, we have
\begin{align*}
  \left\lvert\Cov\left[\tZ_{t-k} Y_t, \tZ_{t'-k} Y_{t'}\right] \right\rvert = \left\lvert \E\left[\tZ_{t-k} Y_t \tZ_{t'-k} Y_{t'}\right] - \E\left[\tZ_{t-k} Y_t\right] \E\left[\tZ_{t'-k} Y_{t'}\right] \right\rvert \lesssim \theta_{T,\lvert t-t' \rvert}.
\end{align*}

By Assumption~\ref{asp:moment} and Assumption~\ref{asp:ani}, we can show similarly that
\begin{align*}
  \left\lvert \Cov\left[\tZ_{t-k}Y_t, \tZ_{t'-k} \sum_{\ell=0}^K \tZ_{t'-\ell}w_\ell \tau_\ell\right] \right\rvert\lesssim \theta_{T,\lvert t-t' \rvert}.
\end{align*}

By $t > t'$ and $\lvert t-t' \rvert > K$, we have
\begin{align*}
  \Cov\left[\tZ_{t-k} \sum_{\ell=0}^K \tZ_{t-\ell}w_\ell \tau_\ell, \tZ_{t'-k} Y_{t'}\right] = 0.
\end{align*}

Therefore, we have
\begin{align*}
  \left\lvert\Cov\left[\tZ_{t-k}\left(Y_t - \sum_{\ell=0}^K \tZ_{t-\ell}w_\ell \tau_\ell\right), \tZ_{t'-k}\left(Y_{t'} - \sum_{\ell=0}^K \tZ_{t'-\ell}w_\ell \tau_\ell\right)\right]\right\rvert \lesssim \theta_{T,\lvert t-t' \rvert}.
\end{align*}

Summing over all time pairs gives
\begin{align*}
\Var\left[\frac{1}{T-K} \sum_{t=K+1}^T \tZ_{t-k}\left(Y_t - \tZ_{t:t-K}^\top \mW_K \tau\right)\right] \lesssim \frac{1}{T-K} \left(K + \sum_{k=K+1}^{T-K} \theta_{T,k}\right).
\end{align*}

\subsubsection{m-dependence condition (Assumption~\ref{asp:mdep})}

By Assumption~\ref{asp:moment} and Assumption~\ref{asp:mdep}, we have
\[
\left\lvert\Cov\left[\tZ_{t-k}\left(Y_t - \sum_{\ell=0}^K \tZ_{t-\ell}w_\ell \tau_\ell\right), \tZ_{t'-k}\left(Y_{t'} - \sum_{\ell=0}^K \tZ_{t'-\ell}w_\ell \tau_\ell\right)\right]\right\rvert \lesssim \ind\left(\left\lvert t-t' \right\rvert \le m\right).
\]

Summing over all time pairs gives
\[
\Var\left[\frac{1}{T-K} \sum_{t=K+1}^T \tZ_{t-k}\left(Y_t - \tZ_{t:t-K}^\top \mW_K \tau\right)\right] \lesssim \frac{1}{T-K} \left(K + m \right).
\]

\subsubsection{Final step}

By Assumption~\ref{asp:ani}\ref{asp:ani:i} with $K \sum_{k=1}^T \theta_{T,k}/T \to 0$ or Assumption~\ref{asp:mdep}\ref{asp:mdep:i} with $K m/T \to 0$, and $K^2/T \to 0$, we have
\begin{align*}
  \E \left\lVert \frac{1}{T-K} \mZ_K^\top (\mY_K - \mZ_K \mW_K \tau) \right\rVert_2^2 = \sum_{k=0}^K \Var\left[\frac{1}{T-K} \sum_{t=K+1}^T \tZ_{t-k}\left(Y_t - \tZ_{t:t-K}^\top \mW_K \tau\right)\right] \to 0.
\end{align*}

This completes the proof by Markov's inequality.

\end{proof}

\subsection{Proof of Lemma~\ref{lemma:clt}}

\begin{proof}[Proof of Lemma~\ref{lemma:clt}]

In the same way as the proof of Lemma~\ref{lemma:consistency}, expanding the term, we have:
\[
  \frac{1}{T-K} \mZ_K^\top (\mY_K - \mZ_K \mW_K \tau) = \frac{1}{T-K} \sum_{t=K+1}^T \tZ_{t:t-K} (Y_t - \tZ_{t:t-K}^\top \mW_K \tau).
\]

We can write
\begin{align*}
  &\sqrt{T-K} (\lambda_K^\top \mV \lambda_K)^{-1/2} \lambda_K^\top \left[\frac{1}{T-K} \mZ_K^\top (\mY_K - \mZ_K \mW_K \tau)\right]\\
  =&\sqrt{T-K} (\lambda_K^\top \mV \lambda_K)^{-1/2} \lambda_K^\top \left[\frac{1}{T-K} \sum_{t=K+1}^T \tZ_{t:t-K} (Y_t - \tZ_{t:t-K}^\top \mW_K \tau)\right]\\
  =& \sum_{t=K+1}^T (T-K)^{-1/2}  (\lambda_K^\top \mV \lambda_K)^{-1/2}  \lambda_K^\top \tZ_{t:t-K} (Y_t - \tZ_{t:t-K}^\top \mW_K \tau).
\end{align*}

Define the triangular array $[U_{T,t}]_{t=K+1}^T$ by
\begin{align*}
  U_{T,t} = (T-K)^{-1/2} (\lambda_K^\top \mV \lambda_K)^{-1/2}  \lambda_K^\top \tZ_{t:t-K} (Y_t - \tZ_{t:t-K}^\top \mW_K \tau).
\end{align*}

As shown earlier in the proof of Lemma~\ref{lemma:consistency},
\[
\E\left[\tZ_{t:t-K}(Y_t-\tZ_{t:t-K}^\top \mW_K\tau)\right]=0,
\]
hence \(\E[U_{T,t}]=0\). By the definition of \(\mV\) in \eqref{eq:mV}, we have
\[
\Var\left[\sum_{t=K+1}^T U_{T,t}\right] = 1.
\]

We then establish the moment bound for $U_{T,t}$. For $p \ge 2$, using $(\lambda_K^\top \mV \lambda_K)\ge \lambda_{\min}(\mV)\lVert \lambda_K \rVert_2^2$, the $p$-th moment of $U_{T,t}$ is bounded by
\begin{align*}
  \E[\lvert U_{T,t}\rvert^p] =& (T-K)^{-p/2} (\lambda_K^\top \mV \lambda_K)^{-p/2} \E[\lvert \lambda_K^\top \tZ_{t:t-K} (Y_t - \tZ_{t:t-K}^\top \mW_K \tau) \rvert^p]\\
  \le & (T-K)^{-p/2} \lVert \lambda_K \rVert_2^{-p} \lambda_{\min}(\mV)^{-p/2} \E[\lvert \lambda_K^\top \tZ_{t:t-K} (Y_t - \tZ_{t:t-K}^\top \mW_K \tau)\rvert^p]\\
  =& (T-K)^{-p/2} \lVert \lambda_K \rVert_2^{-p} \lambda_{\min}(\mV)^{-p/2} \E[\lvert \lambda_K^\top \tZ_{t:t-K}\rvert^p \lvert Y_t - \tZ_{t:t-K}^\top \mW_K \tau \rvert^p]\\
  \lesssim& (T-K)^{-p/2} \lVert \lambda_K \rVert_2^{-p} \lambda_{\min}(\mV)^{-p/2} \lVert \lambda_K \rVert_1^p \E[\lvert Y_t - \tZ_{t:t-K}^\top \mW_K \tau \rvert^p]\\
  \lesssim & (T-K)^{-p/2} \lambda_{\min}(\mV)^{-p/2} (\lVert \lambda_K \rVert_1/\lVert \lambda_K \rVert_2)^p \E[\lvert Y_t \rvert^p + \lVert \tau \rVert_1^p]\\
  \lesssim& (T-K)^{-p/2},
\end{align*}
where we used $0 < \inf_{T \ge 1} (\lVert \lambda_K \rVert_2/\lVert \lambda_K \rVert_1) \le \sup_{T \ge 1} (\lVert \lambda_K \rVert_2/\lVert \lambda_K \rVert_1) < \infty$, the uniform boundedness in Assumption~\ref{asp:moment}, and the lower bound of the smallest eigenvalue in Assumption~\ref{asp:eigen}. So the triangular array $[(T-K)^{1/2}U_{T,t}]_{t=K+1}^T$ has uniformly bounded $p$-th moments.

We now handle the dependence structure in two cases corresponding to Assumptions~\ref{asp:ani} and~\ref{asp:mdep}, respectively.

\subsubsection{Decaying carryover effects condition (Assumption~\ref{asp:ani}).}

Before starting the proof, we first introduce the CLT for dependent triangular arrays with general covariance-based conditions in \cite{chandrasekhar2023general}. 

\begin{lemma}[Theorem 1 in \cite{chandrasekhar2023general} under dimension one]\label{lemma:clt_general}
  Let $[Z_{n,i}]_{i=1}^n$ be a triangular array of random variables with mean zero. For each $n$, we define the affinity set $\mathcal A^n_i \subset \{1, \ldots, n\}$ for each $i$. Let $\Omega_n := \sum_{i=1}^n \sum_{j \in \mathcal A^n_i} \Cov[Z_{n,i}, Z_{n,j}]$. Let $Z_{n,-i} := \sum_{j \notin \mathcal A^n_i} Z_{n,j}$. Assume that
  \begin{enumerate}
    \item $\sum_{i=1}^n \sum_{j,k \in \mathcal A^n_i} \E[\lvert Z_{n,i} \rvert Z_{n,j} Z_{n,k}] = o(\lvert \Omega_n \rvert^{3/2})$.
    \item $\sum_{i,j=1}^n \sum_{k \in \mathcal A^n_i, \ell \in A^n_j} \Cov[Z_{n,i} Z_{n,k}, Z_{n,j} Z_{n,\ell}] = o(\lvert \Omega_n \rvert^2)$.
    \item $\sum_{i=1}^n \E[\lvert\E[Z_{n,i} Z_{n,-i} \given Z_{n,-i}]\rvert] = \sum_{i=1}^n \E[\lvert Z_{n,-i} \E[Z_{n,i}  \given Z_{n,-i}]\rvert] = o(\lvert \Omega_n \rvert)$.
  \end{enumerate}
  Then the triangular array $[Z_{n,i}]_{i=1}^n$ satisfies the CLT:
  \[
    \Omega_n^{-1/2} \sum_{i=1}^n Z_{n,i} \stackrel{\sf d}{\longrightarrow} N(0,1).
  \]
\end{lemma}

To apply Lemma~\ref{lemma:clt_general}, we follow the proof in Section 3 of \cite{chandrasekhar2023general}. We define $\mathcal A^T_t = \{t': \lvert t-t' \rvert \le D\}$, where $D = T^{\epsilon} \to \infty$ as $T \to \infty$ for some small $\epsilon > 0$.

By the selection of $D$, we have
\begin{align*}
  \left\lvert\sum_{t=K+1}^T \sum_{t' \notin \mathcal A^T_t} \Cov\left[U_{T,t}, U_{T,t'}\right] \right\vert \lesssim \sum_{t=K+1}^T \sum_{t' \notin \mathcal A^T_t} \frac{1}{T-K}  \theta_{T, \lvert t-t' \rvert} \lesssim \sum_{k=D}^T \theta_{T,k} \lesssim D^{-\delta} = o(1),
\end{align*}
which implies that
\begin{align*}
  \Omega_T =& \sum_{t=K+1}^T \sum_{t' \in \mathcal A^T_t} \Cov\left[U_{T,t}, U_{T,t'}\right] = \Var\left[\sum_{t=K+1}^T U_{T,t}\right] - \sum_{t=K+1}^T \sum_{t' \notin \mathcal A^T_t} \Cov\left[U_{T,t}, U_{T,t'}\right]\\
  =& 1 - o(1) \to 1.
\end{align*}

For the first condition in Lemma~\ref{lemma:clt_general}, we have
\begin{align*}
  & \left\lvert \sum_{t=K+1}^T \sum_{t',t'' \in \mathcal A^T_t} \E[\lvert U_{T,t} \rvert U_{T,t'} U_{T,t''}] \right\rvert \lesssim D^2 \sum_{t=K+1}^T \E[\lvert U_{T,t} \rvert^3] \lesssim D^2 (T-K) (T-K)^{-3/2}\\
  =& D^2 (T-K)^{-1/2} = o(1) = o(\lvert \Omega_T \rvert^{3/2}).
\end{align*}

For the second condition in Lemma~\ref{lemma:clt_general}, we have
\begin{align*}
  & \left\lvert \sum_{s,t=K+1}^T \sum_{s' \in \mathcal A^T_s, t' \in \mathcal A^T_t} \Cov[U_{T,s} U_{T,s'}, U_{T,t} U_{T,t'}] \right\rvert \lesssim \sum_{t=K+1}^T \sum_{\lvert s-t \rvert \le D, \lvert t'-t \rvert \le D, \lvert s'-t \rvert \le 2D} \E[\lvert U_{T,t} \rvert^4]\\
  \lesssim & D^3 (T-K) (T-K)^{-2} = D^3 (T-K)^{-1} = o(1) = o(\lvert \Omega_T \rvert^2).
\end{align*}

For the third condition in Lemma~\ref{lemma:clt_general}, defining the random variable $\xi_{t,-t}$ such that $\xi_{t,-t} = 1$ if $\E[U_{T,t} \given U_{T,-t}] > 0$ and $\xi_{t,-t} = -1$ otherwise, then we have
\begin{align*}
  & \left\lvert \sum_{t=K+1}^T \E[\left\lvert \E[U_{T,t} U_{T,-t} \given U_{T,-t}] \right\rvert] \right\rvert = \left\lvert \sum_{t=K+1}^T \sum_{t' \notin \mathcal A^T_t} \Cov\left[\xi_{t,-t} U_{T,t'}, U_{T,t}\right]\right\rvert\\
  \lesssim & \sum_{t=K+1}^T \sum_{t' \notin \mathcal A^T_t} \frac{1}{T-K} \theta_{T, \lvert t-t' \rvert} \lesssim \sum_{k=D}^T \theta_{T,k} \lesssim D^{-\delta} = o(1) = o(\lvert \Omega_T \rvert).
\end{align*}

Then by Lemma~\ref{lemma:clt_general} and $\Omega_T \to 1$, we have
\[
\sum_{t=K+1}^{T} U_{T,t} \stackrel{\sf d}{\longrightarrow} N(0,1),
\]
which is exactly
\[
  \sqrt{T-K} (\lambda_K^\top \mV \lambda_K)^{-1/2} \lambda_K^\top \left[\frac{1}{T-K} \mZ_K^\top (\mY_K - \mZ_K \mW_K \tau)\right] \stackrel{\sf d}{\longrightarrow} N(0,1).
\]

\subsubsection{m-dependence condition (Assumption~\ref{asp:mdep}).}

Before starting the proof, we first introduce the CLT for $m$-dependent triangular arrays in \cite{chen2004normal}.

\begin{lemma}[Theorem 2.6 in \cite{chen2004normal}]\label{lemma:clt_mdependent}
  Let $[Z_{n,i}]_{i=1}^n$ be a $m$-dependent triangular array of random variables with mean zero, $\Var[\sum_{i=1}^n Z_{n,i}] = 1$, and finite $\E[\lvert Z_{n,i} \rvert^p] < \infty$ for $2 < p \le 3$. Let $F$ be the distribution function of $\sum_{i=1}^n Z_{n,i}$ and $\Phi$ be the distribution function of $N(0,1)$. Then
  \[
  \sup_{z} \lvert F(z) - \Phi(z) \rvert \le 75 (10m+1)^{p-1} \sum_{i=1}^n \E[\lvert Z_{n,i} \rvert^p].
  \]
\end{lemma}

To apply Lemma~\ref{lemma:clt_mdependent}, note that by Assumption~\ref{asp:mdep}, the triangular array $[U_{T,t}]_{t=K+1}^T$ is $(K \vee m)$-dependent since each $U_{T,t}$ is measurable with respect to $\sigma(Y_t, Z_t, \ldots,Z_{t-K}) = \sigma(Y_t(Z_{t:t-m}), Z_t, \ldots,Z_{t-K})$. From the moment bound already established for $U_{T,t}$ (take $p=3$), $\sup_{t} \E\left[\lvert U_{T,t} \rvert^3\right] \lesssim (T-K)^{-3/2}$. Let $F_T$ be the distribution function of $\sum_{t=K+1}^T U_{T,t}$, then by Lemma~\ref{lemma:clt_mdependent} with $p=3$, we have
\begin{align*}
  \sup_{z} \lvert F_T(z) - \Phi(z) \rvert \lesssim (K \vee m)^2 (T-K) (T-K)^{-3/2} = (K \vee m)^2(T-K)^{-1/2}.
\end{align*}

Then by Assumption~\ref{asp:mdep}\ref{asp:mdep:ii} with $K^4/T \to 0$, we have
\[
\sum_{t=K+1}^{T} U_{T,t} \stackrel{\sf d}{\longrightarrow} N(0,1),
\]
which is exactly
\[
  \sqrt{T-K} (\lambda_K^\top \mV \lambda_K)^{-1/2} \lambda_K^\top \left[\frac{1}{T-K} \mZ_K^\top (\mY_K - \mZ_K \mW_K \tau)\right] \stackrel{\sf d}{\longrightarrow} N(0,1).
\]

\end{proof}

\section{Proofs}\label{sec:proofs}

\subsection{Proof of Theorem~\ref{thm:consistency}}

\begin{proof}[Proof of Theorem~\ref{thm:consistency}]

Recall \eqref{eq:hat_tau} that 
\[
  \hat{\tau} = \tau + \mW_K^{-1} [(T-K)^{-1}\mZ_K^\top \mZ_K]^{-1} [(T-K)^{-1}\mZ_K^\top (\mY_K - \mZ_K \mW_K \tau)].
\]
  
By  Lemma~\ref{lemma:matrix}, under the assumptions of the theorem, we have
\[
  \left\lVert \frac{1}{T-K} \mZ_K^\top \mZ_K - \mW_K^{-1} \right\rVert_\F \stackrel{\sf p}{\longrightarrow} 0.
\]

Since $\mW_K^{-1}$ is positive definite with eigenvalues uniformly bounded away from 0 by Assumption~\ref{asp:ps}, we have
\[
  \left\Vert \mW_K^{-1} \left(\frac{1}{T-K}\mZ_K^\top \mZ_K\right)^{-1}\right\Vert_2 \stackrel{\sf p}{\longrightarrow} \lVert \mI_{K+1} \rVert_2 = 1.
\]

By Lemma~\ref{lemma:consistency}, under the assumptions of the theorem, we have
\[
  \left\lVert \frac{1}{T-K} \mZ_K^\top (\mY_K - \mZ_K \mW_K \tau) \right\rVert_2 \stackrel{\sf p}{\longrightarrow} 0.
\]
  
Then we obtain
\[
  \left\lVert \hat\tau - \tau \right\rVert_2
  \le
  \left\lVert \mW_K^{-1} \left(\frac{1}{T-K}\mZ_K^\top \mZ_K\right)^{-1}\right\Vert_2
  \left\lVert \frac{1}{T-K} \mZ_K^\top (\mY_K - \mZ_K \mW_K \tau) \right\rVert_2
  \stackrel{\sf p}{\longrightarrow} 0,
\]
which completes the proof.
\end{proof}

\subsection{Proof of Theorem~\ref{thm:clt}}

\begin{proof}[Proof of Theorem~\ref{thm:clt}]

  Recall \eqref{eq:hat_tau} that 
  \[
    \hat{\tau} = \tau + \mW_K^{-1} [(T-K)^{-1}\mZ_K^\top \mZ_K]^{-1} [(T-K)^{-1}\mZ_K^\top (\mY_K - \mZ_K \mW_K \tau)].
  \]

Then we have the following decomposition:
\begin{align*}
  \lambda_K^\top (\hat\tau - \tau) =& \lambda_K^\top \mW_K^{-1} [(T-K)^{-1}\mZ_K^\top \mZ_K]^{-1} [(T-K)^{-1}\mZ_K^\top (\mY_K - \mZ_K \mW_K \tau)] \\
  =& \lambda_K^\top \left[\frac{1}{T-K} \mZ_K^\top \left(\mY_K - \mZ_K \mW_K \tau\right)\right] \\
  &+ \lambda_K^\top \left[\mW_K^{-1} \left(\frac{1}{T-K}\mZ_K^\top \mZ_K\right)^{-1} - \mI_{K+1}\right] \left[\frac{1}{T-K} \mZ_K^\top \left(\mY_K - \mZ_K \mW_K \tau\right)\right].
\end{align*}

By Lemma~\ref{lemma:clt}, under the assumptions of the theorem, we have
\begin{align*}
  \sqrt{T-K} (\lambda_K^\top \mV \lambda_K)^{-1/2} \lambda_K^\top \left[\frac{1}{T-K} \mZ_K^\top (\mY_K - \mZ_K \mW_K \tau)\right] \stackrel{\sf d}{\longrightarrow} N(0,1).
\end{align*}

By Lemma~\ref{lemma:matrix}, under the assumptions of the theorem, we have
\[
  \left\lVert \mW_K^{-1} \left(\frac{1}{T-K}\mZ_K^\top \mZ_K\right)^{-1} - \mI_{K+1} \right\rVert_2 \stackrel{\sf p}{\longrightarrow} 0.
\]

Combining the above two results, we have
\begin{align*}
  \sqrt{T-K} (\lambda_K^\top \mV \lambda_K)^{-1/2} \lambda_K^\top (\hat\tau - \tau) \stackrel{\sf d}{\longrightarrow} N(0,1),
\end{align*}
which completes the proof.

\end{proof}

\subsection{Proof of Theorem~\ref{thm:var}}

\begin{proof}[Proof of Theorem~\ref{thm:var}]

By Lemma~\ref{lemma:matrix}, we have $\lVert \frac{1}{T-K}\mZ_K^\top \mZ_K - \mW_K^{-1} \rVert_\F \stackrel{\sf p}{\longrightarrow} 0$, and then we have the asymptotic approximation:
\begin{align*}
  \hat\mV = \frac{1}{T-K}\mZ_K^\top \hat\mU_K \mQ_K \hat\mU_K \mZ_K + o_\P(1).
\end{align*}
We decompose $\hat\mV$ as follows:
\begin{align}\label{eq:appendix_2}
  &\hat\mV - \mV - \mB_K \nonumber \\
  =& \frac{1}{T-K} \left(\mZ_K^\top \hat\mU_K \mQ_K \hat\mU_K \mZ_K - \mZ_K^\top \mU_K \mQ_K \mU_K \mZ_K\right) \nonumber \\
  &+ \frac{1}{T-K} \left(\mZ_K^\top \mU_K \mQ_K \mU_K \mZ_K - \mB_K - (\mU_K \mZ_K - \E[\mU_K \mZ_K])^\top \mQ_K (\mU_K \mZ_K - \E[\mU_K \mZ_K])\right) \nonumber \\
  &+  \frac{1}{T-K} \Big((\mU_K \mZ_K - \E[\mU_K \mZ_K])^\top \mQ_K (\mU_K \mZ_K - \E[\mU_K \mZ_K]) \nonumber \\
  &- \E[(\mU_K \mZ_K - \E[\mU_K \mZ_K])^\top \mQ_K (\mU_K \mZ_K - \E[\mU_K \mZ_K])]\Big) \nonumber \\
  &+ \frac{1}{T-K} \left(\E[(\mU_K \mZ_K - \E[\mU_K \mZ_K])^\top \mQ_K (\mU_K \mZ_K - \E[\mU_K \mZ_K])] - \mV\right) + o_\P(1).
\end{align}

We bound the Frobenius norm of each term in \eqref{eq:appendix_2}.

\subsubsection{First term.}

For entries $(k,k')$, by the definition of $\hat\mU_K$,$\mU_K$ and $\mQ_K$, we have
\begin{align*}
  \left[\mZ_K^\top \hat\mU_K \mQ_K \hat\mU_K \mZ_K\right]_{k,k'} = \sum_{t=K+1}^T \sum_{t'=K+1}^T \tZ_{t-k+1} \hat U_t \tZ_{t'-k'+1} \hat U_{t'} \left(1 - \frac{\lvert t - t' \rvert}{L+1}\right) \ind \left(\lvert t - t' \rvert \le L\right), 
\end{align*}
and
\begin{align*}
  \left[\mZ_K^\top \mU_K \mQ_K \mU_K \mZ_K\right]_{k,k'} = \sum_{t=K+1}^T \sum_{t'=K+1}^T \tZ_{t-k+1} U_t \tZ_{t'-k'+1} U_{t'} \left(1 - \frac{\lvert t - t' \rvert}{L+1}\right) \ind \left(\lvert t - t' \rvert \le L\right).
\end{align*}

Then we obtain
\begin{align*}
  &\left\lvert \left[\frac{1}{T-K} \left(\mZ_K^\top \hat\mU_K \mQ_K \hat\mU_K \mZ_K - \mZ_K^\top \mU_K \mQ_K \mU_K \mZ_K\right)\right]_{k,k'} \right\rvert\\
  =& \left\lvert \frac{1}{T-K}\sum_{t,t'=K+1, \lvert t-t' \rvert \le L}^T \left(\tZ_{t-k+1} \hat U_t \tZ_{t'-k'+1} \hat U_{t'} - \tZ_{t-k+1} U_t \tZ_{t'-k'+1} U_{t'}\right) \left(1 - \frac{\lvert t - t' \rvert}{L+1}\right) \right\rvert\\
  \le& \frac{1}{T-K}\sum_{t,t'=K+1, \lvert t-t' \rvert \le L}^T \lvert \tZ_{t-k+1} \tZ_{t'-k'+1} \rvert (\lvert \hat U_t - U_t \rvert \lvert \hat U_{t'} \rvert + \lvert \hat U_{t'} - U_{t'} \rvert \lvert U_t \rvert) \left(1 - \frac{\lvert t - t' \rvert}{L+1}\right).
\end{align*}

By Theorem~\ref{thm:clt},
\[
  \left\lvert \hat U_t - U_t \right\rvert = \left\lvert \tZ_{t:t-K}^\top \mW_K (\hat\tau - \tau) \right\rvert \lesssim \left\lVert \hat\tau - \tau \right\rVert_1 = O_\P\left(\frac{K}{\sqrt{T-K}}\right),
\]
and by Assumption~\ref{asp:moment},
\[
  \left\lvert U_t \right\rvert + \left\lvert \hat U_t \right\rvert = \left\lvert Y_t - \tZ_{t:t-K}^\top \mW_K \hat\tau \right\rvert + \left\lvert Y_t - \tZ_{t:t-K}^\top \mW_K \tau \right\rvert \lesssim \left\lvert Y_t \right\rvert + \left\lVert \tau \right\rVert_1 = O_\P(1).
\]

Then we have
\begin{align*}
  \left\lvert \left[\frac{1}{T-K} \left(\mZ_K^\top \hat\mU_K \mQ_K \hat\mU_K \mZ_K - \mZ_K^\top \mU_K \mQ_K \mU_K \mZ_K\right)\right]_{k,k'} \right\rvert= O_\P\left(\frac{KL}{\sqrt{T-K}}\right),
\end{align*}
and then
\begin{align*}
  &\left\lVert \frac{1}{T-K} \left(\mZ_K^\top \hat\mU_K \mQ_K \hat\mU_K \mZ_K - \mZ_K^\top \mU_K \mQ_K \mU_K \mZ_K\right)\right\rVert_\F = O_\P\Big(\frac{K^2 L}{\sqrt{T-K}}\Big) = o_\P(1),
\end{align*}
when $K^2 L/T^{1/2} \to 0$.

\subsubsection{Second term.}

Since $\mb_K = \E[\mU_K \mZ_K]$ and $\mB_K = \mb_K^\top \mQ_K \mb_K$, we have
\begin{align*}
  &\mZ_K^\top \mU_K \mQ_K \mU_K \mZ_K - \mB_K - (\mU_K \mZ_K - \E[\mU_K \mZ_K])^\top \mQ_K (\mU_K \mZ_K - \E[\mU_K \mZ_K]) \\
  =& 2 (\mU_K \mZ_K - \E[\mU_K \mZ_K])^\top  \mQ_K \E[\mU_K \mZ_K].
\end{align*}

For entries $(k,k')$, by the definition of $\mU_K$, $\mZ_K$ and $\mQ_K$, we have
\begin{align*}
  &\left[(\mU_K \mZ_K - \E[\mU_K \mZ_K])^\top  \mQ_K \E[\mU_K \mZ_K]\right]_{k,k'}\\
  =&\sum_{t=K+1}^T \sum_{t'=K+1}^T \Big(\tZ_{t-k+1} U_t - \E[\tZ_{t-k+1} U_t] \Big) \E[\tZ_{t-k'+1} U_{t'}] \left(1 - \frac{\lvert t - t' \rvert}{L+1}\right) \ind \left(\lvert t - t' \rvert \le L\right),
\end{align*}
which implies
\begin{align*}
  &\left[\frac{1}{T-K} \left((\mU_K \mZ_K - \E[\mU_K \mZ_K])^\top  \mQ_K \E[\mU_K \mZ_K]\right)\right]_{k,k'}\\
  =& \frac{1}{T - K} \sum_{t,t'=K+1, \lvert t-t' \rvert \le L}^T \Big(\tZ_{t-k+1} U_t - \E[\tZ_{t-k+1} U_t] \Big) \E[\tZ_{t-k'+1} U_{t'}] \Big(1 - \frac{\lvert t - t' \rvert}{L+1}\Big)\\
  =& \frac{1}{T-K} \sum_{t=K+1}^T \Big[\sum_{t'=K+1, \lvert t'-t \rvert \le L}^T \E[\tZ_{t-k'+1} U_{t'}] \Big(1 - \frac{\lvert t - t' \rvert}{L+1}\Big)\Big] \Big(\tZ_{t-k+1} U_t - \E[\tZ_{t-k+1} U_t] \Big).
\end{align*}

Since each term in the sum is independent and mean zero, by Assumption~\ref{asp:moment}, we have
\begin{align*}
  &\E\left[\frac{1}{T-K} \left((\mU_K \mZ_K - \E[\mU_K \mZ_K])^\top  \mQ_K \E[\mU_K \mZ_K]\right)\right]_{k,k'}^2\\
  =& \frac{1}{(T-K)^2} \sum_{t=K+1}^T \left[\sum_{t'=K+1, \lvert t'-t \rvert \le L}^T \E[\tZ_{t-k'+1} U_{t'}] \left(1 - \frac{\lvert t - t' \rvert}{L+1}\right)\right]^2 \Var[\tZ_{t-k+1} U_t]\\
  \lesssim& \frac{1}{(T-K)^2} \sum_{t=K+1}^T \left[\sum_{t'=K+1, \lvert t'-t \rvert \le L}^T \left(1 - \frac{\lvert t - t' \rvert}{L+1}\right)\right]^2\\
  \lesssim& \frac{L^2}{T-K},
\end{align*}
and then by Markov's inequality,
\begin{align*}
  & \left\lVert \frac{1}{T-K} \left(\mZ_K^\top \mU_K \mQ_K \mU_K \mZ_K - \mB_K - (\mU_K \mZ_K - \E[\mU_K \mZ_K])^\top \mQ_K (\mU_K \mZ_K - \E[\mU_K \mZ_K])\right)\right\rVert_\F \\
  =& O_\P\left(\frac{K L}{\sqrt{T-K}}\right) = o_\P(1),
\end{align*}
when $K L/T^{1/2} \to 0$.

\subsubsection{Third term.}

For entries $(k,k')$, by the definition of $\mU_K$, $\mZ_K$ and $\mQ_K$, we have
\begin{align*}
  &\left[\frac{1}{T-K} \left((\mU_K \mZ_K - \E[\mU_K \mZ_K])^\top \mQ_K (\mU_K \mZ_K - \E[\mU_K \mZ_K])\right)\right]_{k,k'} \\
  =& \frac{1}{T-K}\sum_{t,t'=K+1, \lvert t-t' \rvert \le L}^T \Big(\tZ_{t-k+1} U_t - \E[\tZ_{t-k+1} U_t] \Big) \Big(\tZ_{t'-k'+1} U_{t'} - \E[\tZ_{t'-k'+1} U_{t'} ]\Big) \Big(1 - \frac{\lvert t - t' \rvert}{L+1}\Big),
\end{align*}

Since $\tZ_t$ are independent, by Assumption~\ref{asp:moment}, we have
\begin{align*}
  &\Var\Big[\frac{1}{T-K}\sum_{t,t'=K+1, \lvert t-t' \rvert \le L}^T \Big(\tZ_{t-k+1} U_t - \E[\tZ_{t-k+1} U_t] \Big) \Big(\tZ_{t'-k'+1} U_{t'} - \E[\tZ_{t'-k'+1} U_{t'} ]\Big) \Big(1 - \frac{\lvert t - t' \rvert}{L+1}\Big)\Big]\\
  =& \frac{1}{(T-K)^2} \sum_{t,t'=K+1, \lvert t-t' \rvert \le L}^T \Var\Big[\Big(\tZ_{t-k+1} U_t - \E[\tZ_{t-k+1} U_t] \Big) \Big(\tZ_{t'-k'+1} U_{t'} - \E[\tZ_{t'-k'+1} U_{t'} ]\Big) \Big(1 - \frac{\lvert t - t' \rvert}{L+1}\Big)\Big]\\
  \lesssim & \frac{1}{(T-K)^2} (T-K) L = \frac{L}{T-K}.
\end{align*}

Then
\begin{align*}
  \E \left[\frac{1}{T-K} \left((\mU_K \mZ_K - \E[\mU_K \mZ_K])^\top \mQ_K (\mU_K \mZ_K - \E[\mU_K \mZ_K])\right)\right]_{k,k'}^2 \lesssim \frac{L}{T-K},
\end{align*}
and then by Markov's inequality,
\begin{align*}
  &\Big\lVert \frac{1}{T-K} \Big((\mU_K \mZ_K - \E[\mU_K \mZ_K])^\top \mQ_K (\mU_K \mZ_K - \E[\mU_K \mZ_K]) \\
  &- \E[(\mU_K \mZ_K - \E[\mU_K \mZ_K])^\top \mQ_K (\mU_K \mZ_K - \E[\mU_K \mZ_K])]\Big) \Big\rVert_\F = O_\P\Big(\frac{K \sqrt{L}}{\sqrt{T-K}}\Big) = o_\P(1),
\end{align*}
when $K \sqrt{L}/T^{1/2} \to 0$.

\subsubsection{Fourth term.}

For entries $(k,k')$, by the definition of $\mU_K$, $\mZ_K$ and $\mQ_K$, we have
\begin{align*}
  &\left[\frac{1}{T-K} \E\left[(\mU_K \mZ_K - \E\left[\mU_K \mZ_K\right])^\top \mQ_K (\mU_K \mZ_K - \E\left[\mU_K \mZ_K\right])\right]\right]_{k,k'}\\
  =& \frac{1}{T-K}\sum_{t,t'=K+1, \lvert t-t' \rvert \le L}^T \E\left[\left(\tZ_{t-k+1} U_t - \E[\tZ_{t-k+1} U_t] \right) \left(\tZ_{t'-k'+1} U_{t'} - \E[\tZ_{t'-k'+1} U_{t'} ]\right)\right] \left(1 - \frac{\lvert t - t' \rvert}{L+1}\right)\\
  =& \frac{1}{T-K}\sum_{t,t'=K+1, \lvert t-t' \rvert \le L}^T \Cov\left[\tZ_{t-k+1} U_t, \tZ_{t'-k'+1} U_{t'}\right] \left(1 - \frac{\lvert t - t' \rvert}{L+1}\right),
\end{align*}
and by the definition of $\mV$ in \eqref{eq:mV}, we have
\begin{align*}
  & \left[\mV\right]_{k,k'} = \frac{1}{T-K}\Cov\left[\sum_{t=K+1}^T \tZ_{t-k+1} U_t, \sum_{t'=K+1}^T \tZ_{t'-k'+1} U_{t'}\right]\\
  =& \frac{1}{T-K}\sum_{t,t'=K+1, \lvert t-t' \rvert \le L}^T \Cov\left[\tZ_{t-k+1} U_t, \tZ_{t'-k'+1} U_{t'}\right] + \frac{1}{T-K}\sum_{t,t'=K+1, \lvert t-t' \rvert > L}^T \Cov\left[\tZ_{t-k+1} U_t, \tZ_{t'-k'+1} U_{t'}\right].
\end{align*}

Then the entrywise difference is
\begin{align*}
  &\left\lvert \left[\frac{1}{T-K} \E\left[(\mU_K \mZ_K - \E[\mU_K \mZ_K])^\top \mQ_K (\mU_K \mZ_K - \E[\mU_K \mZ_K])\right] - \mV\right]_{k,k'} \right\rvert\\
  \le& \frac{1}{T-K}\sum_{t,t'=K+1, \lvert t-t' \rvert \le L}^T \frac{\lvert t - t' \rvert}{L+1} \left\lvert \Cov\left[\tZ_{t-k+1} U_t, \tZ_{t'-k'+1} U_{t'}\right] \right\rvert \\
  &+ \frac{1}{T-K}\sum_{t,t'=K+1, \lvert t-t' \rvert > L}^T \left\lvert \Cov\left[\tZ_{t-k+1} U_t, \tZ_{t'-k'+1} U_{t'}\right] \right\rvert.
\end{align*}

We now handle the dependence structure in two cases.

Assume decaying carryover effects condition (Assumption~\ref{asp:ani}). Assume $L > K$. Consider $t,t'$ such that $\lvert t-t' \rvert > K$, by Assumption~\ref{asp:moment} and the covariance inequality for the decaying carryover effects sequence in the proof of Theorem~\ref{thm:consistency}, we have
\[
\left\lvert\Cov\left[\tZ_{t-k+1} U_t, \tZ_{t'-k'+1} U_{t'}\right] \right\rvert \lesssim \theta_{T, \lvert t-t' \rvert}.
\]

Then by Assumption~\ref{asp:ani}\ref{asp:ani:ii}, we have
\begin{align*}
  &\left\lvert \left[\frac{1}{T-K} \E\left[(\mU_K \mZ_K - \E[\mU_K \mZ_K])^\top \mQ_K (\mU_K \mZ_K - \E[\mU_K \mZ_K])\right] - \mV\right]_{k,k'} \right\rvert\\
  \lesssim& \frac{K^2}{L} + \frac{1}{L} \sum_{k=1}^L k \theta_{T, k} + \sum_{k=L+1}^T \theta_{T, k} \lesssim \frac{K^2}{L} + \frac{1}{L} + \frac{1}{L^\delta},
\end{align*}
and then
\begin{align*}
  \left\lVert \frac{1}{T-K} \E\left[(\mU_K \mZ_K - \E[\mU_K \mZ_K])^\top \mQ_K (\mU_K \mZ_K - \E[\mU_K \mZ_K])\right] - \mV \right\rVert_\F \lesssim \frac{K^3 }{L} + \frac{K}{L^\delta} \to 0,
\end{align*}
when $K / L^{(1/3) \wedge \delta} \to 0$.

Assume $m$-dependence condition (Assumption~\ref{asp:mdep}). Assume $L > K$. Consider $t,t'$ such that $\lvert t-t' \rvert > K$, by Assumption~\ref{asp:moment} and the definition of $m$-dependence, we have
\[
\left\lvert\Cov\left[\tZ_{t-k+1} U_t, \tZ_{t'-k'+1} U_{t'}\right] \right\rvert \lesssim \ind(\lvert t-t' \rvert \le m).
\]

Assume further $L > m$. Then by Assumption~\ref{asp:mdep}, we have
\begin{align*}
  &\left\lvert \left[\frac{1}{T-K} \E\left[(\mU_K \mZ_K - \E[\mU_K \mZ_K])^\top \mQ_K (\mU_K \mZ_K - \E[\mU_K \mZ_K])\right] - \mV\right]_{k,k'} \right\rvert \lesssim \frac{m^2}{L},
\end{align*}
and then
\begin{align*}
  \left\lVert \frac{1}{T-K} \E\left[(\mU_K \mZ_K - \E[\mU_K \mZ_K])^\top \mQ_K (\mU_K \mZ_K - \E[\mU_K \mZ_K])\right] - \mV \right\rVert_\F \lesssim \frac{Km^2}{L} \to 0,
\end{align*}
when $Km^2 / L \to 0$.

\subsubsection{Final step.}

Combining the four terms using the decomposition of $\hat\mV$ completes the proof.

\end{proof}

\subsection{Proofs of Theorem~\ref{thm:consistency_cont} and Theorem~\ref{thm:clt_cont}}

\begin{proof}[Proofs of Theorem~\ref{thm:consistency_cont} and Theorem~\ref{thm:clt_cont}]

The proofs of Theorem~\ref{thm:consistency_cont} and Theorem~\ref{thm:clt_cont} are similar to the proofs of Theorem~\ref{thm:consistency} and Theorem~\ref{thm:clt}, we only outline the key steps here.

Recall that
\[
  \hat{\tau}_{\rm cont} = \tau + \mW_{{\rm cont},K}^{-1} [(T-K)^{-1}\mZ_{{\rm cont},K}^\top \mZ_{{\rm cont},K}]^{-1} [(T-K)^{-1}\mZ_{{\rm cont},K}^\top (\mY_K - \mZ_{{\rm cont},K} \mW_{{\rm cont},K} \tau)].
\]

We can establish Lemma~\ref{lemma:matrix}, Lemma~\ref{lemma:consistency}, Lemma~\ref{lemma:clt} in the same way, except replacing the mean $(p_1,\ldots, p_T)$ and the variance $(p_1(1-p_1),\ldots, p_T(1-p_T))$ with the mean $(\E[Z_1],\ldots, \E[Z_T])$ and the variance $(\Var[Z_1],\ldots, \Var[Z_T])$. To see why the estimand is
\begin{align*}
  \tau_k = \frac{1}{T-K} \sum_{t=K+1}^T \E\left[Y_t(Z_{t:1})\frac{(Z_{t-k} - \E[Z_{t-k}])}{\Var[Z_{t-k}]}\right],
\end{align*}
we show $\E[(T-K)^{-1}\mZ_K^\top (\mY_K - \mZ_K \mW_K \tau)] = 0$ by following the proof of Lemma~\ref{lemma:consistency}. For any $k=0,\ldots, K$, we have
\begin{align*}
  &\E\left[\frac{1}{T-K} \sum_{t=K+1}^T \tZ_{t-k}\left(Y_t - \tZ_{t:t-K}^\top \mW_K \tau\right)\right] = \E\left[\frac{1}{T-K} \sum_{t=K+1}^T \tZ_{t-k}\left(Y_t - \sum_{\ell=0}^K \tZ_{t-\ell}w_\ell \tau_\ell\right)\right]\\
  =& \frac{1}{T-K} \sum_{t=K+1}^T \left(\E[\tZ_{t-k} Y_t] - \Var[\tZ_{t-k}] w_k \tau_k\right) \\
  =& \left(\frac{1}{T-K} \sum_{t=K+1}^T \E[\tZ_{t-k} Y_t]\right) - \left(\frac{1}{T-K} \sum_{t=K+1}^T \Var[\tZ_{t-k}]\right) w_k \tau_k \\
  =& \left(\frac{1}{T-K} \sum_{t=K+1}^T \E\left[Y_t(Z_{t:1}) \frac{(Z_{t-k} - \E[Z_{t-k}])}{\Var[Z_{t-k}]}\right]\right) - \left(\frac{1}{T-K} \sum_{t=K+1}^T \Var[Z_{t-k}]^{-1}\right) w_k \tau_k \\
  =& \left(\frac{1}{T-K} \sum_{t=K+1}^T \E\left[Y_t(Z_{t:1}) \frac{(Z_{t-k} - \E[Z_{t-k}])}{\Var[Z_{t-k}]}\right]\right) - \tau_k = 0.
\end{align*}

We then complete the proofs.

\end{proof}

\subsection{Proof of Theorem~\ref{thm:var_cont}}

\begin{proof}[Proof of Theorem~\ref{thm:var_cont}]
  The proof is similar to the proof of Theorem~\ref{thm:var} by considering the decomposition of $\hat\mV$ in \eqref{eq:appendix_2}.
\end{proof}

\subsection{Proofs of Theorem~\ref{thm:consistency_int} and Theorem~\ref{thm:clt_int}}

\begin{proof}[Proofs of Theorem~\ref{thm:consistency_int} and Theorem~\ref{thm:clt_int}]

The proofs of Theorem~\ref{thm:consistency_int} and Theorem~\ref{thm:clt_int} are similar to the proofs of Theorem~\ref{thm:consistency} and Theorem~\ref{thm:clt}, we only outline the key steps here.

Recall that
\[
  \hat{\tau}_{\rm int} = \tau + \mW_{{\rm int},K}^{-1} [(T-K)^{-1}\mZ_{{\rm int},K}^\top \mZ_{{\rm int},K}]^{-1} [(T-K)^{-1}\mZ_{{\rm int},K}^\top (\mY_K - \mZ_{{\rm int},K} \mW_{{\rm int},K} \tau)].
\]

Similar to Lemma~\ref{lemma:matrix}, we can show that under the same assumptions of Lemma~\ref{lemma:matrix}, we have
\begin{align*}
  \left\lVert \frac{1}{T-K} \mZ_{{\rm int},K}^\top \mZ_{{\rm int},K} - \mW_{{\rm int},K}^{-1} \right\rVert_\F \stackrel{\sf p}{\longrightarrow} 0.
\end{align*}

Similar to Lemma~\ref{lemma:consistency}, we can show that under the same assumptions of Lemma~\ref{lemma:consistency}, we have
\begin{align*}
  \left\lVert \frac{1}{T-K} \mZ_{{\rm int},K}^\top (\mY_K - \mZ_{{\rm int},K} \mW_{{\rm int},K} \tau) \right\rVert_2 \stackrel{\sf p}{\longrightarrow} 0.
\end{align*}

As we have different asymptotic variance: $\mV_{\rm int}$ instead of $\mV$, similar to Lemma~\ref{lemma:clt}, we can show that under the same assumptions of Lemma~\ref{lemma:clt} except replacing Assumption~\ref{asp:eigen} with Assumption~\ref{asp:eigen_int}, we have
\begin{align*}
  \sqrt{T-K} (\lambda_K^\top \mV_{\rm int} \lambda_K)^{-1/2} \lambda_K^\top \left[\frac{1}{T-K} \mZ_{{\rm int},K}^\top (\mY_K - \mZ_{{\rm int},K} \mW_{{\rm int},K} \tau)\right] \stackrel{\sf d}{\longrightarrow} N(0,1).
\end{align*}

We then complete the proofs.

\end{proof}

\subsection{Proofs of Theorem~\ref{thm:consistency_general} and Theorem~\ref{thm:clt_general}}

\begin{proof}[Proofs of Theorem~\ref{thm:consistency_general} and Theorem~\ref{thm:clt_general}]

The proofs of Theorem~\ref{thm:consistency_general} and Theorem~\ref{thm:clt_general} are similar to the proofs of Theorem~\ref{thm:consistency} and Theorem~\ref{thm:clt}, we only outline the key steps here.

Recall that
\[
  \hat{\tau}_{\rm g} = \tau + \mW_{{\rm g},K}^{-1} [(T-K)^{-1}\mZ_{{\rm g},K}^\top \mZ_{{\rm g},K}]^{-1} [(T-K)^{-1}\mZ_{{\rm g},K}^\top (\mY_K - \mZ_{{\rm g},K} \mW_{{\rm g},K} \tau)].
\]

Similar to Lemma~\ref{lemma:matrix}, we can show that under the same assumptions of Lemma~\ref{lemma:matrix} except replacing Assumption~\ref{asp:ps} with Assumption~\ref{asp:ps_general}, we have
\begin{align*}
  \left\lVert \frac{1}{T-K} \mZ_{{\rm g},K}^\top \mZ_{{\rm g},K} - \mW_{{\rm g},K}^{-1} \right\rVert_\F \stackrel{\sf p}{\longrightarrow} 0.
\end{align*}

Similar to Lemma~\ref{lemma:consistency}, we can show that under the same assumptions of Lemma~\ref{lemma:consistency} except replacing Assumption~\ref{asp:ps}, \ref{asp:moment} with Assumption~\ref{asp:ps_general}, \ref{asp:moment_general}, we have
\begin{align*}
  \left\lVert \frac{1}{T-K} \mZ_{{\rm g},K}^\top (\mY_K - \mZ_{{\rm g},K} \mW_{{\rm g},K} \tau) \right\rVert_2 \stackrel{\sf p}{\longrightarrow} 0.
\end{align*}

As we have different asymptotic variance: $\mV_{\rm g}$ instead of $\mV$, similar to Lemma~\ref{lemma:clt}, we can show that under the same assumptions of Lemma~\ref{lemma:clt} except replacing Assumption~\ref{asp:ps}, \ref{asp:moment}, \ref{asp:eigen} with Assumption~\ref{asp:ps_general}, \ref{asp:moment_general}, \ref{asp:eigen_general}, we have
\begin{align*}
  \sqrt{T-K} (\lambda_K^\top \mV_{\rm g} \lambda_K)^{-1/2} \lambda_K^\top \left[\frac{1}{T-K} \mZ_{{\rm g},K}^\top (\mY_K - \mZ_{{\rm g},K} \mW_{{\rm g},K} \tau)\right] \stackrel{\sf d}{\longrightarrow} N(0,1).
\end{align*}

We then complete the proofs.

\end{proof}

\subsection{Proof of Proposition~\ref{prop:var}}

\begin{proof}[Proof of Proposition~\ref{prop:var}]
By the model assumption, we can equivalently write the model as
\begin{align*}
  Y_t(z_{t:1})=\sum_{k=0}^{t-1} \beta_k z_{t-k} + \epsilon_t = p(1-p)\sum_{k=0}^{t-1} \beta_k \frac{z_{t-k} - p}{p(1-p)} + p\sum_{k=0}^{t-1} \beta_k + \epsilon_t.
\end{align*}

By the model assumption, we have $w_k = p(1-p)$ and $\tau_k = \beta_k$. Then the asymptotic variance of full OLS is
\begin{align*}
  &\frac{1}{T-K} \Var\left[\sum_{t=K+1}^T \tZ_{t-k} \left(Y_t - \sum_{\ell=0}^K \tZ_{t-\ell} w_\ell \tau_\ell\right)\right] \\
  =& \frac{1}{T-K} \Var\left[\sum_{t=K+1}^T \tZ_{t-k} \left(p(1-p)\sum_{\ell=K+1}^{t-1} \beta_\ell \tZ_{t-\ell} + p\sum_{k=0}^{t-1} \beta_k + \epsilon_t \right)\right]\\
  =& \frac{1}{T-K} \sum_{t=K+1}^T \Var\left[\tZ_{t-k} \left(p(1-p)\sum_{\ell=K+1}^{t-1} \beta_\ell \tZ_{t-\ell} + p\sum_{k=0}^{t-1} \beta_k + \epsilon_t \right)\right]\\
  =& \frac{1}{T-K} \sum_{t=K+1}^T \Var\left[p(1-p) \tZ_{t-k} \sum_{\ell=K+1}^{t-1} \beta_\ell \tZ_{t-\ell} \right] + \frac{1}{T-K} \sum_{t=K+1}^T \Var\left[\left(p\sum_{k=0}^{t-1} \beta_k + \epsilon_t \right)\tZ_{t-k} \right]\\
  =& \frac{1}{T-K} \sum_{t=K+1}^T \sum_{\ell=K+1}^{t-1} \Var\left[p(1-p) \tZ_{t-k} \beta_\ell \tZ_{t-\ell} \right] + \frac{1}{T-K} \sum_{t=K+1}^T \Var\left[\left(p\sum_{k=0}^{t-1} \beta_k + \epsilon_t \right)\tZ_{t-k} \right]\\
  =& \frac{1}{T-K} \sum_{t=K+1}^T \sum_{\ell=K+1}^{t-1} \beta_\ell^2 + \frac{1}{T-K} \frac{1}{p(1-p)}\sum_{t=K+1}^T \left(p\sum_{k=0}^{t-1} \beta_k + \epsilon_t \right)^2.
\end{align*}

The asymptotic variance of marginal OLS is
\begin{align*}
  &\frac{1}{T-K} \Var\left[\sum_{t=K+1}^T \tZ_{t-k} (Y_t - \tZ_{t-k} w_k \tau_k)\right]\\
  =& \frac{1}{T-K} \Var\left[\sum_{t=K+1}^T \tZ_{t-k} \left(p(1-p)\sum_{\ell=0,\ell \neq k}^{t-1} \beta_\ell \tZ_{t-\ell} + p\sum_{k=0}^{t-1} \beta_k + \epsilon_t \right)\right]\\
  =& \frac{1}{T-K} \sum_{t=K+1}^T \Var\left[\tZ_{t-k} \left(p(1-p)\sum_{\ell=0,\ell \neq k}^{t-1} \beta_\ell \tZ_{t-\ell} + p\sum_{k=0}^{t-1} \beta_k + \epsilon_t \right)\right]\\
  =& \frac{1}{T-K} \sum_{t=K+1}^T \Var\left[p(1-p) \tZ_{t-k} \sum_{\ell=0,\ell \neq k}^{t-1} \beta_\ell \tZ_{t-\ell} \right] + \frac{1}{T-K} \sum_{t=K+1}^T \Var\left[\left(p\sum_{k=0}^{t-1} \beta_k + \epsilon_t \right)\tZ_{t-k} \right]\\
  =& \frac{1}{T-K} \sum_{t=K+1}^T \sum_{\ell=0,\ell \neq k}^{t-1} \Var\left[p(1-p) \tZ_{t-k} \beta_\ell \tZ_{t-\ell} \right] + \frac{1}{T-K} \sum_{t=K+1}^T \Var\left[\left(p\sum_{k=0}^{t-1} \beta_k + \epsilon_t \right)\tZ_{t-k} \right]\\
  =& \frac{1}{T-K} \sum_{t=K+1}^T \sum_{\ell=0,\ell \neq k}^{t-1} \beta_\ell^2 + \frac{1}{T-K} \frac{1}{p(1-p)}\sum_{t=K+1}^T \left(p\sum_{k=0}^{t-1} \beta_k + \epsilon_t \right)^2.
\end{align*}

\end{proof}

\subsection{Proof of Proposition~\ref{prop:int}}

\begin{proof}[Proof of Proposition~\ref{prop:int}]
  By the model assumption, we can equivalently write the model as
  \begin{align*}
    &Y_t(z_{t:1})=\sum_{k=0}^{t-1} \beta_k z_{t-k} + \sum_{k=1}^{t-1} \beta_{k-1,k} z_{t-k} z_{t-k+1} + \epsilon_t \\
    =& p(1-p)\sum_{k=0}^{t-1} \beta_k \frac{z_{t-k} - p}{p(1-p)} + p^2(1-p)^2\sum_{k=1}^{t-1} \beta_{k-1,k} \frac{z_{t-k} - p}{p(1-p)}\frac{z_{t-k+1} - p}{p(1-p)} \\
    &+ p\sum_{k=0}^{t-1} 
    \beta_k + p \sum_{k=1}^{t-1} \beta_{k-1,k}(z_{t-k} + z_{t-k+1} - p) + \epsilon_t\\
    =& p(1-p)\sum_{k=0}^{t-1} \beta_k \frac{z_{t-k} - p}{p(1-p)} + p^2(1-p)^2\sum_{k=1}^{t-1} \beta_{k-1,k} \frac{z_{t-k} - p}{p(1-p)}\frac{z_{t-k+1} - p}{p(1-p)} \\
    &+ p^2(1-p) \sum_{k=1}^{t-1} \beta_{k-1,k}\frac{z_{t-k} - p}{p(1-p)} + p^2(1-p) \sum_{k=0}^{t-2} \beta_{k,k+1}\frac{z_{t-k} - p}{p(1-p)}+ p\sum_{k=0}^{t-1} 
    \beta_k  + p^2 \sum_{k=1}^{t-1} \beta_{k-1,k} + \epsilon_t.
  \end{align*}
  
  By the definition of main effect weights and interaction effect weights, and the model assumption, we have $w_k = p(1-p)$ and $w_{k-1,k} = p^2(1-p)^2$. The main effect is 
  \begin{align*}
    \tau_k =& \frac{1}{T-K} \sum_{t=K+1}^T \E\left[Y_t(Z_{t:t-k+1},1,Z_{t-k-1:1}) - Y_t(Z_{t:t-k+1},0,Z_{t-k-1:1})\right] \\
    =& \frac{1}{T-K} \sum_{t=K+1}^T \E\left[\beta_k + \beta_{k-1,k} z_{t-k+1} \ind(k \neq 0) + \beta_{k,k+1} z_{t-k-1} \ind(k \neq T-1)\right]\\
    =& \beta_k + p \beta_{k-1,k} \ind(k \neq 0) + p \beta_{k,k+1} \ind(k \neq T-1),
  \end{align*}
  and the interaction effect is
  \begin{align*}
    \tau_{k-1,k} = \beta_{k-1,k}.
  \end{align*}

  Then the residual of OLS with interaction is
  \begin{align*}
    &Y_t - \sum_{\ell=0}^K \tZ_{t-\ell} w_\ell \tau_\ell - \sum_{\ell=1}^K \tZ_{t-\ell} \tZ_{t-\ell+1} w_{\ell-1,\ell} \tau_{\ell-1,\ell}\\
    =& p(1-p)\sum_{\ell=K+1}^{t-1} \beta_\ell \tZ_{t-\ell} + p^2(1-p)^2\sum_{\ell=K+1}^{t-1} \beta_{\ell-1,\ell} \tZ_{t-\ell} \tZ_{t-\ell+1}\\
    &+ p^2(1-p) \sum_{\ell=K+1}^{t-1} \beta_{\ell-1,\ell} \tZ_{t-\ell} + p^2(1-p) \sum_{\ell=K+1}^{t-2} \beta_{\ell,\ell+1} \tZ_{t-\ell} + p\sum_{k=0}^{t-1} 
    \beta_k  + p^2 \sum_{k=1}^{t-1} \beta_{k-1,k} + \epsilon_t.
  \end{align*}
  
  Then the asymptotic variance of OLS with interaction is
  \begin{align*}
    &\frac{1}{T-K} \Var\left[\sum_{t=K+1}^T \tZ_{t-k} (Y_t - \sum_{\ell=0}^K \tZ_{t-\ell} w_\ell \tau_\ell - \sum_{\ell=1}^K \tZ_{t-\ell} \tZ_{t-\ell+1} w_{\ell-1,\ell} \tau_{\ell-1,\ell})\right]  \\
    =& \frac{1}{T-K} \sum_{t=K+1}^T \Var\left[\tZ_{t-k} \left(Y_t - \sum_{\ell=0}^K \tZ_{t-\ell} w_\ell \tau_\ell - \sum_{\ell=1}^K \tZ_{t-\ell} \tZ_{t-\ell+1} w_{\ell-1,\ell} \tau_{\ell-1,\ell}\right)\right] \\
    =& \frac{1}{T-K} \sum_{t=K+1}^T \Var\left[\tZ_{t-k} \left(p(1-p)  \sum_{\ell=K+1}^{t-1} \beta_\ell \tZ_{t-\ell} + p^2(1-p) \sum_{\ell=K+1}^{t-1} \beta_{\ell-1,\ell} \tZ_{t-\ell} + p^2(1-p) \sum_{\ell=K+1}^{t-2} \beta_{\ell,\ell+1} \tZ_{t-\ell}\right) \right] \\
    &+ \frac{1}{T-K} \sum_{t=K+1}^T \Var\left[\tZ_{t-k} p^2(1-p)^2\sum_{\ell=K+1}^{t-1} \beta_{\ell-1,\ell} \tZ_{t-\ell} \tZ_{t-\ell+1}\right] \\
    &+ \frac{1}{T-K} \sum_{t=K+1}^T \Var\left[\left(p\sum_{k=0}^{t-1} 
    \beta_k  + p^2 \sum_{k=1}^{t-1} \beta_{k-1,k} + \epsilon_t\right)\tZ_{t-k} \right]\\
    =& \frac{1}{T-K} \sum_{t=K+1}^T \sum_{\ell=K+1}^{t-1} \Var\left[p(1-p) \tZ_{t-k} \beta_\ell \tZ_{t-\ell} \right] + \frac{1}{T-K} \sum_{t=K+1}^T \sum_{\ell=K+1}^{t-1} \Var\left[p^2(1-p) \tZ_{t-k} \beta_{\ell-1,\ell} \tZ_{t-\ell}\right]\\
    &+ \frac{1}{T-K} \sum_{t=K+1}^T \sum_{\ell=K+1}^{t-2} \Var\left[p^2(1-p) \tZ_{t-k} \beta_{\ell,\ell+1} \tZ_{t-\ell}\right] \\
    &+ \frac{1}{T-K} \sum_{t=K+1}^T \sum_{\ell=K+1}^{t-1} \Var\left[p^2(1-p)^2 \tZ_{t-k} \beta_{\ell-1,\ell} \tZ_{t-\ell} \tZ_{t-\ell+1}\right]\\
    &+\frac{1}{T-K} \sum_{t=K+1}^T \Var\left[\left(p\sum_{k=0}^{t-1} 
    \beta_k  + p^2 \sum_{k=1}^{t-1} \beta_{k-1,k} + \epsilon_t \right)\tZ_{t-k} \right]\\
    =& \frac{1}{T-K} \sum_{t=K+1}^T \sum_{\ell=K+1}^{t-1} \left(\beta_\ell^2+ p^2 \beta_{\ell-1,\ell}^2 + p^2 \beta_{\ell,\ell+1}^2\right) + \frac{1}{T-K} \sum_{t=K+1}^T \sum_{\ell=K+1}^{t-1} p(1-p)\beta_{\ell-1,\ell}^2\\
    &+ \frac{1}{T-K} \frac{1}{p(1-p)}\sum_{t=K+1}^T \left(p\sum_{k=0}^{t-1} \beta_k + p^2 \sum_{k=1}^{t-1} \beta_{k-1,k} + \epsilon_t \right)^2.
  \end{align*}

  On the other hand, the residual of OLS without interaction is
  \begin{align*}
    &Y_t - \sum_{\ell=0}^K \tZ_{t-\ell} w_\ell \tau_\ell\\
    =& p(1-p)\sum_{\ell=K+1}^{t-1} \beta_\ell \tZ_{t-\ell} + p^2(1-p)^2\sum_{\ell=1}^{t-1} \beta_{\ell-1,\ell} \tZ_{t-\ell} \tZ_{t-\ell+1}\\
    &+ p^2(1-p) \sum_{\ell=K+1}^{t-1} \beta_{\ell-1,\ell} \tZ_{t-\ell} + p^2(1-p) \sum_{\ell=K+1}^{t-2} \beta_{\ell,\ell+1} \tZ_{t-\ell} + p\sum_{k=0}^{t-1} 
    \beta_k  + p^2 \sum_{k=1}^{t-1} \beta_{k-1,k} + \epsilon_t.
  \end{align*}
  
  Then the asymptotic variance of OLS without interaction is
  \begin{align*}
    &\frac{1}{T-K} \Var\left[\sum_{t=K+1}^T \tZ_{t-k} \left(Y_t - \sum_{\ell=0}^K \tZ_{t-\ell} w_\ell \tau_\ell\right)\right]  \\
    =& \frac{1}{T-K} \sum_{t=K+1}^T \Var\left[\tZ_{t-k} \left(Y_t - \sum_{\ell=0}^K \tZ_{t-\ell} w_\ell \tau_\ell\right)\right] \\
    =& \frac{1}{T-K} \sum_{t=K+1}^T \Var\left[\tZ_{t-k} \left(p(1-p)  \sum_{\ell=K+1}^{t-1} \beta_\ell \tZ_{t-\ell} + p^2(1-p) \sum_{\ell=K+1}^{t-1} \beta_{\ell-1,\ell} \tZ_{t-\ell} + p^2(1-p) \sum_{\ell=K+1}^{t-2} \beta_{\ell,\ell+1} \tZ_{t-\ell}\right) \right] \\
    &+ \frac{1}{T-K} \sum_{t=K+1}^T \Var\left[\tZ_{t-k} p^2(1-p)^2\sum_{\ell=1}^{t-1} \beta_{\ell-1,\ell} \tZ_{t-\ell} \tZ_{t-\ell+1}\right] \\
    &+ \frac{1}{T-K} \sum_{t=K+1}^T \Var\left[\left(p\sum_{k=0}^{t-1} 
    \beta_k  + p^2 \sum_{k=1}^{t-1} \beta_{k-1,k} + \epsilon_t\right)\tZ_{t-k} \right]\\
    =& \frac{1}{T-K} \sum_{t=K+1}^T \sum_{\ell=K+1}^{t-1} \Var\left[p(1-p) \tZ_{t-k} \beta_\ell \tZ_{t-\ell} \right] + \frac{1}{T-K} \sum_{t=K+1}^T \sum_{\ell=K+1}^{t-1} \Var\left[p^2(1-p) \tZ_{t-k} \beta_{\ell-1,\ell} \tZ_{t-\ell}\right]\\
    &+ \frac{1}{T-K} \sum_{t=K+1}^T \sum_{\ell=K+1}^{t-2} \Var\left[p^2(1-p) \tZ_{t-k} \beta_{\ell,\ell+1} \tZ_{t-\ell}\right] \\
    &+ \frac{1}{T-K} \sum_{t=K+1}^T \sum_{\ell=1}^{t-1} \Var\left[p^2(1-p)^2 \tZ_{t-k} \beta_{\ell-1,\ell} \tZ_{t-\ell} \tZ_{t-\ell+1}\right]\\
    &+\frac{1}{T-K} \sum_{t=K+1}^T \Var\left[\left(p\sum_{k=0}^{t-1} 
    \beta_k  + p^2 \sum_{k=1}^{t-1} \beta_{k-1,k} + \epsilon_t \right)\tZ_{t-k} \right]\\
    =& \frac{1}{T-K} \sum_{t=K+1}^T \sum_{\ell=K+1}^{t-1} (\beta_\ell^2+ p^2 \beta_{\ell-1,\ell}^2 + p^2 \beta_{\ell,\ell+1}^2) + \frac{1}{T-K} \sum_{t=K+1}^T \sum_{\ell=1}^{t-1} p(1-p)\beta_{\ell-1,\ell}^2\\
    &+ \frac{1}{T-K} \frac{1}{p(1-p)}\sum_{t=K+1}^T \Big(p\sum_{k=0}^{t-1} \beta_k + p^2 \sum_{k=1}^{t-1} \beta_{k-1,k} + \epsilon_t \Big)^2.
  \end{align*}
  
  \end{proof}
\end{document}